\documentclass{sig-alternate-05-2015}

\usepackage{xspace}
\usepackage{graphicx}
\usepackage{amsmath}
\usepackage{amssymb}

\usepackage{mathtools}
\usepackage{MnSymbol-vecfixed} 
\usepackage{comment}

\usepackage{caption}
\usepackage{subcaption}

\usepackage{units}      
\usepackage{booktabs,tabularx}
\usepackage{tabu,paralist}

\usepackage{thmtools, thm-restate}

\newcommand{\eat}[1]{}

\newcommand{\responseref}[2]{\hyperref[#1]{\color{red}{\emph{#2}}}}

\newcommand{\qqquad}{\qquad\qquad}
\newcommand{\qqqquad}{\qqquad\qqquad}

\newcommand{\String}{\text{String}}

\newcommand{\arr}{\rightarrow}

\newcommand{\darr}{\mapsto}
\newcommand{\x}{\times}
\newcommand{\co}{\circ}

\newcommand{\ffrac}[2]{\genfrac{}{}{0pt}{0}{#1}{#2}}

\newcommand{\keyw}[1]{\boldsymbol{\kw{#1}}}

\newcommand{\meane}[2]{\lbrack\!\lbrack #1 \rbrack\!\rbrack_{#2}}

\newcommand{\meanl}[2]{\lbrack\!\lbrack #1 \rbrack\!\rbrack^{L}_{#2}}

\newcommand{\lbag}{\beta}
\newcommand{\expf}{\kw{exp}}

\newcommand{\ACz}{\ensuremath{\kw{AC_0}}}
\newcommand{\TCz}{\ensuremath{\kw{TC_0}}}
\newcommand{\NCz}{\ensuremath{\kw{NC_0}}}
\newcommand{\logspace}{\ensuremath{\kw{LOGSPACE}}}

\newcommand{\shred}{\kw{sh}}

\newcommand{\Shred}{\kw{Sh}}
\newcommand{\Value}{\kw{Val}}
\newcommand{\Pair}{\kw{Pair}}
\newcommand{\Seq}{\kw{Seq}}
\newcommand{\Elem}{\kw{Elem}}

\newcommand{\Pred}{\kw{EndsWith}}
\newcommand{\Succ}{\kw{StartsWith}}
\newcommand{\Dict}{\kw{Dict}}
\newcommand{\ToSeq}{\kw{ToSeq}}
\newcommand{\implyA}{\rightarrow}
\newcommand{\FStr}{N^{Str}}
\newcommand{\FBin}{FBin}
\newcommand{\FSet}{F^{Set}}
\newcommand{\FBag}{F^{Bag}}
\newcommand{\Eq}{\kw{Eq}}

\newcommand{\monco}{\odot}

\newcommand{\movs}{M}
\newcommand{\Movie}{Movie}
\newcommand{\shows}{\mathit{Sh}}
\newcommand{\exQ}{DOz}

\newcommand{\cprod}{\x}

\renewcommand{\deg}[1]{#1^\co}
\newcommand{\pbN}{\bN^{+}}
\newcommand{\Base}{Base}
\newcommand{\TBase}{TBase}
\newcommand{\uZ}{\emptyset}

\newcommand{\etpl}{\sng(\tuple{})}
\newcommand{\card}[1]{| #1 |}
\newcommand{\tcost}{\kw{tcost}}
\newcommand{\cost}[3]{\cC_{#2}\lbrack\!\lbrack #1 \rbrack\!\rbrack(#3)}
\newcommand{\costc}[3]{\cC_{#2}\lbrack\!\lbrack #1 \rbrack\!\rbrack_{#3}}
\newcommand{\costcx}[3]{\cC_{#2}\lbrack\!\lbrack #1
\rbrack\!\rbrack_{\Irhoc;\Vrhoc[x\coloneq #3]}}
\newcommand{\costcy}[3]{\cC_{#2}\lbrack\!\lbrack #1
\rbrack\!\rbrack_{\Irhoc;\Vrhoc[y\coloneq #3]}}
\newcommand{\costt}[2]{\cC_{#2}\lbrack\!\lbrack #1 \rbrack\!\rbrack}
\newcommand{\sz}{\kw{size}}
\newcommand{\sOne}{1}
\newcommand{\dg}{\kw{deg}}
\newcommand{\sngr}{\keyw{sng}^*}

\newcommand{\piT}[1]{\sng(\pi_{#1})}
\newcommand{\piTC}[2]{\sng(\pi_{#2}(#1))}

\newcommand{\eptybag}{\set{}}
\newcommand{\dcup}{\cup}
\newcommand{\gdic}{\cD}
\newcommand{\sngv}[1]{\set{#1}}
\newcommand{\sngvar}[1]{\sng(#1)}

\newcommand{\fc}[1]{#1^F}
\newcommand{\gc}[1]{#1^\gG}
\newcommand{\ggc}[2]{#1^{\gG_{#2}}}
\newcommand{\fcf}[1]{\shred^F(#1)}
\newcommand{\gcf}[1]{\shred^\gG(#1)}

\newcommand{\shF}{\fc{\sh}}
\newcommand{\shG}{\gc{\sh}}
\newcommand{\shGG}[1]{\ggc{\sh}{#1}}

\newcommand{\Any}{\kw{Any}}

\newcommand{\pp}{\keyw{p}}
\newcommand{\pcost}{\kw{cost}}
\newcommand{\kcons}{\kw{consistent}}

\newcommand{\eff}[1]{}

\newcommand{\NRC}{\kw{NRC}}
\newcommand{\pNRC}{\kw{NRC}^+}
\newcommand{\ipNRC}{\kw{IncNRC}^+}

\newcommand{\lpNRC}{\pNRC_l}
\newcommand{\ilpNRC}{\kw{IncNRC}^+_l}

\newcommand{\apNRC}{\kw{NRC}^+_{\Sigma}}
\newcommand{\aipNRC}{\kw{IncNRC}^+_{\Sigma}}

\newcommand{\alpNRC}{\kw{NRC}^+_{l,\Sigma}}
\newcommand{\ailpNRC}{\kw{IncNRC}^+_{l,\Sigma}}

\newcommand{\apRA}{\kw{RA}^+_{\Sigma}}

\newcommand{\ipNRA}{\kw{IncNRA}^+}

\newcommand{\IGamma}{\Gamma}
\newcommand{\VGamma}{\Pi}

\newcommand{\Irho}{\gamma}
\newcommand{\Vrho}{\varepsilon}
\newcommand{\Drho}{\phi}

\newcommand{\Irhoc}{\deg{\Irho}}
\newcommand{\Vrhoc}{\deg{\Vrho}}

\newcommand{\rename}{\rho}

\newcommand{\gG}{\Gamma}

\newcommand{\inL}{\kw{inL}}
\newcommand{\outL}{\kw{outL}}

\newcommand{\nst}{\kw{u}}
\newcommand{\sh}{\kw{s}}

\newcommand{\Num}{\boldsymbol{Num}}
\newcommand{\nzero}{\underline{0}}

\newcommand{\nplus}{+}
\newcommand{\nminus}{-}
\newcommand{\nprod}{\cdot}
\newcommand{\nflt}{\keyw{sum}}

\newcommand{\app}{\kw{app}}

\newcommand{\deltac}{\inc}

\newcommand{\filter}{\kw{filter}}

\newcommand{\forz}{\keyw{for}}
\newcommand{\for}[2]{\keyw{for}\ #1\ \keyw{in}\ #2}
\newcommand{\fortight}[2]{\keyw{for}\;#1\;\keyw{in}\;#2}
\newcommand{\forr}[3]{\keyw{for}\ #1\ \keyw{in}\ #2\collects #3}
\newcommand{\forx}[2]{\keyw{for}\ x\ \keyw{in}\ #1\collects #2}
\newcommand{\fory}[2]{\keyw{for}\ y\ \keyw{in}\ #1\collects #2}

\newcommand{\lletz}{\keyw{let}}
\newcommand{\llet}[2]{\keyw{let}\;#1 \coloneq #2\;\keyw{in}}
\newcommand{\lletBin}[4]{\keyw{let}\;#1 \coloneq #2,\,#3 \coloneq #4\;\keyw{in}}

\newcommand{\where}{\keyw{where}}
\newcommand{\nest}{\texttt{nest}}
\newcommand{\collect}{\keyw{union}}
\newcommand{\collects}{\;\keyw{union}\;}
\newcommand{\sng}{\keyw{sng}}

\newcommand{\bagS}[2]{{#2}\set{#1}}

\newcommand{\related}{\texttt{related}}
\newcommand{\relB}{\texttt{relB}}
\newcommand{\isRelated}{\texttt{isRelated}}

\newcommand{\oR}{\overline{R}}
\newcommand{\oDR}{\overline{\Delta R}}

\newcommand{\kw}[1]{\operatorname{#1}}

\newtheorem{theorem}{Theorem} 
\newtheorem{fact}{Fact} [section]
\newtheorem{corollary}[fact]{Corollary}

\newtheorem{lemma}[theorem]{Lemma}

\newtheorem{definition}{Definition}
\newtheorem{example}{Example}

\newenvironment{changemargin}[2]{%
\begin{list}{}{%
\setlength{\leftmargin}{#1}%
\setlength{\rightmargin}{#2}%
}%
\item[]}{\end{list}}

\newcommand{\sch}{\kw{Sch}}

\newcommand{\la}{\langle}
\newcommand{\ra}{\rangle}

\newcommand{\tuple}[1]{\langle{#1}\rangle}
\newcommand{\set}[1]{\{{#1}\}}

\newcommand{\supp}{\kw{supp}}

\newcommand{\flt}{\keyw{flatten}}

\newcommand{\true}{\kw{true}}
\newcommand{\false}{\kw{false}}

\newcommand{\size}{\kw{size}}

\newcommand{\cC}{\ensuremath{\mathcal{C}}}
\newcommand{\cD}{\ensuremath{\mathcal{D}}}

\newcommand{\bD}{\ensuremath{\mathbb{D}}}

\newcommand{\bL}{\ensuremath{\mathbb{L}}}

\newcommand{\bN}{\ensuremath{\mathbb{N}}}

\newcommand{\mean}[3]{\lbrack\!\lbrack #1 \rbrack\!\rbrack_{#2}(#3)}

\newcommand{\dblbrk}[1]{\lbrack\!\lbrack #1\rbrack\!\rbrack}

\newcommand{\inc}{\delta}
\newcommand{\incR}[1]{\delta_{#1}}

\newcommand{\Bag}{\keyw{Bag}}

\newcommand{\scream}[1]{\textit{\large ~***~#1~***~}}
\newcommand{\val}[1]{\scream{val:~#1}}
\newcommand{\christoph}[1]{\scream{ck:~#1}}

\usepackage{listings}
\usepackage{color}

\newcommand{\codesize}{\normalsize}

\lstdefinelanguage{Scala}%
{morekeywords={abstract,%
  case,catch,char,class,%
  def,else,extends,final,finally,for,%
  if,import,implicit,%
  match,module,%
  new,null,%
  object,override,%
  package,private,protected,public,%
  for,public,return,super,%
  this,throw,trait,try,type,%
  val,var,%
  with,while,%
  yield%
  },%
  sensitive,%
  morecomment=[l]//,%
  morecomment=[s]{/*}{*/},%
  morestring=[b]",%
  morestring=[b]',%
  showstringspaces=false%
}[keywords,comments,strings]%

\definecolor{listingbg}{RGB}{240, 240, 240}

\allowdisplaybreaks[4]

\date{}

\begin{document}


\title{Incremental View Maintenance For Collection Programming\thanks{
This work was supported by ERC grant 279804 and 
NSF grants IIS-1217798 and IIS-1302212.}}

\numberofauthors{3}

\author{
\alignauthor Christoph Koch\\
   \affaddr{EPFL}\\
   \email{christoph.koch@epfl.ch}
\alignauthor Daniel Lupei\titlenote{Corresponding author.}\\
   \affaddr{EPFL}\\
   \email{daniel.lupei@epfl.ch}
\alignauthor Val Tannen\\
   \affaddr{University of Pennsylvania}\\
   \email{val@cis.upenn.edu}
}

\maketitle


\begin{abstract}

In the context of incremental view maintenance (IVM), {\em delta} query
derivation is an essential technique for speeding up the processing of large,
dynamic datasets.
The goal is to generate delta queries that, given a small
change in the input, can update the materialized view more efficiently than via
recomputation.  

In this work we propose the first solution for the efficient incrementalization
of positive nested relational calculus ($\pNRC$) on bags
(with integer multiplicities).
More precisely, we model the cost of $\pNRC$ operators and classify
queries as efficiently incrementalizable if their delta has
a strictly lower cost than full re-evaluation.
Then, we identify $\ipNRC,$ a large fragment of $\pNRC$ that is efficiently
incrementalizable and we provide  
a semantics-preserving translation that takes any $\pNRC$ query to a
collection of $\ipNRC$ queries.
Furthermore, we prove that incremental maintenance 
for $\pNRC$ is within the complexity class $\NCz$
and
we showcase how {\em recursive} IVM, a technique that has provided
significant speedups over traditional IVM in the case of flat
queries~\cite{ahmad14}, can also be applied to $\ipNRC.$

\end{abstract}


\section{Introduction}

Large-scale
collection processing 
in frameworks such as Spark~\cite{Zaharia10} or LINQ~\cite{LINQ}
can greatly benefit from incremental maintenance in order to
minimize query latency in the face of updates. 
These frameworks provide 
collection abstractions equivalent to nested relational operators
that are embarrassingly parallelizable.
Also, they can be aggressively optimized using
powerful algebraic laws.
Language-integrated querying makes use of this
algebraic framework to turn declarative collection processing queries
into efficient nested calculus expressions.
%

%
%
Incremental view maintenance (IVM) by static query rewriting
(a.k.a.\ delta query derivation) has proven to be a highly useful and, 
for instance in the context of data warehouse loading,
an indispensable feature of many commercial data management systems.
With delta processing,
the results of a query are incrementally maintained by a 
delta query that, given the original input and an incremental update, 
computes the corresponding change of the output.
Query execution can thus be staged 
into an offline phase for running the query over an initial database and
materializing the result,
followed by an online phase in which the delta query is evaluated and its
result applied to the materialized view upon receiving updates.
This execution model means that one can do as much as possible 
once and for all before any updates are first seen, 
rather than process the entire input every time data changes.

Delta processing is worthwhile only if delta query evaluation is much cheaper
than full re-computation. In many cases deltas are actually 
asymptotically faster --
for instance, filtering the input based on some predicate takes linear time, 
whereas the corresponding delta query does not need to access the database but only considers
the incremental update, and thus runs in time proportional to the size of the update (in practice, usually constant time).

%
%
The benefits of incremental maintenance can be amplified if one applies
it recursively \cite{koch_ring} --
one can also speed up the evaluation of delta queries by 
materializing and incrementally maintaining their results using
second-order delta-queries (deltas of the delta queries).
One can build a hierarchy of delta queries, where the deltas at
each level are used to maintain the materialization of deltas above them, 
all the way up to the original query.  
This approach of {\em higher-order} delta derivation 
(a.k.a.\ recursive IVM)
admits a complexity-theoretic separation
between re-evaluation and incremental maintenance of 
positive relational queries with aggregates ($\apRA$)
\cite{koch_ring}, and
outperforms classical IVM by many orders of magnitude~\cite{ahmad14}.
Unfortunately,
the techniques described above target only flat relational
queries and as such cannot be used to enable incremental maintenance for  
collection processing engines.

%
%

In this work we address the problem of delta processing for 
positive nested-relational calculus on bags ($\pNRC$).
Specifically, we consider deltas for updates 
that are applied to the input
relations via a generalized bag union $\uplus$ (which sums up multiplicities),
where tuples have integer multiplicities in
order to support both insertions and deletions.
We formally define what it means for a nested update to be {\em incremental}
and a $\pNRC$ query to be {\em efficiently} incrementalizable, and we propose
the first solution for the efficient incremental maintenance of $\pNRC$ queries.

We say that a query is efficiently incrementalizable
if its delta has a lower cost than recomputation.
We define cost domains equipped with partial orders for every nested type
in $\pNRC$ and determine cost functions for the constructs of $\pNRC$
based on their semantics and a lazy evaluation strategy.
The cost domains that we use 
attach a cardinality estimate to 
each nesting level of a bag, where the cardinality of a nesting level
is defined as the maximum cardinality of all the bags with the same nesting
level.
For example, to the nested bag $\set{\set{a},\set{b},\set{c,d}}$ we associate a
cost value of $3\set{2},$ since the top bag has $3$ elements and the inner bags
have a maximum cardinality of $2$.
This choice of cost domains was motivated by the fact that  
data may be distributed unevenly across the nesting levels of a 
bag, while one can write queries that operate just on a particular 
nested level of the input.
Even though our cost model makes several conservative 
approximations,
it is still precise enough to separate incremental maintenance from
re-evaluation for a large fragment of $\pNRC$.

We efficiently incrementalize $\pNRC$ in two steps.
We first establish $\ipNRC,$ the largest fragment for which we can derive 
efficient deltas.
Then,
for queries in $\pNRC\setminus\ipNRC$, we provide a semantics
preserving translation into a collection of $\ipNRC$ queries on a differently
represented database.

For $\ipNRC$ we leverage the fact that our delta transformation is {\em closed}
(i.e.\ maps to the same query language)
and illustrate how to further optimize delta processing using recursive IVM:
if the delta of an $\ipNRC$ query still depends on the database,
it follows that it can be partially evaluated and 
efficiently maintained using a higher-order delta. 
We show that for any $\ipNRC$ query there are only a 
finite number of higher-order delta derivations possible before the resulting
expressions no longer depend on the database (but are purely functions of the update), and thus no longer require
maintenance.

The only queries that fall outside $\ipNRC$ are those that use the singleton
bag constructor $\sng(e),$ where $e$ depends on the database. 
This is supported by the intuition that in  $\pNRC$ 
we do not have an efficient way to modify
$\sng(e)$ into $\sng(e\uplus\Delta e),$ without first removing $\sng(e)$ from
the view and then adding $\sng(e\uplus \Delta e),$ which amounts to 
recomputation.
The challenge of efficiently applying updates to inner bags, a.k.a.\ {\em
deep} updates, does not lie in designing an operator that 
navigates the structure of a nested object and applies the update to the right
inner bag, but doing so while providing useful re-writing rules wrt.\ 
the other language constructs, which can be used to derive efficient delta
queries.
Previous approaches to incremental maintenance of nested views have 
either ignored the issue of deep updates \cite{gluche97}, 
handled it by triggering recomputation of nested bags \cite{liu99} 
or defaulted to change propagation \cite{nakamura01, kawaguchi97}.

We address the problem of efficiently incrementalizing $\sng(e)$ 
with {\em shredding}, a semantics-preserving transformation
that replaces the inner bag introduced by $\sng(e)$ with a label $l$ and
separately maintains the mapping between $l$ and its defining query $e.$   
Therefore, deep updates can be applied by simply modifying the label definition
corresponding to the inner bag being updated.
As such, the problem of incrementalizing $\pNRC$ queries is reduced to that
of incrementalizing the collection of $\ipNRC$ queries resulting from the
shredding transformation.
Furthermore, based on this reduction we also show that,
analogous to the flat relational case~\cite{koch_ring},
incremental processing of $\pNRC$ queries is in a strictly lower complexity
class than re-evaluation (NC$_0$ vs.\ TC$_0$).

The idea of encoding inner bags by fresh indices/labels 
and then keeping track of the mapping between the labels and the contents
of those bags has been studied before in the literature in various contexts
\cite{CheneyLW14,Koch05,Bussche,LevyS97,Suciu93,Grust2010}.
However we are, to the best of our knowledge, 
the first to propose a generic and compositional shredding transformation
for solving the problem of
efficient IVM for $\pNRC$ queries.
The compositional nature of our solution is essential for applications where 
nested data is exchanged between several layers of the system.

We summarize our contributions as follows:
\begin{itemize}
\item We define the notions of {\em incremental} nested update and 
{\em efficient incrementalization} of nested queries,
based on cost domains and a cost interpretation over $\pNRC$'s constructs. 
\item We  provide 
the first solution for the efficient incrementalization of
positive nested-relational calculus ($\pNRC$).

\item  We show how delta processing of nested queries can be further
optimized using recursive IVM~\cite{koch_ring}.

\item
We
show that incremental evaluation is in a strictly lower complexity class than
re-computation ($\NCz$ vs.\ $\TCz$).
\end{itemize}

The rest of the paper is organized as follows.
We first
introduce our approach for
the incrementalization of $\pNRC$ queries
on a motivating example and  
formally define
the variant of positive nested relational calculus that we use
.
The efficient delta processing of a large fragment of $\pNRC$
is discussed in Section~\ref{sec:nrc-inc}.
and in Section~\ref{sec:shred} we show how the full $\pNRC$ can be efficiently  
maintained. 
Finally, 
in Section~\ref{sec:related} we review the related literature.
Each of the main sections of the paper 
(sec.~\ref{sec:mot_ex}-\ref{sec:shred}) has a corresponding
appendix containing the proofs (omitted for space reasons)
and additional examples/discussions referenced in the body of that section.


\section{Motivating example}
\label{sec:mot_ex}



We follow the classical approach to incremental query evaluation,
which is based on applying certain syntactic transformations called 
``delta rules'' to the query expressions of interest
(in Appendix~\ref{app:flat-delta} we revisit how delta processing works for 
the flat relational case). 
In the following, we give some intuition for the difficulties that arise in 
finding a delta rules approach to the problem of incremental computation on
\emph{nested} bag relations.

{\bf Notation.} For a query $Q$ and relation $R$, 
we denote by $Q[R]$ the fact
that $Q$ is defined in terms of relation $R.$ 
We will sometimes simply write $Q,$ if $R$ is obvious from the context.

\begin{example}
We consider the query $\related$ that
computes for every movie in relation $\movs(name,gen,dir)$ 
a set of related movies 
which are either in the same genre $gen$ or share the same artistic director
$dir$.
We define $\related$ in Spark\footnote{%
To improve the presentation
we omitted Spark's boilerplate code.
}:
\vspace{-5pt}
\begin{small}
\begin{verbatim}     
case class Movie(name: String, gen: String, dir: String)
val M: RDD[Movie] = ...
val related = for(m <- M) yield (m.name, relB(m))
def relB(m: Movie) = 
  for(m2 <- M if isRelated(m,m2)) yield m2.name 
def isRelated(m: Movie, m2: Movie) = 
  m.name != m2.name && (m.gen==m2.gen || m.dir==m2.dir)
\end{verbatim}
\end{small}
\vspace{-5pt}
where 
\texttt{RDD} is Spark's collection type for distributed datasets,
$\relB(m)$ computes the names of all the  movies related to $m$ 
and
$\isRelated$ tests if two different movies are related by genre or director. 
We evaluate $\related$ on an example instance. 

\noindent
\begin{tabular}{l l l}
$\;\movs $ 
&&
$\;\related[\movs] $
\\
\begin{small}
\begin{tabular}{l | c | c}
\hline
$name$ & $gen$ & $dir$
\\
\hline
Drive & Drama & Refn 
\\
Skyfall & Action & Mendes 
\\
Rush & Action & Howard 
\\
\hline
\end{tabular}
\end{small}
&&
\begin{small}
\begin{tabular}{l | c}
\hline
$name$ & $\set{name}$ 
\\
\hline
Drive 
& $\set{}$
\\
Skyfall
& $\set{\text{Rush}}$
\\
Rush
& $\set{\text{Skyfall}}$
\\
\hline
\end{tabular}
\end{small}
\end{tabular}
\vspace*{2mm}

Now consider the outcome of updating $\movs$ with $\Delta\movs$ via bag union
$\uplus$, where $\Delta\movs$ is a relation with the same schema as $\movs$ and
contains a single tuple $\tuple{\text{Jarhead},\text{Drama},\text{Mendes}}$.

\noindent
\begin{tabular}{l@{\hskip2pt}l}
\\[-2mm]
\noindent
$\;\movs \uplus \Delta \movs $
&
$\;\related[\movs \uplus \Delta \movs] $
\\
\begin{small}
\noindent
\begin{tabular}{l | c | c}
\hline
$name$ & $gen$ & $dir$
\\
\hline
Drive & Drama & Refn 
\\
Skyfall & Action & Mendes
\\
Rush & Action & Howard 
\\
Jarhead & Drama & Mendes 
\\
\hline
\end{tabular}
\end{small}
&
\begin{small}
\begin{tabular}{l | c@{}}
\hline
$name$ & $\set{name}$ 
\\
\hline
Drive 
& $\set{\text{Jarhead}}$
\\
Skyfall
& $\set{\text{Rush, Jarhead}}$
\\
Rush 
& $\set{\text{Skyfall }}$
\\
Jarhead 
& $\set{\text{Drive, Skyfall}}$
\\
\hline
\end{tabular}
\end{small}
\end{tabular}

\end{example}

To incrementally update the result of $\related$ we design a set of delta rules 
that, when applied to the definition of $\related[\movs]$, give us an
expression $\inc(\related)[\movs,\Delta\movs]$ s.t.:
\begin{align*}
\\[-5mm]
\related[\movs \uplus \Delta\movs]~=~
\related[\movs] \uplus \inc(\related)[\movs,\Delta\movs].
\\[-5mm]
\end{align*}

\begin{sloppypar}
For our example, in order to 
modify $\related[\movs]$ into 
$\related [ \movs \uplus \Delta \movs]$,
without completely replacing the existing tuples\footnote{%
Maintaining the result of $\related$ by completely replacing the affected
tuples defeats the goal of making incremental computation more
efficient than full re-evaluation, as these tuples could be arbitrarily large.  
},
one would have to add the movie Jarhead to the inner bag of related movies for
Drive (same genre) and Skyfall (same director).
However, 
our target language of Nested Relational Calculus
(NRC)~\cite{buneman-kleisli:95,LellahiT97,BusscheGV07,BusscheV13}
(with bag semantics, where tuples have integer multiplicities
in order to support both insertions and deletions~\cite{LibkinW97,koch_ring}) is
not equipped with the necessary constructs for expressing this kind of changes, and efficiently processing such `deep' updates represents the main
challenge in incrementally maintaining nested queries.
Although update operations able to perform deep changes have been
proposed in the literature~\cite{Liefke99}, they lack the necessary re-write
rules needed for a {\em closed} delta transformation, 
which is a prerequisite for recursive IVM.
\end{sloppypar}

\sloppy
In order to make inner bags accessible by `deep' updates,
we must first devise a naming scheme to address them. 
We have two options:
i) we can either associate a label to each tuple in a bag and then
identify an inner bag based on this label and the index of the tuple component 
that contains the bag, or 
ii) we can associate a label to each inner bag, and separately maintain a
mapping between the label and the corresponding inner bag.
In other words, labels can either identify the position of an inner bag
within the nested value or serve as an alias for the contents of the inner bag.
For example, given a value
$X = \set{\tuple{a,\set{x_1,x_2}},\tuple{b,\set{x_3}}},$
the first alternative decorates it with labels as follows: 
$\set{l_1 \mapsto \tuple{a,\set{x_1,x_2}}, l_2 \mapsto \tuple{b,\set{x_3}}},$
and then addresses the inner bags by $l_1.2$ and $l_2.2.$
By contrast, the second approach creates the mappings
$l_1 \mapsto \set{x_1,x_2}$ and $l_2 \mapsto \set{x_3}$, and then
represents the original value as the flat bag 
$\fc{X} = \set{\tuple{a,l_1},\tuple{b,l_2}}.$

\fussy
Even though both schemes faithfully represent the original nested value,
we prefer the second one,  a.k.a.\ {\em
shredding}~\cite{CheneyLW14,Grust2010}, as it offers a couple of advantages.
Firstly, 
it makes the contents of the inner bags conveniently accessible to updates
via regular bag addition, without the need to introduce a custom update 
operation\footnote{%
The authors investigated this alternative and found it particularly
challenging due to the complex ways in which this custom operation would
interact with the existing constructs of the language.}.
Secondly, since inner bags are represented by labels it also avoids duplicating 
their contents. For example, when computing the Cartesian product of $X$ with
some bag $Y,$ one would normally create a copy of the tuples in $X$, along
with their inner bags, for each element of $Y$. 
Moreover, any update of an inner bag from $X$ would also have to be applied to 
every instance of that bag appearing in the output of $X \x Y$.  
By contrast, the second scheme computes the Cartesian product only 
between $X^F$ and $Y$, while the mappings between labels and the contents of 
the inner bags remain untouched. 
Therefore, any update to an inner bag of $X$ can
be efficiently applied just by updating its corresponding mapping.

For operating over nested values represented in shredded form,
we propose a semantics-preserving transformation that rewrites a query with 
nested output $Q[R]$ into a query $\fc{Q}$ returning the flat representation of 
the result, 
along with a series of queries $\gc{Q},$ 
computing the contents of its inner bags.

\subsection{Incrementalizing $\related$}
\label{sec:inc_related}

\sloppy
We showcase our approach on the motivating example by first expressing it in 
NRC. 
The main constructs that we use are:
i) the for-comprehension $\for{x}{Q_1~}\where~p(x)~\collect ~Q_2(x)$, which
iterates over all the elements $x$ from the output of query $Q_1$ that satisfy
predicate $p(x)$ and unions together the results of each $Q_2(x)$, and
ii) the singleton constructor $\sng(e)$, which creates a bag with the result of
$e$ as its only element.
\begin{align*}
\\[-6mm]
\related 
& \equiv
\for{m}{\movs}\,\collect\;\sng(\tuple{m.name,\relB(m)})
\\
\relB(m) 
& \equiv 
\for{m_2}{\movs}\;\where\;\isRelated(m,m_2)
\\
&\quad
\;\collect\;\sng(m_2.name). 
\\[-6mm]
\end{align*}
\fussy
Next, we investigate 
the incrementalization of the constructs used by the $\related$ query in order
to identify which one of them can lead to the problem of deep updates. 
The delta rule of the $\forz$ construct is a natural generalization of the 
rule for Cartesian product in relational algebra\footnote{%
$\inc(e_1 {\x} e_2) = \inc(e_1) {\x} e_2~\uplus~e_1 {\x} \inc(e_2)
                      ~\uplus~ \inc(e_1) {\x} \inc(e_2)
$}:
\begin{align}
\nonumber
\\[-7mm]
\label{eq:delta-for}
\inc(\for{x}{Q_1} \collect\; Q_2) 
& = 
\for{x}{\inc(Q_1)} \,\collect\; Q_2
\\
\nonumber
& \uplus 
\for{x}{Q_1} \,\collect\; \inc(Q_2)
\\
\nonumber
& \uplus 
\for{x}{\inc(Q_1)} \,\collect\; \inc(Q_2)
\\[-7mm]
\nonumber
\end{align}
assuming we can derive corresponding deltas for $Q_1$ and $Q_2$.

If the $\where$ clause is also present, the same rule applies because 
we only consider the positive fragment of nested bag languages, 
for which predicates are not allowed to test expressions of bag type 
(the reasoning behind this decision is detailed in 
Appendix~\ref{app:challenges}).
Therefore the predicates in the $\where$ clause can only be boolean
combinations of comparisons involving base type expressions and these
are not affected by updates of the database.

The difficulty arises when we try to design a delta rule for
singleton, specifically, how to deal with $\sng(e)$ when $e$ depends
on some database relation. There is plainly no way
in our calculus to express the change from $\sng(\movs)$ to 
$\sng(\movs \uplus \Delta \movs)$ in an efficient manner,
i.e., one that is 
proportional to the size of $\Delta \movs$ and not the size of the output.
This is the same problem that we saw with the $\related$ example above.
In Section~\ref{sec:nrc-inc} we will show that
$\sng(e)$ is the only construct in our calculus  
whose efficient incrementalization relies on `deep' updates.

\subsection{Maintaining inner bags}
\label{sec:app_upd}

In order to facilitate the maintenance of the bags produced by $\relB(m)$, we
associate to each one of them a label, and we store separately a mapping between
the label and its bag.
Then, for implementing updates to a nested bag, 
we can simply modify the definition of its associated label via bag union.
We note that this strategy can be applied for enacting `deep' changes to both
nested materialized views as well as nested relations in the database.

Since the bags created by $\relB(m)$ clearly depend on the variable $m$ bound by
the $\forz$ construct, we also incorporate the values that $m$ takes in the
labels that replace them.
The simplest way of doing so is to use labels that are pairs of 
indices and values, where the index uniquely identifies the inner query being
replaced. 
In our running example, as we have just a single inner query,
we only need one index $\iota$.
 
The shredding of $\related$ yields two queries, 
$\fc{\related}$ producing a flat version of $\related$
with its inner bags replaced by labels, and
$\gc{\related}$ that computes the value of a nested bag given a
label {\em parameter} $\ell$ of the form $\tuple{\iota,m}$
\begin{align*}
\\[-6mm]
\fc{\related} 
&\equiv 
\for{m}{\movs}\,\collect\;\sng(\tuple{m.name,\tuple{\iota,m}})
\\
\gc{\related}(\ell) 
&\equiv 
\for{m_2}{\movs}
               \;\where\;\isRelated(\ell.2,m_2)
\\
&\quad         \;\collect\;\sng(m_2.name)
\\[-6mm]
\end{align*}
The output of these queries on our running example is:
\newline
\begin{tabular}{>{$}l<{$} >{$}l<{$}}
\\[-2mm]
\;\fc{\related}[\movs]
&
\;\gc{\related}[\movs]
\\[-2mm]
\begin{small}
\begin{tabular}[t]{l | >{$}c<{$}}
\hline
name & \ell 
\\
\hline
Drive & \tuple{\iota,\tuple{\text{Drive,..}}} 
\\
Skyfall & \tuple{\iota,\tuple{\text{Skyfall,..}}}
\\
Rush & \tuple{\iota,\tuple{\text{Rush,..}}}
\\
\hline
\end{tabular}
\end{small}
&
\begin{small}
\begin{tabular}[t]{>{$}c<{$}@{\hskip8pt}>{$}c<{$}@{\hskip8pt}>{$}c<{$}}
\hline
\ell & \mapsto & \set{name}
\\
\hline
\tuple{\iota,\tuple{\text{Drive,..}}} & \mapsto
& \set{} 
\\
\tuple{\iota,\tuple{\text{Skyfall,..}}} & \mapsto
& \set{\text{Rush}} 
\\
\tuple{\iota,\tuple{\text{Rush,..}}} & \mapsto
& \set{\text{Skyfall}} 
\\
\hline
\end{tabular}
\end{small}
\end{tabular}
\vspace*{2mm}

Although in our example the generated queries are completely flat, this need not
always be the case. 
In particular, in order to avoid expensive pre-/post-processing steps, one
should perform shredding only down to the nesting level that is affected by the
changes in the input.

Upon shredding, the strategy for incrementally maintaining $\related$
is to materialize and
incrementally maintain $\fc{\related}$ and $\gc{\related},$ and then recover
$\related$ from the results based on the following equivalence:
\begin{align}
&
\nonumber
\\[-6.5mm]
&
\related =  
\for{r}{\fc{\related}}\;\collect\;
\nonumber
\\
&\qqquad\qquad\;\;
\sng(\tuple{r.1,\,\gc{\related}(r.2)}),
\nonumber
\\[-6.5mm]
\nonumber
\end{align}
which holds since the values that $m$ takes are incorporated in the 
labels $\ell$, and
$\gc{\related}(\ell)$ is essentially a rewriting of the subquery $\relB(m).$

We remark that, while being able to reconstruct $\related$ 
from $\fc{\related}$ and $\gc{\related}$ is important for proving the 
correctness of our transformation (see Section~\ref{sec:shred-corr}),
it is not essential for representing the final result since the labels 
that appear in $\fc{\related}$ can simply be seen as references to the inner 
bags.
We also note that even though $\gc{\related}$ is parameterized by $\ell$, 
one can use standard domain maintenance techniques to materialize it
since the relevant values of $\ell$ are ultimately 
those found in the tuples of $\fc{\related}$.
Finally, in this example the labels
are in bijection with the values over which $m$ ranges, and hence,
one could use those values themselves as labels. 
In general however we may have several nested subqueries that depend on the 
same variable $m$.

In the process of shredding queries
we replace every subquery of a singleton construct that depends
on the database with a label that does not. 
This is the case with the subquery $\relB(m)$
in $\related,$ and we have a very simple delta rule for 
expressions that do not depend on the input bags:
$\inc(\sng(\tuple{m.name,\,\tuple{\iota,m}})) = \inc(\sng(m_2.name)) =
\emptyset$.
Therefore, applying delta rules such as~(\ref{eq:delta-for}) gives us:
\begin{align}
\nonumber
\\[-6mm]
&
\inc(\fc{\related})
= 
\for{m}{\Delta\movs}\,
\collect\,\sng(\tuple{m.name,\tuple{\iota,m}})
\nonumber
\\
&
\inc(\gc{\related})(\ell) 
= 
\fortight{m_2}{\Delta\movs}\,\where\,\isRelated(\ell.2,m_2)
\nonumber
\\
&
\qqqquad
  \collect\;\sng(m_2.name)
\nonumber
\\[-6mm]
\nonumber
\end{align}

We shall prove in Section~\ref{sec:nrc-inc} that,
for the class of queries to which $\fc{\related}$ and $\gc{\related}$ belong,
the delta rules do indeed produce a proper update.
We remark that
since the domain of $\gc{\related}$ is determined by the labels in
$\fc{\related}$, it may be extended by the $\inc(\fc{\related})$ update.
Thus,
when updating the materialization of $\gc{\related}$
with the change produced by $\inc(\gc{\related})$, one must
also check whether each label in its domain has an associated definition, and if
not initialize it accordingly.

{\bf Cost analysis.}
In the following we show that maintaining $\related$ incrementally 
is more efficient than its re-evaluation 
(for the general case see Section~\ref{sec:cost-transf}). 
Let us assume that $M$ and $\Delta M$ have $n$ and $d$ tuples, respectively, 
including repetitions. 
From the expressions above it follows that the costs of computing
the original queries ($\fc{\related}$ and $\gc{\related}(\ell)$) is proportional
to the input,
while their deltas cost $O(d)$.

As previously noted,
$\related[\movs\uplus\Delta\movs]$ 
can be recovered from:
\begin{align*}
\\[-9mm]
&
\for{r}
{\fc{\related}[\movs\uplus\Delta\movs]}\;\collect\;
\\
&
\quad\;\;\;
\sng(\tuple{r.1,\,\gc{\related}[\movs\uplus\Delta\movs](r.2)}),
\\[-6mm]
\end{align*}
and by the properties of delta queries and one of
the general equivalence laws of the NRC~\cite{buneman-kleisli:95}, 
this becomes $V \uplus W$ where
\begin{align}
\nonumber
\\[-6mm]
\label{eq:V}
V  = &\; \for{r}{\fc{\related}[\movs]}\; \collect\;
\\
\nonumber
&\qquad
        \sng(\tuple{r.1,\,\gc{\related}[\movs](r.2)
               \uplus\inc(\gc{\related})(r.2)})
\\
\label{eq:W}
W  = &\; \for{r}{\inc(\fc{\related})}
\;\collect\;
\\
\nonumber
&\qquad
        \sng(\tuple{r.1,\,\gc{\related}[\movs\uplus\Delta\movs](r.2)})
\\[-6mm]
\nonumber
\end{align}


Even counting repetitions, 
we have $O(n)$ tuples in
the materialization of $\fc{\related}[\movs]$ while the result of computing 
$\inc(\fc{\related})$ has $O(d)$ tuples.
From~(\ref{eq:V}) the cost of computing $V$ is $O(nd)$ 
and from~(\ref{eq:W}) the cost of computing $W$ is $O(d(n+d))$,
where we assumed that unioning two already materialized bags
takes time proportional to the smaller one, and looking up the definition of a
label takes constant (amortized) time.
Thus, the incremental computation of $\related$ costs $O(nd+d^2)$. 
For the costs of maintaining $\fc{\related}$ and $\gc{\related}$ we have $O(d)$
and $O(d(n+d))$, respectively, considering that initializing the new labels
introduced by $\inc(\fc{\related})$ takes $O(dn)$ and then updating all the
definitions in $\gc{\related}$ takes $O((n+d)d)$ (which includes the cost of
rehashing the labels in $\gc{\related}$ as may be required due to its increase
in size).
It follows that the overall cost of IVM is $O(nd+d^2)$ and when $n \gg d$,
performing IVM is clearly much better than recomputing
$\related[\movs\uplus\Delta\movs]$ which costs $\Omega((n+d)^2)$ (in the
step-counting model we have been using).

In the next sections we develop our approach in detail.

\section{Calculus}
\label{sec:nrc-def}

We describe the version of the positive nested relational calculus ($\pNRC$) 
on bags that we use%
. 
Its types are:
\begin{align*}
\\[-6mm]
A, B, C\ &\coloneq\ 1\ \mid\ \Base\ \mid\ A \x B\ \mid\ \Bag(C),
\\[-6mm]
\end{align*}
where 
$\Base$ is the type of the database domain and
$1$ denotes the ``unit'' type (a.k.a. the type of the $0$-ary tuple $\tuple{}$).
We also use $\TBase$ to denote nested tuple types 
with components of only $\Base$ type.

In order to capture all updates, i.e., both insertions and deletions,
we use a generalized notion of bag where elements have
(possibly negative) integer multiplicities
and bag addition $\uplus$ sums multiplicities as integers.
In addition, for every bag type we have an empty bag constructor $\uZ$,
as well as construct $\ominus(e)$ that 
negates the multiplicities of all the elements produced by $e$.
We remark that, semantically, bag types along with empty bag $\uZ$,
bag addition $\uplus$ and bag minus $\ominus$ exhibit the structure of a
commutative group.
This implies that given any two query results $Q_{old}$ and $Q_{new}$,
there will always exist a value $\Delta Q$ s.t.\ $Q_{new} = Q_{old} \uplus
\Delta Q$.
This rich algebraic structure that bags exhibit is also the reason why we use
a calculus with bag, as opposed to set semantics.

\sloppy
Typed calculus expressions $\IGamma;\VGamma \vdash e : \Bag(B)$ have 
two sets of type assignments to variables
$\IGamma=X_1{:}\Bag(C_1),\cdots,X_m{:}\Bag(C_m)$ and 
$\VGamma=x_1{:}A_1,\cdots,x_n{:}A_n$,
in order to distinguish between variables $X_i$ defined via
$\keyw{let}$ bindings and which reference top level bags, and variables $x_i$
which are introduced within $\forz$ comprehensions and bind the inner elements
of a bag.
We maintain this distinction since in the process of shredding we will use the
latter set to generate unique labels, identifying shredded bags 
(section~\ref{sec:shredding_trans}).

\fussy
The typing rules and semantics of $\pNRC$ are given in 
Figure~\ref{fig:nrc_type_sem}, 
where 
$R$ ranges over the relations in the database,
$X$ and $x$ range over the variables in the contexts $\IGamma$ and $\VGamma$,
respectively,
$\lletz$ binds the result of $e_1$ to $R$ and uses it in the evaluation of
$e_2$, $\cprod$ performs Cartesian product of bags,
$\forz$ iteratively evaluates $e_2$ with $x$ bound to every element of $e_1$ 
and then unions together all the resulting bags, 
$\flt$ turns a bag of bags into just one bag by unioning the inner bags,
$\sng$ places its input into a singleton bag and
$p$ stands for any predicate over tuples of primitive values.
Compared to the standard formulation given in~\cite{buneman-kleisli:95}
we use a calculus version that is ``delta-friendly'' in that all expressions
have bag type and more importantly most of its constructs are either linear or
distributive wrt.\ to bag union, with the notable exception of $\sng(e)$.
Therefore
we have a bag (Cartesian) product construct instead of a pairing construct, 
we have a separate flattening construct, and we control carefully how
singletons are constructed (note that we have four rules for
singletons but they do not ``overlap''). 
Finally,
$\Irho$ and $\Vrho$ are assignments of values to variables,
and we denote their extension with a new assignment by 
$\Irho[X \coloneq v]$ and $\Vrho[x \coloneq v]$, respectively.
Throughout the presentation, we will omit such value assignments whenever they
are not explicitly needed for resolving variable names.

\begin{figure}[t!]

\begin{subfigure}{\columnwidth}
{\small
\begin{align*}
\\[-7mm]
&
\frac{ \sch(R) {=} B }
{
R : \Bag(B)
}
\;\;\;
\frac{ 
\IGamma; \VGamma \vdash e_1{:}\;\Bag(C)
\;\;\; 
\IGamma, X {:} \Bag(C); \VGamma \vdash e_2{:}\;\Bag(B)
}{ 
\IGamma; \VGamma \vdash \llet{ X }{ e_1 }\ e_2  : \Bag(B)
} 
\\[2mm]
&
\frac{}
{
\IGamma, X\,{:}\,\Bag(C); \VGamma \vdash X :\Bag(C)
}
\;\;\,
\frac{}
{
\IGamma; \VGamma, x\,{:}\,\TBase \vdash p(x) :\Bag(1) 
}
\\[2mm]
&
\frac{}
{ 
\IGamma; \VGamma, x{:}A \vdash \sng(x) : \Bag(A)
}
\quad\;\,
\frac{}
{
 \etpl : \Bag(1)
}
\quad\;
\frac{}
{
 \uZ : \Bag(B)
}
\\[5pt]
&
\frac{i=1,2}
{ 
\IGamma; \VGamma, x{:}A_1 {\x} A_2 \vdash \piTC{x}{i} : \Bag(A_i)
}
\;\;\,
\frac{
 e  : \Bag(B) 
}
{
 \sng(e) : \Bag(\Bag(B))
}
\\[5pt]
&
\frac{ 
\IGamma; \VGamma \vdash e_1  : \Bag(A)
\quad 
\IGamma; \VGamma, x {:} A \vdash e_2 : \Bag(B)
}{ 
\IGamma; \VGamma \vdash \for{ x }{ e_1 } \collects e_2  : \Bag(B)
}
\quad\;\;
\frac{
 e_{1,2}     : \Bag(B) 
}{ 
 e_1 {\uplus} e_2 : \Bag(B)
} 
\\[5pt]
&
\frac{
  e_i : \Bag(B_i), i=1,2 
}{ 
 e_1 \cprod e_2  : \Bag(B_1 {\x} B_2)
} 
\quad\;
\frac{
 e  : \Bag(\Bag(B)) 
}
{
 \flt(e) : \Bag(B)
}
\quad\;
\frac{
 e  : \Bag(B) 
}
{
 \ominus(e) : \Bag(B)
}
\end{align*}
}

\end{subfigure}
~
\begin{subfigure}{\columnwidth}
\begin{align*}
\\[-5mm]
&
\meane{R}{} = R
\qquad\quad\;\;\;
\meane{ \llet{ X }{ e_1 }\;e_2 }{\Irho;\Vrho} 
= \meane{e_2}{\Irho[X \coloneq \meane{e_1}{\Irho;\Vrho}];\Vrho}
\\
&
\meane{X}{\Irho;\Vrho} = \Irho(X)
\quad
\meane{p(x)}{\Irho;\Vrho} = 
\text{ if } p(\Vrho(x))\ \text{then}\ \;\set{\tuple{}}\ \;\text{else}\ \;\set{}
\\
&
\meane{\sngvar{x}}{\Irho;\Vrho} = \sngv{\Vrho(x)}
\qquad\quad\;\,
\meane{\piTC{x}{i}}{\Irho;\Vrho} = 
\sngv{\pi_i(\Vrho(x))}
\\
&
\meane{\sng(e)}{} = \set{\meane{e}{}} 
\qqquad\;\;\;
\meane{ \flt(e) }{} = {\biguplus}_{ v \in \meane{e}{} } v
\\
&
\meane{ \for{ x }{ e_1 } \collects e_2 }{\Irho;\Vrho} 
= {\biguplus}_{ v \in \meane{e_1}{\Irho;\Vrho} } 
  \meane{e_2}{\Irho;\Vrho[x \coloneq v]}
\\
&
\meane{e_1 \cprod e_2}{} = 
{\biguplus}_{ v_1 \in \meane{e_1}{} }
{\biguplus}_{ v_2 \in \meane{e_2}{} }
\set{\tuple{v_1,v_2}}
\quad
\meane{\etpl}{} = \sngv{\tuple{}}
\\
&
\meane{\uZ}{} = \set{}
\qquad
\meane{e_1 \uplus e_2}{} = \meane{e_1}{} \uplus
                               \meane{e_2}{}
\quad\;\;
\meane{\ominus(e)}{} = \ominus( \meane{e}{} ) 
\end{align*}

\end{subfigure}

\caption{Typing rules and semantics for $\pNRC$.}
\label{fig:nrc_type_sem}

\end{figure}

Booleans are simulated by $\Bag(1)$, 
with the singleton bag $\etpl$ denoting \textit{true} and the 
empty bag $\uZ$ denoting \textit{false}. Consequently, the return type of
predicates $p(x)$ is also $\Bag(1)$. 
The ``positivity''
of the calculus is captured by the restriction put on (comparison) predicates
$p(x)$ to only act on tuples of basic values 
since comparisons involving bags can be used to simulate
negation~\cite{buneman-kleisli:95}.
We discuss in Appendix~\ref{app:challenges} the challenges posed by
negation wrt.\ efficient maintenance within our framework.

\vspace{-2mm}
\begin{example}
\label{ex:filter}
Filtering an input bag $R$ according to some predicate $p$ can be defined
in $\pNRC$ as:
\begin{align*}
\\[-6mm]
\filter_p[R] &= \for{ x }{ R }\ \where\ p(x)\ \collect\ \sng(x)
\\[-6mm]
\end{align*}
considering that
the $\forz$ construct with $\where$
clause (also used in Section~\ref{sec:mot_ex}) can be expressed as follows:
\begin{align*}
\\[-6mm]
&
\for{ x }{ e_1 }\ \where\ p(x)\ \collect\ e_2 =
\\
&\qqquad
\for{ x }{ e_1 }\ \collect\ \for{ \_ }{ p(x) }\ \collect\ e_2,
\\[-6mm]
\end{align*}
where we ignore the variable binding the contents of the bag returned by 
predicate $p$ since its only possible value is $\tuple{}.$  
\end{example}

\sloppy
For a variable $X$ we say that an expression $e$ is 
{\em $X$-dependent} if $X$ appears as a free variable in $e$,
and {\em $X$-independent} otherwise.
Also,
among $\pNRC$ expressions we distinguish
between those that are {\em input-independent}, 
i.e.\ are {\em $R$-independent} for all relations $R$ in the database,
and those that are {\em input-dependent}.
We define $\ipNRC$ as the fragment of $\pNRC$ that uses a
syntactically restricted singleton construct $\sngr(e)$, 
where $e$ must be {\em input-independent}.
While this prevents $\ipNRC$ queries from adding nesting levels to their
inputs\footnote{%
We note that the query from Section~\ref{sec:mot_ex}
does not belong to $\ipNRC.$
}, it does provide the useful guarantee that their deltas
do not require deep updates. 
We take advantage of this fact in the next section, when we discuss the 
efficient delta-processing of $\ipNRC$.
For the incrementalization of the full $\pNRC$, 
we provide a shredding transformation taking any $\pNRC$ query into 
a series of $\ipNRC$ queries (see Section~\ref{sec:shred}).  

\fussy


\section{Incrementalizing $\ipNRC$}
\label{sec:nrc-inc}

In the following
we show that any query in $\ipNRC$ admits a delta expression with
a lower cost estimate than re-evaluation.
Since the derived deltas are also $\ipNRC$ queries, 
their evaluation can be optimized in the same way as the original query,
i.e.\ materialize and maintain them via delta-processing. 
We call the resulting expressions \emph{higher-order} deltas.
As each derivation produces `simpler' queries,
we show that the entire process has a finite number of steps and the final one
is reached when the generated delta no longer depends on the database.
Thus the maintenance of nested queries can be further optimized
using the technique of recursive IVM, which has delivered important speedups
for the flat relational case~\cite{ahmad14}. 

To simplify the presentation, we consider a database where a single relation $R$
is being updated.
Nonetheless, the discussion and the results carry over in a
straightforward manner when updates are applied to several relations.   

The delta rules for the constructs of $\ipNRC$ wrt.\ the update of
bag $R$ are given in Figure~\ref{fig:ipnra-delta}, where $\Delta R$ is
a bag containing the elements to be added/removed from $R$ (with
positive/negative multiplicity for insertions/deletions)
and we use $\lletBin{X}{e_1}{Y}{e_2}\ e$ as a shorthand for
$\llet{X}{e_1}\ (\llet{Y}{e_2}\ e)$.
The delta of constructs that do not depend on $R$ is the empty bag,
while the rules for the other constructs are a direct consequence of their
linear or distributive behavior wrt.\ bag union.
We show that indeed, the derived delta queries $\incR{R}(h)[R, \Delta R]$ 
 produce a correct update for the return value of $h:$

\begin{figure}[t]
\begin{align*}
\\[-8mm]
&
\incR{R}(R) = \Delta R
\quad\;\;\;
\incR{R}(X) = \uZ
\qquad\;\,
\incR{R}(p(x)) = \uZ
\quad\;\;\;
\incR{R}(\uZ) = \uZ 
\\[2pt]
&
\incR{R}( \llet{X}{e_1}\ e_2 ) = 
\lletBin{X}{e_1}{\Delta X}{\incR{R}(e_1)}\ 
\\[2pt]
&
\qqqquad\qquad
\incR{R}(e_2)\ \uplus\ \incR{X}(e_2)\ \uplus\ \incR{R}(\incR{X}(e_2))
\\[2pt]
&
\incR{R}(\sngvar{x}) = \uZ
\quad\;\;\;
\incR{R}(\piTC{x}{i}) = \uZ
\quad\;\;\,
\incR{R}(\etpl) = \uZ
\\[2pt]
&
\incR{R}(\sngr(e)) = \uZ
\quad\;\,
\incR{R}(\flt(e)) = \flt(\incR{R}(e))
\\[2pt]
&
\incR{R}(\for{x}{e_1} \collects e_2) = 
\for{x}{\incR{R}(e_1)} \collects e_2
\\[2pt]
&
\qqqquad\qquad\quad\,
\uplus\ \for{x}{e_1} \collects \incR{R}(e_2)
\\[2pt]
&
\qqqquad\qquad\quad\,
\uplus\ \for{x}{\incR{R}(e_1)} \collects \incR{R}(e_2)
\\[2pt]
&
\incR{R}(e_1 \cprod e_2) = 
\incR{R}(e_1) \cprod e_2 
\ \uplus\ e_1 \cprod \incR{R}(e_2)
\ \uplus\ \incR{R}(e_1) \cprod \incR{R}(e_2)
\\[2pt]
&
\incR{R}(e_1 {\uplus} e_2) = \incR{R}(e_1) \uplus \incR{R}(e_2)                      
\qqquad\;\;\,
\incR{R}(\ominus(e)) = \ominus(\incR{R}(e))                      
\\[-8mm]
\end{align*}

\caption{Delta rules for the constructs of $\ipNRC$}
\label{fig:ipnra-delta}
\end{figure}

\begin{restatable}{proposition}{deltathipNRA}
\label{prop:ipnra-delta}
Given an $\ipNRC$ expression $h[R]:\Bag(B)$ with input $R: \Bag(A)$ 
and update $\Delta R: \Bag(A)$, then:
\begin{align*}
\\[-6mm]
&
h[R\uplus\Delta R]~=~ 
h[R] ~\uplus~ \incR{R}(h)[R,\Delta R].
\\[-5mm]
\end{align*}
\end{restatable}
\begin{proof}(sketch)
The proof follows via structural induction on $h$ and from the semantics of 
$\ipNRC$ constructs (extended proof in Appendix~\ref{sec:app:delta}).
\end{proof}

\vspace{-5mm}

\begin{restatable}{lemma}{deltalemipNRAconst}
\label{lem:delta-const}
The delta of an {\em input-independent} $\ipNRC$ expression $h$ 
is the empty bag, $\incR{R}(h) = \uZ.$
\end{restatable}

The lemma above is useful for deriving in a single step the delta of 
{\em input-independent} subexpressions 
(as in Example~\ref{ex:filter-delta}), but it also plays
an important role in showing that deltas are cheaper than the original queries
(Theorem~\ref{prop:ipnra-cost})
and in the discussion of higher-order incrementalization
(Section~\ref{sec:higher-order-delta}).
 
{\bf Notation.} We sometimes write $\inc(h)$ instead of $\incR{R}(h)$ if the
updated bag $R$ can be easily inferred from the context. 

\begin{example}
\label{ex:filter-delta}
Taking the delta of the $\ipNRC$ query presented in Example~\ref{ex:filter}
results in:
\begin{align*}
\\[-6mm]
\incR{R}(\filter_p)
&= \for{ x }{ \Delta R }\ \where\ p(x) \collects \sng(x),
\\[-6mm]
\end{align*}
since 
$\incR{R}(\for{ \_ }{ p(x) } \collects \sng(x)) = \uZ$
(from Lemma~\ref{lem:delta-const})
and 
$\for{ x }{ e } \collects \uZ = \uZ.$
As expected the delta query of $\filter_p$ amounts to filtering the update:
$\filter_p[\Delta R]$. 
\end{example}

\subsection{Higher-order delta derivation}
\label{sec:higher-order-delta}

The technique of higher-order delta derivation 
stems from the intuition that if the evaluation of a query can be sped up
by re-using a previous result and evaluating a cheaper delta, 
then the same must be true for the delta query itself.
This has brought about an important leap forward in the incremental
maintenance of flat queries~\cite{ahmad14},
and in the following we show 
that our approach to delta-processing enables recursive IVM for
$\pNRC$ as well  
(since we derive `simpler' deltas expressed in the same language
as the original query).

The delta queries $\inc(h)[R, \Delta R]$ 
we generate may depend on both
the update $\Delta R$ as well as the initial bag $R.$
Considering that typically the updates are much smaller than the original bags
and thus the cost of evaluating $\inc(h)$ is most likely dominated by the
subexpressions that depend on $R$,
it is beneficial to partially evaluate $\inc(h)[R, \Delta R]$ offline
wrt.\ those subexpressions that depend only on $R$. 
Once $\Delta R$ becomes available, one can use the partially
evaluated expression of $\inc(h)$ to quickly compute the final update for 
$h[R].$   

\sloppy
However, 
since the underlying bag $R$ is continuously being updated,
in order to keep using this strategy we must be able to efficiently
maintain the partial evaluation of $\inc(h)$. 
Fortunately, $\inc(h)[R, \Delta R]$ is an $\ipNRC$ expression just like $h$,
and thus we can incrementally maintain its partial evaluation wrt.\ $R$ 
based on its second-order delta $\inc^2(h)[R, \Delta R, \Delta' R],$ 
as in
\begin{align*}
\\[-6.5mm]
\inc(h)[R \uplus \Delta' R, \Delta R]
=
\inc(h)[R, \Delta R]
\uplus
\inc^2(h)[R, \Delta R, \Delta' R],
\\[-6mm]
\end{align*}
where $\Delta' R$ binds the update applied to $R$ in $\inc(h)[R, \Delta R].$

\fussy
The same strategy can be applied to $\inc^2(h)$, leading to a series 
$\inc^k(h)$$[R, \Delta R,$$ \cdots, \Delta^{(k-1)} R]$
of partially evaluated higher-order deltas.
Each is used to incrementally maintain the preceding delta
$\inc^{k-1}(h),$ all the way up to the original query $h$.

\vspace{-2mm}
\begin{example}
Given bag $R:\Bag(\Bag(A))$ let us consider the first and second order deltas 
of query $h$
\begin{align*}
\\[-6mm]
&
h[R] = \flt(R) {\x} \flt(R)
\\
&
\inc(h)[R,\Delta R] = \flt(R) {\x} \flt(\Delta R) 
\\
&
\qqquad\quad\;\;
	\uplus 
     		\flt(\Delta R) {\x} 
			( \flt(R) {\uplus} \flt(\Delta R) )
\\
&
\inc^2(h)[\Delta R, \Delta' R]
= \flt(\Delta' R) {\x} \flt(\Delta R)
\\
&
\qqquad\qquad\quad
	\uplus \flt(\Delta R) {\x} \flt(\Delta' R).
\end{align*}
In the initial stage of delta-processing,
besides materializing $h[R]$ as $H_0,$ 
we also partially evaluate $\inc(h)$ wrt.\ $R$ as $H_1[\Delta R]$. 
Then, for each update $U$,
we maintain $H_0$ and $H_1[\Delta R]$ using:
\begin{align*}
\\[-6.5mm]
H_0 = H_0 \uplus H_1[U] 
\quad\;\;
H_1[\Delta R] = H_1[\Delta R] \uplus \inc^2(h)[\Delta R, U].
\\[-6.5mm]
\end{align*}
We note that one can apply updates over partially evaluated expressions
like $H_1[\Delta R]$ due to the rich algebraic structure of the calculus (bags
with addition and Cartesian product form a semiring) which makes it possible to
factorize $H_1[\Delta R] \uplus \inc^2(h)[\Delta R, U]$ into subexpressions that
depend on $\Delta R$, and subexpressions that do not.
Nonetheless, the process of compiling these expressions into highly optimized
trigger programs is outside the scope of this work.

Finally, we remark that in the traditional IVM approach, the value of
$\flt(R)$ which depends on the entire input $R$ is recomputed for each
evaluation of $\inc(h)[R,U]$, whereas with recursive IVM we evaluate it only
once during the initialization phase.
\end{example}
\vspace{-2mm}

Since we can always derive an extra delta query, this process 
could in principle generate an infinite series of deltas and thus render the
approach of recursive IVM inapplicable.
By contrast, we say that a query is {\em recursively incrementalizable}
if there exists a $k$ such that $\inc^k(h)$ no longer depends on the input
(and therefore there is no reason to continue the recursion and to derive a delta for it).
In our previous example, this happened for $k = 2$.
In the following we will show that any $\ipNRC$ query is 
{\em recursively incrementalizable}.

\sloppy
In order to determine the minimum $k$ for which $\inc^k(h)$ is 
input-independent we associate to every $\ipNRC$ expression a degree
$\dg_{\Drho}(h) {:} \bN$ as follows:
$\dg_{\Drho}(R) {\,=} 1$,
$\dg_{\Drho}(X) {\,=} \Drho(X)$,
$\dg_{\Drho}(h) {\,=} 0$ for 
$h \in \{\Delta R, \sngvar{x}, \piTC{x}{i}, \sngr(e), \uZ, p,$ $\etpl \}$
and:
\begin{align*}
\\[-6.5mm]
&
\dg_{\Drho}(e_1 \uplus e_2) = \max(\dg_{\Drho}(e_1), \dg_{\Drho}(e_2))
\\
&
\dg_{\Drho}(\for{x}{e_1} \collects e_2) = 
\dg_{\Drho}(e_1) + \dg_{\Drho}(e_2)
\\
&
\dg_{\Drho}(e_1 \cprod e_2) = 
\dg_{\Drho}(e_1) + \dg_{\Drho}(e_2)
\\
&
\dg_{\Drho}(\flt(e)) =
\dg_{\Drho}(\ominus(e)) = \dg_{\Drho}(e)
\\
&
\dg_{\Drho}(\llet{X}{e_1}\ e_2) = \dg_{\Drho[X\coloneq \dg_{\Drho}(e_1)]}(e_2)
,
\\[-7mm]
\end{align*}
where $\Drho$ associates a degree to each free variable $X$, 
corresponding to the degree of its defining expression.

\fussy
We remark that the expressions $h$ 
that have degree $0$ are exactly those which are {\em
input-independent}.
Therefore, determining the minimum $k$ s.t.\ $\inc^{k}(h)$ is 
{\em input-independent} means finding the minimum $k$ s.t.\ 
$\dg(\inc^{k}(h)) = 0,$ where $\inc^{0}(h) = h$.
In order to show that this $k$ is in fact the degree of $h,$
we give the following theorem, 
relating the degree of an expression to the degree of its delta.
   
\begin{restatable}{theorem}{PropDecreasingDegree}
\label{prop:delta-dec-deg}
Given an input-dependent $\ipNRC$ expression $h[R]$ then   
$\dg(\inc(h)) = \dg(h) - 1$.
\end{restatable}
\begin{proof}(sketch)
The proof follows by induction on the structure of $h$ and
from the definition of $\inc(\cdot)$ and $\dg(\cdot)$ (for the 
extended proof see Appendix \ref{sec:app:higher-delta}).
\end{proof}

Theorem~\ref{prop:delta-dec-deg} 
captures the fact that the delta of a $\ipNRC$ query is `simpler' than the
original query and we can infer from it that
$\dg(\inc^k(h)) = \dg(h) - k$. 
It then follows that $\dg(h)$ is the minimum $k$ s.t.\ 
$\dg(\inc^k(h)) {=} 0$, i.e.\ the minimum $k$ s.t. 
$\inc^{k}(h)$ is {\em input-independent}.

We conclude that with recursive IVM one can avoid computing over
the entire database during delta-processing
by initially materializing the given query and its deltas up to
$\inc^{\max(0,\dg(h)-1)}(h)$,
since those are the only ones that are {\em input-dependent}.
Then, maintaining each such materialized 
$H_i \coloneq \inc^i(h)$ is simply a matter of
partially evaluating $\inc^{i+1}(h)$ wrt.\ the update and applying it to $H_i$. 
Moreover, the ability to derive higher order deltas and materialize them wrt.\ 
the database is the key result that enables 
the \ACz\ vs.\ \NCz\ complexity separation between
nonincremental and incremental evaluation (Theorem~\ref{complex:main_result}).

\subsection{Cost transformation}
\label{sec:cost-transf}

Considering that delta processing is worthwhile only if the size of the change
is smaller than the original input, in this section we discuss what does it mean
in the nested data model for an update to be {\em incremental}.
Then, we provide a cost interpretation to every $\ipNRC$ expression that given
the size of its input estimates the cost of generating the output.
Finally, we prove that 
for incremental updates the derived delta query is indeed cost-effective
wrt.\ the original query.

While for the flat relational case incrementality can be simply defined in terms
of the cardinality of the input bag wrt. the cardinality of the update, this is
clearly not an appropriate measure when working with nested values, since 
an update of small cardinality could have arbitrarily large inner bags. 
In order to adequately capture and compare the size of nested values 
we associate to every type $A$ of our calculus a cost domain $\deg{A}$ 
equipped with a partial order and minimum values.
The definition of $\deg{A}$ is designed to preserve the distribution of
cost across the nested structure of $A$ in order to accurately reflect the size
of nested values and how they impact the processing of queries
operating at different nesting levels.
Thus, for every type in $\ipNRC$ we have:
\begin{align*}
\\[-6mm]
\deg{\Base} = \deg{1}
\quad\;\;
\deg{(A_1 {\x} A_2)} = \deg{A_1} \x \deg{A_2}
\quad\;\;
\deg{\Bag(A)} = \bagS{ \deg{A} }{\pbN},
\\[-6mm]
\end{align*}
where
$\deg{1}$ has only the constant cost 1, 
we individually track the cost of each component in a tuple,
and 
$\bagS{ \deg{A} }{\pbN}$ represents the cost of bags
as the pairing between their cardinality and
the least upper-bound cost of their elements\footnote{We use 
$\bagS{ \deg{A}}{\pbN}$ instead of $\pbN \times \deg{A}$ to distinguish it from
the cost domain of tuples.}. 
Additionally, we define a family of
functions $\sz_A: A \arr \deg{A},$ that associate to any value $a:A$ a cost 
proportional to its size:
\begin{align*}
\\[-7mm]
\sz_{\Base}(x) &= 1
\\
\sz_{A_1\x A_2}(\tuple{x_1,x_2}) &= \tuple{\sz_{A_1}(x_1),\sz_{A_2}(x_2)}
\\
\sz_{\Bag(C)}(X) &=  \bagS{ \sup_{x_i \in X}\ \sz_C(x_i) }{\card{X}}, 
\\[-7mm]
\end{align*} 
where the supremum function is defined based on the following
type-indexed partial ordering relation $\prec_A$:
\begin{align*}
\\[-6.5mm]
x &\prec_{\Base} y 
&&= \false
\\
\tuple{x_1,x_2} &\prec_{A_1 \x A_2} \tuple{y_1,y_2} 
&&= x_1 \prec_{A_1} y_1 \text{ and } x_2 \prec_{A_2} y_2
\\
\bagS{x}{n} &\prec_{\Bag(C)} \bagS{y}{m} 
&&= n < m \text{ and } x \preceq_{C} y.
\\[-6.5mm]
\end{align*} 
Finally, the $x \preceq_{A} y$ ordering is defined analogously to $\prec_A$ by
making all the comparisons above non-strict, with the exception of $\Base$
values for which we have $x \preceq_{\Base} y = \true$.
We denote by $\sOne_A$ the bottom element of $(\deg{A},\prec_A)$.

We can now say that an update $\Delta R$ for a nested bag $R$ is 
{\em incremental} if $\sz(\Delta R)~\prec~\sz(R)$.

\vspace{-2mm}
\begin{example}
\label{ex:cost-bag}
$\text{The}\, 
\text{size}\,
\text{of}\,
\text{bag}\,
R{:}\,\Bag(\String {\x} \Bag(\String)),$
\begin{align*}
\\[-6.5mm]
R &= 
\{ \tuple{ \text{Comedy}, 
                   \set{\text{Carnage}} },
           \tuple{ \text{Animation}, 
                   \set{\text{Up}, \text{Shrek}, \text{Cars}} } \} 
\\[-6.5mm]
\end{align*}
is estimated as
$
\sz(R) : \bagS{ \deg{1} \x \bagS{ \deg{1} }{\pbN} }{\pbN}
= \bagS{ \tuple{ 1, \bagS{1}{3} } }{2}.
$
\end{example}
\vspace{-2mm}

{\bf Notation.} Whenever the cardinality estimation of a bag is $1$, we simply 
write $\set{c}$ as opposed to $1 \set{c},$ where $c$ is the cost estimation for
its elements.

Given an $\ipNRC$ expression $e: \Bag(B),$
we derive its cost $\costt{e}{}: \bagS{\deg{B}}{\pbN}$ 
based on the transformation in Figure~\ref{fig:ipnra-cost},
where $\Irhoc$ and $\Vrhoc$ are cost assignments to variables.
The generated costs have two components: one that computes an upper bound 
for the cardinality of the output bag, 
denoted by $\costt{e}{o} : \pbN$, and another returning
the upper bound for the size of its elements 
$\costt{e}{i} : \deg{B}$.
If $B$ is itself a bag type $\Bag(C)$, we also denote the two components of
$\costt{e}{i}$ by $\costt{e}{oi} : \pbN$ and $\costt{e}{ii} : \deg{C}$.

The cost transformation follows the natural semantics of the constructs
in $\ipNRC$. 
For example, in the case of $\for{x}{e_1}\collects e_2$, the cardinality of the
output is estimated as the product of the cardinalities of the bags returned by
$e_1$ and $e_2$,
while the elements in the output
have the same cost as the elements returned by $e_2.$
We note that in computing the cost of $e_2$ 
we assigned to $x$ the estimated cost for the 
elements of $e_1$.  

\begin{figure}[t]
\begin{align*}
\\[-9mm]
&
\costc{R}{}{} = \size(R)
\qquad\qquad\;\;
\costc{\sngvar{x}}{}{\Irhoc;\Vrhoc} = \set{ \Vrhoc(x) }
\\
&
\costc{X}{}{\Irhoc;\Vrhoc} = \Irhoc(X)
\qquad\quad
\costc{\piTC{x}{i}}{}{\Irhoc;\Vrhoc} = \set{ \pi_i(\Vrhoc(x)) }
\\
&
\costt{p(x)}{} = \sOne_{\Bag(1)} 
\qquad\quad\;\,
\costt{\etpl}{} = \sOne_{\Bag(1)}
\\
&
\costt{\uZ}{} = \sOne_{\Bag(B)}
\qquad\qquad\;\,
\costt{\sngr(e)}{} 
= \set{ \costt{e}{} }
\\
&
\costt{\ominus(e)}{} 
= \costt{e}{}
\qquad\qquad\,
\costt{e_1 {\uplus} e_2}{} 
= \sup ( \costt{e_1}{} , \costt{e_2}{} )
\\
&
\costc{\llet{X}{e_1}\; e_2}{}{\Irhoc;\Vrhoc} = 
\costc{e_2}{}{\Irhoc[X \coloneq \costc{e_1}{}{\Irhoc;\Vrhoc}];\Vrhoc}
\\
&
\costt{e_1 \cprod e_2}{} = 
\bagS{ 
	\tuple{\costc{e_1}{i}{} , \costc{e_2}{i}{}}
}{
	\costc{e_1}{o}{} \cdot \costc{e_2}{o}{} 
}
\\
&
\costt{\flt(e)}{} 
= \bagS{ 
	\costc{e}{ii}{}
}{
	\costc{e}{o}{} \cdot \costc{e}{oi}{} 
}
\\
&
\costt{\for{x}{e_1} \collects e_2}{} = 
\\
&
\quad
\bagS{ \costc{e_2}{i}{ \Irhoc;
                        \Vrhoc[ x \coloneq \costc{e_1}{i}{}] } }
     { \costc{e_1}{o}{\Irhoc;\Vrhoc} {\cdot} 
       \costc{e_2}{o}{ \Irhoc;
                        \Vrhoc[ x \coloneq \costc{e_1}{i}{}] } }
\\[-7mm]
\end{align*}

\caption{
The cost transformation 
$\costt{f}{} {=} \bagS{\costt{f}{i}}{\costt{f}{o}} : \bagS{\deg{B}}{\pbN}$
over the constructs of $\ipNRC.$
}
\label{fig:ipnra-cost}
\end{figure}

Finally, we leverage the estimated cost of an expression to obtain an upper
bound on its running time:
\begin{lemma}
\label{lem:run_time}
An $\ipNRC$ expression $h{\,:\,} \Bag(B)$ can be evaluated in
$\Omega( \tcost_{\Bag(B)}(\costt{h}{}))$, 
where $\tcost_A{:} \deg{A} {\arr} \bN$ is defined as:
\begin{align*}
\\[-7mm]
&
\tcost_{Base}(c) = 1
\qquad
\tcost_{\Bag(C)}(\bagS{c}{n}) = n \cdot \tcost_C(c)
\\
&
\tcost_{A_1 \x A_2}(\tuple{c_1,c_2}) = \tcost_{A_1}(c_1) + \tcost_{A_2}(c_2).
\\[-5mm]
\end{align*}
\end{lemma}
\begin{proof}(Sketch)
\sloppy
In order to show that $h$ can be computed within
$\Omega( \tcost_{\Bag(B)}(\costt{h}{})) = 
\Omega(\costt{h}{o} \cdot \tcost_{B}(\costt{h}{i}))$ we assume that all
$\lletz$-bound variables have been replaced by their definition and we proceed
in two steps.
At first we compute a lazy version of the result $h^L = \meanl{h}{}$,
which instead of inner bags produces lazy bags $\lbag_{e,\Vrho}$, 
i.e.\ closures containing the expression $e$ that would have generated the
inner bag, along with $\Vrho$, the value assignment for $e$'s free variables at
the time of the evaluation.
The lazy evaluation strategy $\meanl{\cdot}{}$ operates similar to the standard
interpretation $\meane{\cdot}{}$, except for the singleton
construct $\meanl{\sng(e)}{\Vrho} = \lbag_{e,\Vrho}$ and for interpreting lazy
values $\meanl{\lbag_{e,\Vrho}}{\Vrho'} = \meanl{e}{\Vrho}$, for which we
replace the current value assignment $\Vrho'$ with the one stored in the
closure.
Considering that producing each element of $h^L$ takes constant time (since
building tuples and closures takes constant time), it follows that this step can
be done in time proportional to the cardinality of the output $O(\costt{h}{o})$.

\fussy
In the second step we expand the lazy values appearing in each element of
$h^L$ in order to obtain the final value of $h$. To do so we use the
following expansion function:
\begin{align*}
\\[-6mm]
&
\expf_{\Base}(x) {=} x,
\;\,
\expf_{A_1{\x}A_2}(\tuple{x_1,x_2}) {=} 
\tuple{\expf_{A_1}(x_1),\expf_{A_2}(x_2)}
\\
&
\quad
\expf_{\Bag(C)}(\lbag_{e,\Vrho}) = \forr{y}{\meanl{e}{\Vrho}}{\sng(\expf_C(y))}. 
\\[-6mm]
\end{align*}

We remark that, by postponing the materialization of inner bags until after the
entire top level bag has been evaluated, we avoid computing the contents of
nested bags that might get projected away in a later stage of the computation
(as might be the case for an eager evaluation strategy). 

Our result then follows from the fact that expanding each element $x:B$ from
$h^L$ takes at most $\tcost_{B}(\costt{h}{i})$, which can be easily shown through
induction over the structure of $B$ and considering that  
$\costt{h}{i}$ represents on upper bound for the size of the elements in the
output bag.
\end{proof}

\vspace{-5mm}
\begin{example}
If we apply the cost transformation to the $\related[M]$ query in
section~\ref{sec:inc_related} we get cost estimate:
\begin{align*}
\\[-6mm]
&
\costt{\related[M]}{} 
= \bagS{ \tuple{1,\bagS{1}{\card{M}}} }{ \card{M} },
\\[-6mm]
\end{align*}  
and an upper bound for its running time as 
$\Omega(\card{M} (1+\card{M}))$, which fits within the expected execution time
for this query.
\end{example}
\vspace{-2mm}

We can now give the main result of this section showing that 
for incremental updates delta-processing is more cost-effective than
recomputation.
\begin{restatable}{theorem}{costthipNRA}
\label{prop:ipnra-cost}
$\ipNRC$ is efficiently incrementalizable, i.e.\ 
for any {\em input-dependent} $\ipNRC$ query $h[R]$
and incremental update $\Delta R$, then: 
\begin{align*}
\\[-6mm]
&
\tcost(\costt{\inc(h)}{}) < \tcost(\costt{h}{}).
\\[-6mm]
\end{align*}
\end{restatable}
\begin{proof}(sketch)
We first show by induction on the structure of $h$ and using the cost
semantics of $\ipNRC$ constructs that 
$\costt{\inc(h)}{} \prec \costt{h}{}.$
Then the result follows immediately from the definition of $\tcost_A(\cdot)$ and
$\prec_A$ (for the extended proof see Appendix \ref{sec:app:cost}).
\end{proof}

\vspace{-2mm}
It can be easily seen that $\filter_p[R]$ is {\em efficiently incrementalizable}
since its delta is $\filter_p[\Delta R]$ and 
$\costt{\filter_p[R]}{} = \costt{R}{},$
therefore 
$\costt{\Delta R}{} \prec \costt{R}{}$
implies
$\costt{\filter_p[\Delta R]}{} \prec \costt{\filter_p[R]}{}.$


\section{Incrementalizing $\pNRC$}
\label{sec:shred}

We now turn to the problem of efficiently incrementalizing $\pNRC$ queries 
that make use of the unrestricted singleton construct.
As showcased in Section~\ref{sec:mot_ex}, 
an efficient delta rule for $\sng(e)$ requires deep updates which 
are not readily expressible in $\pNRC$.
Moreover, deep updates are necessary not only for maintaining the output of a 
$\pNRC$ query, but also for applying local changes to the inner bags of the 
input.
To address both problems  
we propose a shredding transformation that translates any $\pNRC$
query into a collection of efficiently incrementalizable expressions 
whose deltas can be applied via regular bag union. Furthermore, we show that 
our translation generates queries semantically equivalent to the original query,
thus providing the first solution for the efficient delta-processing of
$\pNRC$.

\subsection{The shredding transformation}
\label{sec:shredding_trans}

The essence of the shredding transformation is the replacement of inner bags by 
labels while separately storing their definitions in label dictionaries.  
Accordingly,
we inductively map every type $A$ of $\pNRC$ to a label-based/flat representation 
$\fc{A}$ along with a context component $\gc{A}$ for the corresponding  
label dictionaries:
\begin{align*}
\\[-6.5mm]
\fc{\Base} &= \Base
&
\gc{\Base} &= 1
\\
\fc{(A_1 {\x} A_2)} &= \fc{A_1} {\x} \fc{A_2} \quad
&
\gc{(A_1 {\x} A_2)} &= \gc{A_1} {\x} \gc{A_2}
\\
\fc{\Bag(C)} &= \bL 
&
\gc{\Bag(C)} &= (\bL \mapsto \Bag(\fc{C})) \x \gc{C}
\\[-6.5mm]
\end{align*}
For instance, the flat representation of a bag of type $\Bag(C)$ is a label 
$l:\bL$, whereas its context includes a label dictionary
$\bL{\mapsto}\Bag(\fc{C})$, mapping $l$ to the flattened contents of the bag.

The shredding transformation takes
any $\pNRC$ expression $h[R]{:} \Bag(B)$ to: 
\begin{align*}
\\[-6.5mm]
\fcf{h}[\fc{R}, \gc{R}] : \Bag(\fc{B})
\quad \text{and} \quad
\gcf{h}[\fc{R}, \gc{R}] : \gc{B},
\\[-6.5mm]
\end{align*}
where $\fcf{h}$ computes the flat representation of the output bag,
while the set of queries in $\gcf{h}$ 
define the context, 
i.e.\ the dictionaries corresponding to the labels introduced by $\fcf{h}$.
We note that the shredded expressions depend on the shredded input bag
$\fc{R} = \fcf{R},$ $\gc{R} = \gcf{R}$\footnote{%
We consider a full shredding of the input/output down to flat relations,
although the transformation can be easily fine-tuned in order to expose only
those inner bags that require updates. 
},
and that they make use of several new constructs for working with labels: 
the label constructor $\inL$, 
the dictionary constructor $[l \mapsto e]$, 
and the label union of dictionaries $\dcup$.
We denote by $\lpNRC$ and $\ilpNRC$, the extension with these constructs of
$\pNRC$ and $\ipNRC$, respectively,
but we postpone their formal definition until the following section.
Next, we discuss some of the more interesting cases of the shredding
transformation, for the full definition see Appendix~\ref{app:shred-trans}.

{\bf Notation.} We often shorthand $\fcf{h}$ and $\gcf{h}$ as 
$\fc{h}$ and $\gc{h},$ respectively.
We will also abuse the notation $\VGamma / \Vrho$ representing the type/value
assignment for the free variables of an expression introduced by $\forz$
constructs, to also denote a tuple type/value with one component for each such
free variable.

For the unrestricted singleton construct $\sng(e)$
we tag each of its occurrences in an expression with a unique static index
$\iota$. 
Given the shredding of $e$, 
$\fc{e} : \Bag(\fc{B})$, $\gc{e} : \gc{B},$
we transform $\sng_{\iota}(e)$ as follows: 
we first replace the inner bag $\fc{e}$ in its output
with a label $\tuple{\iota,\Vrho}$ using the label constructor
$\inL_{\iota,\VGamma},$ 
where $\Vrho:\VGamma$ represents 
the value assignment for all the free variables in $\fc{e}$. 
Since $\fc{e}$ operates only over shredded bags, it follows that $\Vrho$ is a
tuple of either primitive values or labels.
Then we extend the context $\gc{e}$ with a dictionary 
$[(\iota,\VGamma) \mapsto \fc{e}]$ mapping 
labels $\tuple{\iota,\Vrho}$ to their definition $\fc{e}$: 
\begin{align*}
\\[-6mm]
\fcf{ \sng_\iota(e) }: \Bag(\bL)\qqquad\;\;\;\; 
&= \inL_{\iota,\VGamma}(\Vrho)
\\
\gcf{ \sng_\iota(e) }: \bL \mapsto \Bag(\fc{B}) \x \gc{B} &= 
	        \tuple{ [(\iota,\VGamma) \mapsto \fc{e}],
				    \gc{e}
	        }. 
\\[-6mm]
\end{align*}
We incorporate the value assignment $\Vrho$ within labels as it allows us to
discuss the creation of labels independently from their defining dictionary.
Also, since the value assignment $\Vrho$ uniquely determines the definition of
a label $\tuple{\iota,\Vrho}$, this also ensures that we do not generate
redundant label definitions.
Since our results hold independently from a particular indexing scheme, we do
not explore possible alternatives, although they
can be found in the literature~\cite{CheneyLW14}.

For the shredding of $\flt(e), e:\Bag(\Bag(B))$, 
we simply expand the labels returned by 
$\fc{e}:\Bag(\bL)$, based on the corresponding
definitions stored in the first component of the context
$\gc{e}:  \bL \mapsto \Bag(\fc{B}) \x \gc{B}$:
\begin{align*}
\\[-7mm]
\fcf{ \flt(e) }: \Bag(\fc{B}) 
&= 
\for{l}{\fc{e}} \collects \ggc{e}{1}(l), 
\\[-7mm]
\end{align*}
where we denote by $\ggc{e}{1}/\ggc{e}{2}$ the first/second component of
$\gc{e}$.

Finally, for adding two queries in shredded form via $\uplus$, we add
their flat components, but we label union their contexts, i.e.\ their label
dictionaries:
\begin{align*}
\\[-7mm]
\fcf{ e_1 \uplus e_2 } &= \fc{e_1} \uplus \fc{e_2}
&&&
\gcf{ e_1 \uplus e_2 } &= \gc{e_1} \dcup  \gc{e_2}. 
\\[-7mm]
\end{align*}
 
To complete the shredding transformation we also inductively define 
$\shF_A : A \arr \Bag(\fc{A}) \text{ and } \shG_A : \gc{A},$ for shredding 
input bags $R : \Bag(A)$, as well as $\nst_{A}[\gc{a}]: \fc{A} \arr \Bag(A)$ 
for converting them back to nested form, as in:
\begin{align*}
\\[-6.5mm]
\fc{R} 
&= 
\for{r}{R} \collects
\sh^F_A(r)
\qquad
\gc{R} = \shG_A
\\
R &= 
\for{\fc{r}}{\fc{R}}
\collects
\nst_A[\gc{R}](\fc{r})
. 
\end{align*}
Shredding primitive values leaves them unchanged and produces no dictionary
($\gc{\Base} = 1$), while tuples get shredded and nested back component-wise.
For shredding inner bag values we rely on an
association between every bag value $v$ in the database and a label $l$, as
given via mappings $\gdic_C, \gdic^{-1}_C$:
\begin{align*}
\\[-6.5mm]
&
\gdic_C : \Bag(C) \arr \Bag(\bL)
&&
\gdic_C(v) = \sngv{l}
\\
&
\gdic^{-1}_C : \bL \darr \Bag(C) 
&&
\gdic^{-1}_C(l) = v.
\\[-6.5mm]
\end{align*} 
The shredding context for these labels is then obtained by 
mapping each label $l$ from the dictionary $\gdic^{-1}_C$
to a shredded version of its original value $v$.
The full details for the definition of $\sh^F, \sh^\gG$ and $\nst$ can be found
in Appendix~\ref{app:shred-trans}.

\subsection{Working with labels}
\label{sec:work_labels}

In the following we detail the semantics of $\ilpNRC$'s constructs for operating
on dictionaries and we show that
$\ilpNRC$ is indeed efficiently incrementalizable.

Given an expression $e : \Bag(B)$ with 
a value assignment for its free variables $\Vrho: \VGamma$, 
we define a label dictionary $[(\iota,\VGamma) \mapsto e]: \bL \darr \Bag(B)$, 
i.e.\ a mapping between labels $l = \tuple{\iota,\Vrho}$ and bag values
$e:\Bag(B)$, as:
\begin{align*}
\\[-6mm]
[(\iota,\VGamma) \mapsto e](\tuple{\iota',\Vrho}) 
&= 
\text{ if }(\iota == \iota')\ \rename_{\Vrho}(e)
\text{ else }\eptybag
\\[-6mm]
\end{align*} 
where $\rename_{\Vrho}(e)$ replaces each free variable from $e$ with its
corresponding projection from $\Vrho$.
A priori, such dictionaries have infinite domain, i.e.\ they produce a bag for
each possible value assignment $\Vrho$.
However, when materializing them as part of a shredding context we need only
compute the definitions of the labels produced by the flat version of the query.
\vspace{0mm}
\sloppy
\begin{example}
Given $\relB(m){:}\,\Bag(String)$, the query from the motivating example in 
section~\ref{sec:mot_ex},  
dictionary
$d\;{=}\;[(\iota,\Movie)\;{\mapsto}\;\relB(m)]$ of type 
$\bL\;{\darr}\;\Bag(String)$ builds a mapping between labels 
$l\;{=}\;\tuple{\iota,m}$ and the bag of related movies computed by $\relB(m)$,
where $l$ need only range over the labels produced by $\fc{\related}$.
\vspace{-2mm}

\fussy
\end{example}

{\bf Notation.}
We will often abuse notation and use $l$ to refer to both the kind of a label
$(\iota,\VGamma)$, as well as an instance of a label $\tuple{\iota,\Vrho}$.

In order to distinguish between an empty definition, $[] = \uZ$, 
and a definition that maps its label to the empty bag, $[l \mapsto \uZ]$, 
we attach support sets to label definitions such that $\supp([]) = \emptyset$ and 
$\supp([l \mapsto e])=\set{ l }.$

For combining dictionaries of labels,
i.e.\ $d = [l_1 {\mapsto} e_1, \cdots, l_n {\mapsto} e_n] : \bL \darr \Bag(B)$,
with $\supp(d) = \set{l_1, \cdots, l_n}$,
we define
the addition of dictionaries 
$(d_1 \uplus d_2)(l) = d_1(l) \uplus d_2(l)$
as well as the {\em label union} of dictionaries 
$d_1 \dcup d_2$,
where $ d_1, d_2: \bL {\darr} \Bag(B)$, 
 $\supp(d_1 \dcup d_2) = \supp(d_1) \cup \supp(d_2)$ and:
\begin{align*}
\\[-5mm]
(d_1 \cup d_2)(l) &= d_1(l), \text{ if } l \in \supp(d_1) {\setminus} \supp(d_2)
\\
(d_1 \cup d_2)(l) &= d_2(l), \text{ if } l \in \supp(d_2) {\setminus} \supp(d_1)
\\
(d_1 \cup d_2)(l) &= d_1(l), \text{ if } l \in \supp(d_1) {\cap} \supp(d_2) 
                             \;\&\;d_1(l)\;\!{=}\;\!d_2(l) 
\\
(d_1 \cup d_2)(l) &= \text{error}, 
                     \text{ if } l \in \supp(d_1) {\cap} \supp(d_2) 
                     \;{\&}\;d_1(l)\;\!{\neq}\;\!d_2(l) 
\\[-5mm]
\end{align*}

We ensure the well definedness of the label union operation 
by requiring that the definitions of labels found in both input dictionaries 
must agree,
i.e.\ for any $l \in \supp(d_1) \cap \supp(d_2)$ we must have
$d_1(l) = d_2(l)$. 
If this condition is not met the evaluation of $\dcup$ will result in an error.
We remark that $\cup$ cannot modify a label definition, only $\uplus$ can 
(for an example contrasting their semantics see 
Appendix~\ref{app:shred:ex-lab-dic}).
Moreover, we formalize the notion of consistent shredded values, 
i.e.\ values that do not contain undefined labels 
or definitions that conflict
and we show that shredding produces 
consistent values and that given consistent inputs, shredded $\pNRC$
expressions also produce consistent outputs
(Appendix~\ref{app:shred-val-cons}). 
This is especially important for guaranteeing
that the union of dictionaries performed by the
shredded version of bag addition cannot change the expansion of any label.

Finally, we introduce the delta rules and the degree and
cost interpretations for the new
label-related constructs:
\begin{align*}
\\[-6mm]
&
\inc( [l \mapsto e] ) = [l \mapsto \inc(e)]
\;\;\;
\inc( \inL_l ) = \uZ
\;\;\;
\inc( e_1 {\dcup} e_2 )\,{=}\,\inc(e_1)\,{\dcup}\,\inc(e_2)
\\
&
\dg([l \mapsto e]) = \dg(e)
\qquad
\qqquad
\dg(\inL_l) = 0
\quad
\\
&
\dg( e_1 \dcup e_2 ) = \max(\dg(e_1), \dg(e_2))
\\
&
\costt{[l \mapsto e](l')}{} = \costt{e}{}
\qquad
\qqquad
\costt{\inL_l(a)}{} = \set{1}
\\
&
\costt{(e_1 {\dcup} e_2)(l)}{} = \sup( \costt{e_1(l)}{}, \costt{e_2(l)}{} )
,
\\[-6mm]
\end{align*}
where the cost domains for labels is $\deg{1}$.
Based on these definitions we prove the following result:

\vspace{-2mm}
\begin{restatable}{theorem}{deltathilpNRA}
\label{th:ilpNRA-delta}
$\ilpNRC$ is recursively and efficiently incrementalizable,
i.e.\ given any input-dependent $\ilpNRC$ query $h[R],$
and incremental update $\Delta R$ then:
\begin{align*}
\\[-5mm]
&
h[R\uplus\Delta R]= 
h[R] \uplus \inc(h)[R,\Delta R],
\quad 
\dg(\inc(h)) = \dg(h) - 1
\\
&
\qquad\text{\ \ and\ \ }\qquad
\tcost(\costt{\inc(h)}{}) < \tcost(\costt{h}{}).
\\[-5mm]
\end{align*}
\end{restatable}

Theorem~\ref{th:ilpNRA-delta} implies that we can
efficiently incrementalize any $\pNRC$ query by 
incrementalizing the $\ilpNRC$ queries resulting from its shredding.
The output of these queries faithfully represents the expected nested
value as we demonstrate in section~\ref{sec:shred-corr}.

\subsection{Correctness}
\label{sec:shred-corr}
In order to prove the correctness of the shredding transformation, 
we show that for any $\pNRC$ query $h[R]: \Bag(B)$,  
shredding the input bag $R : \Bag(A)$, 
evaluating $\fc{h}, \gc{h}$,
and converting the output back to nested form produces the same result as 
$h[R]$, that is: 
\begin{align}
\nonumber
\\[-6mm]
\nonumber
h[R] &= 
\lletBin{ \fc{R} }{ \for{r}{R} \collects \shF(r) }{ \gc{R} }{ \shG }
\\
\label{eq:shred-corr}
&
\quad\,
\for{ \fc{x} }{ \fc{h} } \collects
\nst[ \gc{h} ]( \fc{x} ),
\\[-6mm]
\nonumber
\end{align}
where $\shF(r)$ shreds each tuple in $R$ to its flat representation,
$\shG$ returns the dictionaries corresponding to the labels generated by
$\shF(r)$, and $\nst[ \gc{h} ]( \fc{x} )$ places each tuple from $\fc{h}$ back
in nested form using the dictionaries in $\gc{h}$.  

We proceed with the proof in two steps. We first show that shredding a value
and then nesting the result returns back the original value 
(Lemma~\ref{nst_sh_prop}). 
Then, we show that applying the shredded version of a function over a 
shredded value and then nesting the result is equivalent to first nesting the
input and then applying the original function 
(Lemma~\ref{nst_nat_transf_prop}).
The main result then follows immediately (Theorem~\ref{cor:main_result}).

\vspace{-1mm}
\begin{restatable}{lemma}{nstShProp}
\label{nst_sh_prop}
The nesting function $\nst$ is left inverse wrt.\ the shredding functions
$\shF, \shG$, i.e.\ for nested value $a:A$ we have
$\for{ \fc{a} }{ \shF_A(a) } \collects \nst_A[ \shG_A ](\fc{a}) 
= \sng(a)$.
\end{restatable}

\vspace{-4mm}
\begin{restatable}{lemma}{nstNatTransfProp}
\label{nst_nat_transf_prop}
For any $\pNRC$ query $h[R]: \Bag(B)$ and consistent 
shredded bag $\fc{R},\gc{R}$: 
\begin{align*}
\\[-7mm]
&
\llet{ R }{ \for{ \fc{r} }{ \fc{R} } \collects \nst[\gc{R}]( \fc{r} ) }\ 
h[ R ] 
\\
&
\ =\ 
\for{ \fc{x} }{ \fc{h} } \collects
\nst[ \gc{h} ]( \fc{x} )  
      .
\\[-6mm]
\end{align*}
\end{restatable}

\begin{theorem}
\label{cor:main_result}
For any $\pNRC$ 
query
property (\ref{eq:shred-corr}) holds.
\end{theorem}
\begin{proof}	
The result follows from Lemma~\ref{nst_nat_transf_prop},
if we consider the shredding of $R$ as input, 
and then apply Lemma~\ref{nst_sh_prop}.
\end{proof}

\subsection{Complexity class separation}
\label{sec:complexity}

In terms of data complexity, $\NRC$ belongs to \TCz\ 
\cite{ST1994,Koch05}, the class of languages recognizable by \logspace-uniform
families of circuits of polynomial size and constant depth using and-, or- and 
majority-gates of unbounded fan-in.
The positive fragment of $\NRC$ is in the same complexity class
since just the flatten operation on bag semantics requires the power to compute
the sum of integers, which is in \TCz.
In the following, we show that incrementalizing $\pNRC$ queries in shredded form 
fits within the strictly lower complexity class of \NCz, 
which is a better model for real hardware since,
in contrast to \TCz, it uses only gates with bounded fan-in.
To obtain this result we require that multiplicities are 
represented by fixed size integers of $k$ bits, and thus their value is computed 
modulo $2^k.$ 

Assume that, for the following circuit complexity proof,
shredded values are available as a bit sequence, with
$k$ bits (representing a multiplicity modulo $2^k$)
for each possible tuple constructible from the active domain of
the shredded views and their schema, in some canonical ordering.
For $k=1$, this is the standard representation for circuit complexity proofs
for relational queries with set semantics.
Note that the active domain of a shredded view consists of the active domain
of the nested value it is constructed from, the delimiters
$``\la",``\ra",``,",``\{",``\}"$, as well as an additional
linearly-sized label set. We consider this the {\em natural} bit sequence
representation of shredded values.

It may be worth pointing out that shredding only creates polynomial blow-up compared to a string representation of a complex value (e.g.\ in XML or JSON). This further justifies our representation. Generalizing the classical bit representation
 of relational databases (which has polynomial blow-up) to non-first normal form relations (with, for the simplest possible type $\{\la\{\Base\}\ra\}$, one bit for every possible subset of the active domain) has exponential blow-up.

\vspace{-2mm}
\begin{restatable}{theorem}{apNRANCz}
\label{complex:main_result}
Materialized views of $\pNRC$ queries with multiplicities modulo $2^k$
in shredded form are incrementally
maintainable in $\NCz$ wrt.\ constant size updates.
\end{restatable}
\vspace*{-4mm}
\begin{proof}
We will refer to the database and the update by $d$ and $\Delta d$,
respectively.
By Theorem~\ref{cor:main_result}, every $\pNRC$ query can be simulated by
a fixed number of $\ipNRC$ queries on the shredding of the input.
By Proposition~\ref{prop:ipnra-delta}, for every $\ipNRC$ query $h$, there
is an $\ipNRC$ query $\delta_d(h)$
such that $h(d \uplus \Delta d) = h(d) \uplus \delta_d(h)(d)(\Delta d)$.
We partially evaluate and materialize such delta queries as views
$h' \coloneq\ \delta_d(h)(d)$ which then allow lookup of $h'(\Delta d)$.
By Theorem~\ref{prop:delta-dec-deg}, given an $\ipNRC$ query $h$,
there is a a finite stack of higher-order delta queries 
$h_{0}, \cdots, h_{k}$ 
(with $h_{i} = \delta_d^{(i)}(h)(d)$, $0 \le i \le k$,
and $\delta_d^{(0)}(h)(d) = h(d)$) 
such that $h_{k}$ is input-independent (only depends on $\Delta d$).
Thus, $h_i$ can be refreshed as $h_i := h_i \uplus h_{i+1}(\Delta d)$
for $i < k$.
We can incrementally maintain overall query $h$ on a group of
views in shredded representation using just the $\uplus$ operations
and the operations of $\ipNRC$ on a constant-size input 
(executing queries $h_i$ on the update). 
This is all the work that needs to be done, for an update,
to refresh all the views.

It is easy to verify that in natural bit sequence representation of the 
shredded views, both $\uplus$ (on the full input representations) and
$\ipNRC$ on constantly many input bits can be modeled using \NCz\ circuit
families, one for each meaningful size of input bit sequence.
For $\ipNRC$ on constant-size inputs, this is obvious, since all
Boolean functions over constantly many input bits can the captured
by constant-size bounded fan-in circuits, and since there is really only
one circuit, it can also be output in LOGSPACE.
For $\uplus$, remember that we represent multiplicities modulo $2^k$, i.e.\
by a fixed $k$ bits. Since addition modulo $2^k$ is in 
\NCz, so is $\uplus$:
The view contains aggregate multiplicities, each of which only needs to be 
combined with one multiplicity from the respective delta view.
The overall circuit for an input size is a straightforward composition of these
building blocks.
\end{proof}

In contrast, even when multiplicities are modeled modulo $2^k$ and the input
is presented in flattened form, $\pNRC$ is not in $\NCz$ since multiplicities
of projections (or $\flt$) depend on an unbounded number of input bits.

In Appendix~\ref{app:complex}, we show that shredding 
(for the initial materialization of the views) itself is in $\TCz$; 
it follows immediately that shredding constant-size updates -- the only
shredding necessary during IVM -- is in $\NCz$.

\section{Related Work}
\label{sec:related}

{\bf Delta derivation} was originally proposed
for datalog programs~\cite{DBLP:conf/deductive/GuptaKM92,gupta-dred:93} 
but it is even more natural for algebraic query languages such as
the relational algebra on bags~\cite{GriffinL95,
DBLP:conf/sigmod/BlakeleyLT86mix,
DBLP:conf/vldb/CeriW91, 
roussopoulos-tods:91mix,
dimitra98}, 
simply because the algebraic structure of a group is
the necessary and sufficient environment in which deltas live.
In many cases the derived deltas are asymptotically faster than the original 
queries and the resulting speedups prompted a wide adoption of such 
techniques in commercial database systems.
Our work is an attempt to develop similarly powerful static incrementalization 
tools for languages on nested collections and comes in the context of 
advances in the complexity class separation between recomputation and
IVM~\cite{koch_ring,schwentick14}.
Compared to~\cite{koch_ring} which discusses the recursive incrementalization of
a flat query language, we address the challenges raised by a nested data
model, i.e.\ we design a {\em closed} delta transformation for $\ipNRC$'s constructs 
and a semantics-preserving shredding transformation for 
implementing `deep' updates.
Furthermore, we provide cost domains and a cost interpretations for
 $\ipNRC$'s constructs, according to which
we define the notion of an {\em incremental} nested update and we show that the
deltas we generate have lower upper-bound time estimates than re-evaluation.

The {\bf nested data model} has been thoroughly studied in the literature over 
multiple decades and has enjoyed a wide adoption in industry in the form of data format
standards like XML or JSON. 
However,
solutions to the problem of incremental maintenance for nested queries
 either focus only
on the fragment of the language that does not generate changes to inner
collections~\cite{gluche97}, 
or propagate those changes based on auxiliary data structures designed to
track the lineage of tuples in the 
view~\cite{foster08, dimitrova03, kawaguchi97, nakamura01}.
The use of dedicated data-structures as well as custom update
languages make it extremely difficult to further apply recursive IVM on top of
these techniques.
In contrast, our approach is fully algebraic and both the given query as well
as the generated deltas belong to the same language and thus they can
be further incrementalized via delta processing.

The related topic of {\bf incremental computation} has also received
considerable attention within the programming languages community,
with proposals being divided between dynamic and static approaches.
The dynamic solutions, such as {\em self-adjusting computation} 
\cite{
Acar02, 
Acar08, 
Acar10}, 
record at runtime the dependency-graph of the computation.
Then, upon updates, 
one can easily identify the intermediate results affected and trigger
their re-evaluation.    
As this technique makes few assumptions about its target language,
it is applicable to a variety of languages ranging from Standard ML to C. 
Nonetheless, its generality comes at the price of significant runtime
overheads for building the dependency graph.
Moreover,
while static solutions derive deltas that can be further
optimized via global transformations, such an opportunity is mostly
missed by dynamic approaches
.

Delta derivation has also been proposed in the context of incremental 
computation,
initially only for first-order languages~\cite{Paige82}, and more recently
it has been extended to higher-order languages~\cite{CaiGRO13}.
However, 
these approaches offer no guarantees wrt.\ the efficiency of the generated 
deltas,
whereas in our work we introduce cost interpretations and discuss the
requirements for cost-efficient delta processing.

The challenge of {\bf shredding} nested queries has been previously addressed by
Paredaens et al.~\cite{Paredaens92}, who propose a translation taking
flat-to-flat nested relational algebra expressions into flat relational algebra.
Van den Bussche \cite{Bussche} also showed that it is possible to 
evaluate nested queries over sets via multiple flat queries, 
but his solution may produce results that are quadratically larger than 
needed \cite{CheneyLW14}.

Shredding transformations have been studied more recently in the context of
language integrated querying systems such as Links \cite{Links} and Ferry 
\cite{Ferry}.
In order to efficiently evaluate a nested query,
it is first converted to a series of flat queries 
which are then sent to the database engine for execution.
While these transformations also replace inner collections with flat values, 
they are geared towards generating SQL queries and thus they make assumptions 
that are not applicable to our goal of efficiently incrementalizing
any nested-to-nested expressions.
For example, 
Ferry makes extensive use of On-Line Analytic Processing (OLAP) features of 
SQL:1999, such as \texttt{ROW\_NUMBER} and \texttt{DENSE\_RANK}~\cite{Grust2010}, 
while Links 
relies on a normalization phase and
handles only flat-to-nested expressions \cite{CheneyLW14}. 
More importantly, none of the existing proposals translate $\pNRC$ queries to 
an efficiently incrementalizable language. 

\begin{footnotesize}
\newcommand{\angstrom}{\mbox{\normalfont\AA}}
\bibliography{mix,main,bibtex,bibtex05}{}

\newcommand{\SortNoOp}[1]{}
\begin{thebibliography}{10}

\bibitem{Acar08}
Umut~A. Acar, Amal Ahmed, and Matthias Blume.
\newblock Imperative self-adjusting computation.
\newblock In {\em Proc. POPL}, pages 309--322, 2008.

\bibitem{Acar10}
Umut~A. Acar, Guy Blelloch, Ruy Ley-Wild, Kanat Tangwongsan, and Duru Turkoglu.
\newblock Traceable data types for self-adjusting computation.
\newblock In {\em Proc. PLDI}, pages 483--496, 2010.

\bibitem{Acar02}
Umut~A. Acar, Guy~E. Blelloch, and Robert Harper.
\newblock Adaptive functional programming.
\newblock In {\em Proc. POPL}, pages 247--259, 2002.

\bibitem{BIS1990}
David A.~M. Barrington, Neil Immerman, and H.~Straubing.
\newblock {``On Uniformity within NC1''}.
\newblock {\em Journal of Computer and System Sciences}, {\bf 41}(3):274--306,
  1990.

\bibitem{DBLP:conf/sigmod/BlakeleyLT86mix}
Jos{\'e}~A. Blakeley, Per-{\AA}ke Larson, and Frank~Wm. Tompa.
\newblock Efficiently updating materialized views.
\newblock In {\em Proc.\ SIGMOD Conference}, pages 61--71, 1986.

\bibitem{buneman-kleisli:95}
Peter Buneman, Shamim~A. Naqvi, Val Tannen, and Limsoon Wong.
\newblock Principles of programming with complex objects and collection types.
\newblock {\em Theor. Comput. Sci.}, 149(1):3--48, 1995.

\bibitem{CaiGRO13}
Yufei Cai, Paolo~G. Giarrusso, Tillmann Rendel, and Klaus Ostermann.
\newblock A theory of changes for higher-order languages: Incrementalizing
  $\lambda$-calculi by static differentiation.
\newblock In {\em Proc. PLDI}, pages 145--155, 2014.

\bibitem{DBLP:conf/vldb/CeriW91}
Stefano Ceri and Jennifer Widom.
\newblock Deriving production rules for incremental view maintenance.
\newblock In {\em VLDB}, 1991.

\bibitem{CheneyLW14}
James Cheney, Sam Lindley, and Philip Wadler.
\newblock Query shredding: Efficient relational evaluation of queries over
  nested multisets.
\newblock In {\em Proc. SIGMOD}, pages 1027--1038, 2014.

\bibitem{Bussche}
Jan~Van den Bussche.
\newblock Simulation of the nested relational algebra by the flat relational
  algebra, with an application to the complexity of evaluating powerset algebra
  expressions.
\newblock {\em Theoretical Computer Science}, 254(1--2):363 -- 377, 2001.

\bibitem{BusscheGV07}
Jan~Van den Bussche, Dirk~Van Gucht, and Stijn Vansummeren.
\newblock Well-definedness and semantic type-checking for the nested relational
  calculus.
\newblock {\em Theor. Comput. Sci.}, 371(3):183--199, 2007.

\bibitem{BusscheV13}
Jan~Van den Bussche and Stijn Vansummeren.
\newblock Well-defined {NRC} queries can be typed - (extended abstract).
\newblock In {\em In Search of Elegance in the Theory and Practice of
  Computation - Essays Dedicated to Peter Buneman}, pages 494--506, 2013.

\bibitem{dimitrova03}
Katica Dimitrova, Maged El-Sayed, and ElkeA. Rundensteiner.
\newblock Order-sensitive view maintenance of materialized xquery views.
\newblock In {\em Conceptual Modeling - ER 2003}, volume 2813 of {\em Lecture
  Notes in Computer Science}, pages 144--157. 2003.

\bibitem{foster08}
J.~Nathan Foster, Ravi Konuru, J{\'{e}}r{\^{o}}me Sim{\'{e}}on, and Lionel
  Villard.
\newblock An algebraic approach to view maintenance for {XQ}uery.
\newblock In {\em {PLAN-X} 2008, Programming Language Technologies for XML}.

\bibitem{gluche97}
Dieter Gluche, Torsten Grust, Christof Mainberger, and MarcH. Scholl.
\newblock Incremental updates for materialized oql views.
\newblock In {\em Deductive and Object-Oriented Databases}, volume 1341 of {\em
  Lecture Notes in Computer Science}, pages 52--66. 1997.

\bibitem{GriffinL95}
Timothy Griffin and Leonid Libkin.
\newblock Incremental maintenance of views with duplicates.
\newblock In {\em Proc. SIGMOD}, pages 328--339, 1995.

\bibitem{Ferry}
Torsten Grust, Manuel Mayr, Jan Rittinger, and Tom Schreiber.
\newblock Ferry: Database-supported program execution.
\newblock In {\em Proc. Conference on Management of Data}, SIGMOD '09, pages
  1063--1066, 2009.

\bibitem{Grust2010}
Torsten Grust, Jan Rittinger, and Tom Schreiber.
\newblock Avalanche-safe linq compilation.
\newblock {\em Proc. VLDB Endow.}, 3(1-2):162--172, 2010.

\bibitem{DBLP:conf/deductive/GuptaKM92}
Ashish Gupta, Dinesh Katiyar, and Inderpal~Singh Mumick.
\newblock Counting solutions to the view maintenance problem.
\newblock In {\em Proc.\ Workshop on Deductive Databases, JICSLP}, 1992.

\bibitem{gupta-dred:93}
Ashish Gupta, Inderpal~Singh Mumick, and V.~S. Subrahmanian.
\newblock Maintaining views incrementally.
\newblock In {\em SIGMOD'93}, pages 157--166.

\bibitem{Joh90}
David~S. Johnson.
\newblock A catalog of complexity classes.
\newblock In Jan van Leeuwen, editor, {\em Handbook of Theoretical Computer
  Science}, volume~1, chapter~2, pages 67--161. Elsevier Science Publishers
  B.V., 1990.

\bibitem{kawaguchi97}
Akira Kawaguchi, Daniel Lieuwen, Inderpal Mumick, and Kenneth Ross.
\newblock Implementing incremental view maintenance in nested data models.
\newblock In {\em In Proceedings of the Workshop on Database Programming
  Languages}, pages 202--221, 1997.

\bibitem{Koch05}
Christoph Koch.
\newblock On the complexity of nonrecursive xquery and functional query
  languages on complex values.
\newblock In {\em Proc. PODS}, pages 84--97, 2005.

\bibitem{koch_ring}
Christoph Koch.
\newblock Incremental query evaluation in a ring of databases.
\newblock In {\em Proc. PODS}, pages 87--98, 2010.

\bibitem{ahmad14}
Christoph Koch, Yanif Ahmad, Oliver Kennedy, Milos Nikolic, Andres N{\"o}tzli,
  Daniel Lupei, and Amir Shaikhha.
\newblock Dbtoaster: higher-order delta processing for dynamic, frequently
  fresh views.
\newblock {\em VLDB J.}, 23(2):253--278, 2014.

\bibitem{LellahiT97}
S.~Kazem Lellahi and Val Tannen.
\newblock A calculus for collections and aggregates.
\newblock In {\em Category Theory and Computer Science, 7th International
  Conference, {CTCS} '97, Proceedings}, pages 261--280.

\bibitem{LevyS97}
Alon~Y. Levy and Dan Suciu.
\newblock Deciding containment for queries with complex objects.
\newblock In {\em Proc. PODS}, pages 20--31, 1997.

\bibitem{LibkinW97}
Leonid Libkin and Limsoon Wong.
\newblock Query languages for bags and aggregate functions.
\newblock {\em J. Comput. Syst. Sci.}, 55(2):241--272, 1997.

\bibitem{Liefke99}
Hartmut Liefke and Susan~B. Davidson.
\newblock Specifying updates in biomedical databases.
\newblock In {\em Proc. SSDBM}, 1999.

\bibitem{Links}
Sam Lindley and James Cheney.
\newblock Row-based effect types for database integration.
\newblock In {\em Proc. Workshop on Types in Language Design and
  Implementation}, TLDI '12, pages 91--102, 2012.

\bibitem{liu99}
Jixue Liu, Millist~W. Vincent, and Mukesh~K. Mohania.
\newblock Incremental evaluation of nest and unnest operators in nested
  relations.
\newblock In {\em Proc. of 1999 CODAS Conf}, pages 264--275, 1999.

\bibitem{LINQ}
Erik Meijer, Brian Beckman, and Gavin Bierman.
\newblock Linq: Reconciling object, relations and xml in the .net framework.
\newblock In {\em Proc. SIGMOD}, pages 706--706, 2006.

\bibitem{nakamura01}
Hiroaki Nakamura.
\newblock Incremental computation of complex object queries.
\newblock OOPSLA, pages 156--165, 2001.

\bibitem{Paige82}
Robert Paige and Shaye Koenig.
\newblock Finite differencing of computable expressions.
\newblock {\em ACM Trans. Program. Lang. Syst.}, 4(3):402--454, 1982.

\bibitem{Paredaens92}
Jan Paredaens and Dirk Van~Gucht.
\newblock Converting nested algebra expressions into flat algebra expressions.
\newblock {\em ACM Trans. Database Syst.}, 17(1):65--93, March 1992.

\bibitem{roussopoulos-tods:91mix}
Nick Roussopoulos.
\newblock An incremental access method for viewcache: Concept, algorithms, and
  cost analysis.
\newblock {\em ACM Transactions on Database Systems}, 16(3):535--563, 1991.

\bibitem{Suciu93}
Dan Suciu.
\newblock Bounded fixpoints for complex objects.
\newblock In {\em Database Programming Languages (DBPL-4), Proc. of the Fourth
  International Workshop on Database Programming Languages - Object Models and
  Languages}, pages 263--281, 1993.

\bibitem{ST1994}
Dan Suciu and Val Tannen.
\newblock {``A Query Language for NC''}.
\newblock In {\em Proc. PODS'94}, pages 167--178, 1994.

\bibitem{dimitra98}
Dimitra Vista.
\newblock Integration of incremental view maintenance into query optimizers.
\newblock In {\em Advances in Database Technology-EDBT'98}, volume 1377, pages
  374--388, 1998.

\bibitem{Zaharia10}
Matei Zaharia, Mosharaf Chowdhury, Michael~J. Franklin, Scott Shenker, and Ion
  Stoica.
\newblock Spark: Cluster computing with working sets.
\newblock HotCloud'10, 2010.

\bibitem{schwentick14}
Thomas Zeume and Thomas Schwentick.
\newblock Dynamic conjunctive queries.
\newblock In {\em Proc. ICDT}, pages 38--49, 2014.

\end{thebibliography}
\bibliographystyle{plain}
\end{footnotesize}

\newpage

\onecolumn
\appendix
\section{Delta processing}

\subsection{The flat relational case}
\label{app:flat-delta}

We recall how delta processing works for queries expressed in
the \emph{positive relational algebra}. 
Delta rules were originally defined
for datalog programs~\cite{DBLP:conf/deductive/GuptaKM92,gupta-dred:93} 
but they are even more natural for algebraic query languages such as
the relational algebra on bags~\cite{GriffinL95,koch_ring}, 
simply because the algebraic structure of a group is
the necessary and sufficient environment in which deltas live.

Consider relational algebra expressions
built from table names $R_1,\ldots,R_n$ from some schema and the operators for
selection, projection, Cartesian product, and union, where we denote the 
last one by $\uplus$ to remind us that we assume bag semantics in this paper.

The delta rules constitute an inductive definition of a
transformation that maps every algebra expression
$e$ over table names
$R_i$ 
into another algebra expression 
$\inc(e)$ over table names
$R_i$ and $\Delta R_i, i=1..n$.
The names of the form $\Delta R_i$ designate an update: 
tables that contain tuples  to be added to those in $R_i$ 
(for the moment we focus only on insertions).
The rules are:
\vspace*{-2mm}
\begin{align*}
\inc(R_i)             &= \Delta R_i~~i=1..n \qqquad
&
\inc(\sigma_{p} e)    &= \sigma_{p} \inc(e)
\\
\inc(e_1 \uplus e_2)  &= \inc(e_1) \uplus \inc(e_2)
&
\inc(\Pi_{\bar{i}} e) &= \Pi_{\bar{i}} \inc(e) 
\\
\inc(e_1\times e_2)   &
\multicolumn{4}{l}{$\ =
                    \inc(e_1)\times e_2 
                   ~\uplus~ e_1\times \inc(e_2)
                   ~\uplus~ \inc(e_1) \times \inc(e_2) $}
\end{align*}
We remark that the rule for join is the same as the one for Cartesian product.

The delta rules satisfy the following property\footnote{
This is due to the
commutativity and associativity of bag union as well as the 
distributivity of selection, projection and Cartesian product,
over bag union.},
which also suggests how the incremental computation proceeds:
\vspace*{-2mm}
\begin{align}
&
e[R_1\uplus\Delta R_1,\cdots,R_n\uplus\Delta R_n]~=~ 
e[R_1,\cdots,R_n] ~\uplus~   
\inc(e)[R_1,\cdots,R_n,\Delta R_1,\cdots,\Delta R_n]
\label{prop:alg-delta}
\end{align}
In the statement above we abuse, as usual, the notation by using
the $R_i$'s for both table names and corresponding table
instances and we denote by $e[\oR]$ the table that results from
evaluating the algebra expression $e$ on a database $\oR$, where
$\oR$ stands for $R_1,\ldots,R_n$.
Equation~(\ref{prop:alg-delta}) captures the incremental maintenance
of the query result. Given updates $\oDR$ to the database,
we just compute $\inc(e)[\oR,\oDR]$ and use it
to update the previously materialized answer $e[\oR]$
.

\begin{example}
For a concrete example of incrementalizing a relational algebra query,
we consider a bag of movies $\movs(movie,genre)$, 
a bag containing their showtimes
$\shows(movie,loc,time)$ and the query $\exQ$ returning all the dramas 
playing in Oz:
\begin{equation}
\nonumber
\exQ ~~\equiv~~
\Pi_{\mathit{movie}}(\sigma_{\mathit{loc} = \mathrm{Oz}} \shows
\bowtie
\sigma_{\mathit{genre} = \mathrm{Drama}} \movs
).
\end{equation}
Now suppose that the updates $\Delta\shows$ and $\Delta\movs$
are applied to $\shows$ and $\movs,$ respectively. 
By Equation~\ref{prop:alg-delta} and the delta rules,
the updated $\exQ$ can be computed by $\uplus$-ing
\begin{align*}
&
\Pi_{\mathit{movie}}(
     \sigma_{\mathit{loc} = \mathrm{Oz}} \Delta\shows
     \bowtie
     \sigma_{\mathit{genre} = \mathrm{Drama}} \movs
     \ \uplus\ 
     \sigma_{\mathit{loc} = \mathrm{Oz}} \shows
     \bowtie
     \sigma_{\mathit{genre} = \mathrm{Drama}} \Delta\movs
     \ \uplus\  
     \sigma_{\mathit{loc} = \mathrm{Oz}} \Delta\shows
     \bowtie
     \sigma_{\mathit{genre} = \mathrm{Drama}} \Delta\movs)
\end{align*}
to the previously materialized answer to $\exQ$.
If $\Delta\shows$ and $\Delta\movs$ are much smaller than $\shows$,
respectively $\movs$, this is typically
computationally much cheaper than recomputing the query after updating
the base tables: this is what makes incremental view maintenance worthwhile.
\end{example}

Under reasonable assumptions about the cost of query evaluation algorithms 
and considering small updates compared to the size of the database,
this is better than recomputing the query on the updated database
$e[\overline{R \uplus \Delta R}]$.
For instance,
a query $R \bowtie S$ can have size (and evaluation cost) quadratic in the input 
database.
Assuming $\Delta R$ and $\Delta S$ consist of a constant number of tuples, 
incrementally maintaining the query via
$\delta(R \bowtie S) = 
(\Delta R) \bowtie S
~\uplus~R \bowtie (\Delta S)
~\uplus~(\Delta R) \bowtie (\Delta S)$ has linear size and cost, 
while recomputing it
(as $(R \uplus \Delta R) \bowtie (S \uplus \Delta S)$) has quadratic cost.

As shown by Gupta et al.~\cite{DBLP:conf/deductive/GuptaKM92}, 
the same delta rules
can also be used to propagate deletions if we extend the bag semantics
to allow for \emph{negative} multiplicities: the table 
$\Delta R_i$ associates negative multiplicities to 
the tuples to be deleted from $R_i$.

\subsection{Challenges for efficient incrementalization}
\label{app:challenges}

In the following we discuss the challenges in deriving
a delta query  which is cheaper than full re-evaluation 
for any expression in a language.

Informally, we say that the delta $\inc(e)[R, \Delta R]$ of a query $e[R]$ 
is more {\em efficient} than full recomputation 
(or simply {\em efficient}), if for any update $\Delta R$ s.t.\
$\size(\Delta R) \ll \size(R)$, 
evaluating $\inc(e)[R,\Delta R]$ and applying
it to the output of $e$ is less expensive than re-evaluating $e$ from scratch, 
i.e.:  
\begin{align*}
&
\pcost(\inc(e)[R,\Delta R]) \ll \pcost(e[R \uplus \Delta R]) 
\quad \text{ and }
\\
&
\size(\inc(e)[R,\Delta R]) \ll \size(e[R \uplus \Delta R]),
\end{align*}
where the second equation ensures that applying the update is cheaper than
re-computation 
considering that the cost of applying an update is proportional to its size
and that the cost of evaluating an expression is lowerbounded by the size of
its output ($\size(e[R \uplus \Delta R]) \le \pcost(e[R \uplus \Delta R])$). 
 
One can guarantee that the delta of any expression in a language is  
efficient by requiring that every construct $\pp(e)[R]$ of the language
satisfies the property above, 
i.e. $\size(\inc(e)[R,\Delta R]) \ll \size(e[R])$ implies:
\begin{align}
&
\pcost(\inc(\pp(e))[R,\Delta R]) \ll \pcost(\pp(e)[R \uplus \Delta R]) 
\ \text{ and }
\nonumber
\\
&
\size(\inc(\pp(e))[R,\Delta R]) \ll \size(\pp(e)[R \uplus \Delta R])
\label{eq:delta_eff}
\end{align}

Unfortunately, this property does not hold for constructs $\pp(e)[R]$ which 
take linear time in their inputs $e[R]$
(i.e. $\pcost(\pp(e)[R]) = \size(e[R])$ ) 
and whose delta $\inc(\pp(e))[R,\Delta R]$ depends on the original
input $e[R]$ (therefore $\pcost(e[R]) < \pcost(\inc(\pp(e))[R,\Delta R])$), 
as it leads to the following contradiction:
\begin{align*}
&
\size(e[R]) \le \pcost(e[R]) < \pcost(\inc(\pp(e))[R,\Delta R]) \ll 
\\
&
\ll \pcost(\pp(e)[R \uplus \Delta R]) = \size(e[R \uplus \Delta R])
\approx \size(e[R]),
\end{align*}
where the last approximation follows from the fact that:
\begin{align*}
&
e[R \uplus \Delta R] = 
e[R]\ \uplus\ 
\inc(e)[R,\Delta R]
\quad\text{ and }
\\
&
\size(\inc(e)[R,\Delta R]) \ll \size(e[R]).
\end{align*} 
 
An example of such a construct is bag subtraction 
$(e_1 \setminus e_2)[R]$, that associates to every element $v_i$ in $e_1[R]$
the multiplicity $\max(0, m_1 - m_2)$, where $m_1, m_2$ are $v_i$'s
multiplicities in $e_1[R]$ and $e_2[R]$, respectively.
Indeed, the cost of evaluating bag subtraction is proportional to its
input (i.e. $\pcost(e_1 \setminus e_2)[R] = \size(e_1[R])$,
assuming $e_1[R]$ and $e_2[R]$ have similar sizes)
and the result of $(e_1 \setminus e_2)[R]$ can be maintained when 
$e_2[R]$ changes, only
if the initial value of $e_1[R]$ is known at the time of the update.
The singleton constructor or the
emptiness test over bags also exhibit similar characteristics.   
By contrast, constructs that take time linear in their input, but whose
delta rule depends only on the update do not present this issue (eg. $\flt$).

This problem can be addressed by materializing the result of the
subquery $e[R]$,
such that one does not need to pay its cost again 
when evaluating $\inc(\pp(e))[R,\Delta R]$.
However, this only solves half of the problem, as we also need to make
sure that the outcome of $\inc(\pp(e))[R,\Delta R]$ 
can be efficiently propagated through outer queries $e'$ that may use 
$\pp(e)[R \uplus \Delta R]$ as a subquery.  
Solving this issue requires handcrafted solutions that take into
consideration the particularities of $\pp$ and the ways it can be used.
For example, in our solution for efficiently incrementalizing $\sng(\cdot)$
we take advantage of the fact that the only way of accessing the contents of a
inner bag is via $\flt(\cdot)$.

Finally, 
for constructs $\pp$ with boolean as output domain 
(eg. testing whether a bag is empty), it no longer makes sense to
distinguish between small and large values, and therefore,
the condition (\ref{eq:delta_eff}) can never be satisfied.
This problem extends to a class of primitives that includes 
bag equality, negation, and membership testing, and restricts our
solution for efficient incrementalization to only the positive fragment of 
nested relational calculus $\pNRC$.

\section{Incrementalizing $\ipNRC$}

\subsection{The delta transformation}
\label{sec:app:delta}

\deltathipNRA*
\begin{proof}
The proof follows by structural induction on $h$ and from the
semantics of $\ipNRC$ constructs.

\begin{itemize}
\item
For $h = R$, the result follows immediately.

\item For $h \in \set{\uZ, p, \sngvar{x}, \piTC{x}{i}, \etpl, \sngr(e)}$ 
as the query does not depend on the input bag $R$
we have $h[R \uplus \Delta R] = h[R]$
 and the result follows
immediately.

\item
For $h = \for{x}{e_1} \collects e_2$:
\begin{align*}
&
\meane{ (\for{x}{e_1} \collects e_2)[R \uplus \Delta R] }{\Irho;\Vrho}
=
\\
&
= {\biguplus}_{ v \in \meane{e_1[R \uplus \Delta R]}{\Irho;\Vrho}} 
     \meane{e_2[R \uplus \Delta R]}{\Irho;\Vrho[x\coloneq v]}
\\
&
= {\biguplus}_{ v \in \meane{e_1[R]}{\Irho;\Vrho} \uplus
                      \meane{\inc(e_1)[R, \Delta R] }{\Irho;\Vrho} } 
     \meane{e_2[R \uplus \Delta R]}{\Irho;\Vrho[x\coloneq v]}
\\
&
= [ {\biguplus}_{ v \in \meane{e_1[R]}{\Irho;\Vrho} } 
     \meane{e_2[R \uplus \Delta R]}{\Irho;\Vrho[x\coloneq v]} ]
  \ \uplus\ 
  [ {\biguplus}_{ v \in \meane{\inc(e_1)[R, \Delta R] }{\Irho;\Vrho} } 
     \meane{e_2[R \uplus \Delta R]}{\Irho;\Vrho[x\coloneq v]} ]
\\
&
= [ {\biguplus}_{ v \in \meane{e_1[R]}{\Irho;\Vrho} } 
     \meane{e_2[R]}{\Irho;\Vrho[x\coloneq v]} \uplus
     \meane{\inc(e_2)[R, \Delta R]}{\Irho;\Vrho[x\coloneq v]} ]
  \ \uplus\ 
  [ {\biguplus}_{ v \in \meane{\inc(e_1)[R, \Delta R] }{\Irho;\Vrho} } 
     \meane{e_2[R]}{\Irho;\Vrho[x\coloneq v]} \uplus
     \meane{\inc(e_2)[R, \Delta R]}{\Irho;\Vrho[x\coloneq v]} ]
\\
&
= [ {\biguplus}_{ v \in \meane{e_1[R]}{\Irho;\Vrho} } 
     \meane{e_2[R]}{\Irho;\Vrho[x\coloneq v]} ]
  \ \uplus\ 
  [ {\biguplus}_{ v \in \meane{e_1[R]}{\Irho;\Vrho} } 
     \meane{\inc(e_2)[R, \Delta R]}{\Irho;\Vrho[x\coloneq v]} ]
  \quad\uplus
\\
&
\quad
  [ {\biguplus}_{ v \in \meane{\inc(e_1)[R, \Delta R] }{\Irho;\Vrho} } 
     \meane{e_2[R]}{\Irho;\Vrho[x\coloneq v]} ]
  \ \uplus\ 
  [ {\biguplus}_{ v \in \meane{\inc(e_1)[R, \Delta R] }{\Irho;\Vrho} } 
     \meane{\inc(e_2)[R, \Delta R]}{\Irho;\Vrho[x\coloneq v]} ]
\\
&
= \meane{(\for{x}{e_1} \collects e_2)[R]}{\Irho;\Vrho}     
  \ \uplus\ 
  \meane{(\for{x}{e_1} \collects \inc(e_2))[R, \Delta R] }{\Irho;\Vrho} 
  \quad\uplus\ 
\\
&
\quad
  \meane{(\for{x}{\inc(e_1)} \collects e_2)[R, \Delta R] }{\Irho;\Vrho} 
  \ \uplus\ 
  \meane{(\for{x}{\inc(e_1)} \collects \inc(e_2))[R, \Delta R] }{\Irho;\Vrho} 
\\
&
= \meane{(\for{x}{e_1} \collects e_2)[R]}{\Irho;\Vrho} \ \uplus\ 
  \meane{\inc(\for{x}{e_1} \collects e_2)[R, \Delta R] }{\Irho;\Vrho}
\\
& 
= \meane{(\for{x}{e_1} \collects e_2)[R] \uplus 
\inc(\for{x}{e_1} \collects e_2)[R, \Delta R] }{\Irho;\Vrho} 
\end{align*}

\item
For $h = e_1 \cprod e_2$ the reasoning is similar to the case of 
$h = \for{x}{e_1} \collects e_2$.

\item
For $h = e_1 \uplus e_2$ the result follows from the associativity and 
commutativity of $\uplus$.

\item
For $h = \ominus(e)$ the result follows from the associativity and
commutativity of $\uplus$
and the fact that $\ominus$ is the inverse operation wrt. $\uplus$.

\item For $h = \flt(e)$:
\begin{align*}
&
\meane{\flt(e)[R \uplus \Delta R]}{}
=
{\biguplus}_{v \in \meane{e[R \uplus \Delta R]}{} } v
=
{\biguplus}_{v \in \meane{e[R]}{} \uplus 
                   \meane{\inc(e)[R,\Delta R]}{} } v
=
{\biguplus}_{v \in \meane{e[R]}{} } v \quad\uplus\quad 
{\biguplus}_{v \in \meane{\inc(e)[R,\Delta R]}{} } v
=
\\
&
=
\meane{\flt(e)[R]}{} \uplus \meane{\flt(\inc(e))[R, \Delta R]}{}
=
\meane{\flt(e)[R] \uplus \flt(\inc(e))[R, \Delta R]}{}
=
\\
&
=
\meane{\flt(e)[R] \uplus \inc(\flt(e))[R, \Delta R]}{}
\end{align*}

\item For $h = \llet{X}{e_1}\ e_2$
\begin{align*}
&
\meane{(\llet{X}{e_1}\ e_2)[R \uplus \Delta R]}{\Irho;\Vrho}
=
\meane{e_2[R \uplus \Delta R,X]}
{\Irho;\Vrho[X\coloneq\meane{e_1[R \uplus \Delta R]}{\Irho;\Vrho}]} =
\\
&
=
\meane{e_2[R,X]}
{\Irho;\Vrho[X\coloneq\meane{e_1[R \uplus \Delta R]}{\Irho;\Vrho}]}
\uplus
\meane{\incR{R}(e_2)[R, X, \Delta R]}
{\Irho;\Vrho[X\coloneq\meane{e_1[R \uplus \Delta R]}{\Irho;\Vrho}]}
\\
&
=
\meane{e_2[R,X]}
{\Irho;\Vrho[X\coloneq\meane{e_1[R]}{\Irho;\Vrho} \uplus
                      \meane{\inc(e_1)[R, \Delta R]}{\Irho;\Vrho}]}
\uplus
\meane{\incR{R}(e_2)[R, X, \Delta R]}
{\Irho;\Vrho[X\coloneq\meane{e_1[R]}{\Irho;\Vrho} \uplus
                      \meane{\inc(e_1)[R, \Delta R]}{\Irho;\Vrho}]}
\\
&
=
\meane{e_2[R,X \uplus \Delta X]}
{\Irho;\Vrho[X\coloneq\meane{e_1[R]}{\Irho;\Vrho},
             \Delta X \coloneq \meane{\inc(e_1)[R, \Delta R]}{\Irho;\Vrho}]}
\uplus
\meane{\incR{R}(e_2)[R, X \uplus \Delta X, \Delta R]}
{\Irho;\Vrho[X\coloneq\meane{e_1[R]}{\Irho;\Vrho},
             \Delta X \coloneq \meane{\inc(e_1)[R, \Delta R]}{\Irho;\Vrho}]}
\\
&
=
\meane{e_2[R,X \uplus \Delta X] \uplus
       \incR{R}(e_2)[R, X \uplus \Delta X, \Delta R]
}
{\Irho;\Vrho[X\coloneq\meane{e_1[R]}{\Irho;\Vrho},
             \Delta X \coloneq \meane{\inc(e_1)[R, \Delta R]}{\Irho;\Vrho}]}
\\
&
=
\meane{e_2[R,X] \uplus
       \incR{X}(e_2)[R, X, \Delta X] \uplus
       \incR{R}(e_2)[R, X, \Delta R] \uplus
       \incR{X}(\incR{R}(e_2))[R, X, \Delta X, \Delta R]
}
{\Irho;\Vrho[X\coloneq\meane{e_1[R]}{\Irho;\Vrho},
             \Delta X \coloneq \meane{\inc(e_1)[R, \Delta R]}{\Irho;\Vrho}]}
\\
&
= \lletBin{X}{e_1[R]}{\Delta X}{\inc(e_1)[R, \Delta R]}
\\
&\quad
  (e_2[R,X] \uplus
       \incR{X}(e_2)[R, X, \Delta X] \uplus
       \incR{R}(e_2)[R, X, \Delta R] \uplus
       \incR{X}(\incR{R}(e_2))[R, X, \Delta X, \Delta R])
\end{align*}

\end{itemize}
\end{proof}

\deltalemipNRAconst*
\begin{proof}
We do a case by case analysis on $h.$
\begin{itemize}
\item 
For $h \in \set{\uZ, p, \etpl, \sngvar{x}, \piTC{x}{i}, \sngr(e)}$ 
we have from the definition of
$\inc(\cdot)$ that $\inc(h) = \uZ.$

\item 
For $h = \forx{e_1}{e_2},$ we have by the induction hypothesis that
$\inc(e_1) = \uZ,$
$\inc(e_2) = \uZ,$
therefore 
$\inc(\forx{e_1}{e_2}) 
= (\forx{\uZ}{e_2}) \uplus (\forx{e_1}{\uZ}) \uplus (\forx{\uZ}{\uZ})
= \uZ.$

\item 
For $h = e_1 \cprod e_2$ the reasoning is similar to the case of 
$h = \for{x}{e_1} \collects e_2$.

\item 
For $h = e_1 \uplus e_2,$ we have by the induction hypothesis that
$\inc(e_1) = \uZ,$
$\inc(e_2) = \uZ,$
therefore
$\inc(e_1 \uplus e_2) 
= \uZ \uplus \uZ
= \uZ.$

\item
For $h = \ominus(e),$ we have by the induction hypothesis that
$\inc(e) = \uZ,$
therefore
$\inc(\ominus(e)) = \ominus(\uZ) = \uZ.$

\item
For $h = \flt(e),$ we have by the induction hypothesis that
$\inc(e) = \uZ,$
therefore
$\inc(\flt(e)) = \flt(\uZ) = \uZ.$

\item 
For $h = \llet{X}{e_1}\ e_2,$ we have by the induction hypothesis that
$\incR{R}(e_2) = \uZ,$
$\Delta X = \incR{R}(e_1) = \uZ,$
and the result follows from the fact the $\incR{X}(e_2)[X,\uZ] = \uZ$.

\end{itemize}
\end{proof}

\subsection{Higher-order delta derivation}
\label{sec:app:higher-delta}

\PropDecreasingDegree*
\begin{proof}
The proof follows via structural induction on $h$ and from the definition
of $\inc(\cdot)$ and $\dg(\cdot)$.
For subexpressions of $h$ which are {\em input-independent} we use the fact that
$\inc(e) = \uZ$ and $\dg(e) = \dg(\inc(e)) = 0.$

\begin{itemize}
  \item For $h = R$ we have:
  $\dg(\inc(R)) = \dg(\Delta R) = 0 = 1-1 = \dg(R)-1$
  
  \item For $h = \forx{e_1}{e_2}$ we have the following cases:
  \\
  Case 1: $\dg(\inc(e_1)) = \dg(e_1) -1$
          and $g$ is {\em input-independent}:  
  \begin{align*}
  &
  \dg(\inc(\forx{e_1}{e_2})) 
  = \dg(\forx{\inc(e_1)}{e_2}) 
  = \dg(e_2) + \dg(\inc(e_1)) 
  = \dg(e_2) + \dg(e_1) -1
  = \dg(\forx{e_1}{e_2}) -1.
  \end{align*}
  
  Case 2: $\dg(\inc(e_2)) = \dg(e_2) -1$ 
          and $f$ is {\em input-independent}:
          Analogous to Case 1.

  Case 3: $\dg(\inc(e_2)) = \dg(e_2) -1$ and 
          $\dg(\inc(e_1)) = \dg(e_1) -1$:
  \begin{align*}
  &
   \dg(\inc(\forx{e_1}{e_2})) 
   = \dg((\forx{\inc(e_1)}{e_2}) \uplus 
         (\forx{e_1}{\inc(e_2)}) \uplus 
         (\forx{\inc(e_1)}{\inc(e_2)}))
   \\
   & \qqquad\qquad
   = \max( \dg(\forx{\inc(e_1)}{e_2}), 
           \dg(\forx{e_1}{\inc(e_2)}), 
           \dg(\forx{\inc(e_1)}{\inc(e_2)}) )
   \\
   & \qqquad\qquad
   = \max( \dg(e_2) + \dg(\inc(e_1)), \dg(\inc(e_2)) + \dg(e_1), 
           \dg(\inc(e_2)) + \dg(\inc(e_1)) )
   \\
   & \qqquad\qquad
   = \max( \dg(e_2) + \dg(e_1)-1, \dg(e_2)-1 + \dg(e_1), 
           \dg(e_2)-1 + \dg(e_1)-1 )
   \\
   & \qqquad\qquad
   = \dg(e_2) + \dg(e_1) -1
   = \dg(\forx{e_1}{e_2}) -1.
  \end{align*}

  \item For $h = e_1 \cprod e_2$ the proof is similar to the one for 
        $\forx{e_1}{e_2}$ as the definitions of $\inc(h)$ and $\dg(h)$ are 
        similar.

  \item For $h = e_1 \uplus e_2$ we have the following cases:
  \\
  Case 1: $\dg(\inc(e_1)) = \dg(e_1) -1$
          and $e_2$ is {\em input-independent}:
  \begin{align*}
  &
  \dg(\inc(e_1 \uplus e_2)) 
  = \max(\dg( \inc(e_1) ), 0)
  = \dg( \inc(e_1) )
  = \dg( e_1 ) -1
  = \max(\dg( e_1 ), 0) -1
  = \dg(e_1 \uplus e_2) -1. 
  \end{align*}

  Case 2: $\dg(\inc(e_2)) = \dg(e_2) -1$
          and $e_1$ is {\em input-independent}:
          Analogous to Case 1.
 
  Case 3: $\dg(\inc(e_1)) = \dg(e_1) -1$ and
          $\dg(\inc(e_2)) = \dg(e_2) -1$:
  \begin{align*}
  &
  \dg(\inc(e_1 \uplus e_2)) 
  = \dg(\inc(e_1) \uplus \inc(e_2)) 
  = \max(\dg( \inc(e_1) ), \dg( \inc(e_2) ))
  = \max(\dg( e_1 )-1, \dg( e_2 )-1)
  \\
  & \qqquad\qquad\;\;\;\,
  = \max(\dg( e_1 ), \dg( e_2 ))-1
  = \dg(e_1 \uplus e_2) -1. 
  \end{align*}

  \item For $h = \ominus(e)$ we have that
  $\dg(\inc(e)) = \dg(e)-1$, therefore
  $\dg(\inc(\ominus(e))) 
   = \dg(\ominus(\inc(e)))
   = \dg(\inc(e))
   = \dg(e)-1
   = \dg(\ominus(e))-1.$

  \item For $h = \flt(e)$ the proof is similar to the one for 
        $\ominus(e)$ as the definitions of $\inc(h)$ and $\dg(h)$ are 
        similar.
\end{itemize}
\end{proof}

\subsection{The cost transformation}
\label{sec:app:cost}

\begin{lemma}
\label{lem:ipnra-cost-mon}
For any $\ipNRC$ expression $\IGamma;\VGamma,x:C \vdash h:\Bag(A)$, the cost
interpretation $\costt{h}{}$ is monotonic, i.e.\ $\forall c_1, c_2 \in \deg{C}$ s.t.\ 
$c_1 \preceq c_2$ then 
$\costc{h}{}{\Irhoc;\Vrhoc[x\coloneq c_1]} \preceq 
 \costc{h}{}{\Irhoc;\Vrhoc[x\coloneq c_2]}$.
\end{lemma}
\begin{proof}
The result follows via structural induction on $h$ and from the 
fact that the cost functions of the $\ipNRC$ constructs are themselves
monotonic.

We do a case by case analysis on $h:$
\begin{itemize}
\item For $h \in \set{R, p, \uZ, \etpl}$ the result follows from the fact that
$\forall c_1,c_2.\ 
 \costc{h}{}{\Irhoc;\Vrhoc[x\coloneq c_1]} =
 \costc{h}{}{\Irhoc;\Vrhoc[x\coloneq c_2]}$.

\item For $h = \sngvar{x}:
\costc{\sngvar{x}}{}{\Irhoc;\Vrhoc[x\coloneq c_1]} 
= \set{c_1} \preceq \set{c_2} 
= \costc{\sngvar{x}}{}{\Irhoc;\Vrhoc[x\coloneq c_2]}
$

\item For $h = \piTC{i}{x}:
  \costc{\piTC{i}{x}}{}{\Irhoc;\Vrhoc[x\coloneq \tuple{c_{11},c_{12}}]} 
= \set{c_{1i}} \preceq \set{c_{2i}} 
= \costc{\piTC{i}{x}}{}{\Irhoc;\Vrhoc[x\coloneq \tuple{c_{21},c_{22}}]}
$

\item For $h = \fory{e_1}{e_2},$ we have from the induction hypothesis that:
\begin{align*}
&
\costcx{e_1}{i}{c_1} \preceq 
\costcx{e_1}{i}{c_2}
&
 \costcy{e_2}{i}{\costcx{e_1}{i}{c_1} }
 \preceq 
 \costcy{e_2}{i}{\costcx{e_1}{i}{c_2} ]}
\\
&
\costcx{e_1}{o}{c_1} \le 
 \costcx{e_1}{o}{c_2}
&
 \costcy{e_2}{o}{\costcx{e_1}{i}{c_1} ]} 
 \le 
 \costcy{e_2}{o}{\costcx{e_1}{i}{c_2} ]},
\end{align*}
therefore:
\begin{align*}
&
\costcx{\forx{e_1}{e_2}}{}{c_1}
=
\bagS{ \costcy{e_2}{i}{ \costcx{e_1}{i}{c_1} } }
{ \costcy{e_2}{o}{ \costcx{e_1}{i}{c_1} } \cdot
  \costcx{e_1}{o}{c_1} }
\\
&
\preceq
\bagS{ \costcy{e_2}{i}{ \costcx{e_1}{i}{c_2} } }
{ \costcy{e_2}{o}{ \costcx{e_1}{i}{c_2} } \cdot
  \costcx{e_1}{o}{c_2} } 
\\
&
=
\costcx{\forx{e_1}{e_2}}{}{c_2}
\end{align*}

\item For $h = e_1 \cprod e_2,$ we have from the induction hypothesis that
\begin{align*}
&
\costcx{e_1}{i}{c_1} \preceq \costcx{e_1}{i}{c_2}
&
\costcx{e_2}{i}{c_1} \preceq \costcx{e_2}{i}{c_2}
\\
&
\costcx{e_1}{o}{c_1} \le \costcx{e_1}{o}{c_2}
&
\costcx{e_2}{o}{c_1} \le \costcx{e_2}{o}{c_2}, 
\end{align*}
therefore:
\begin{align*}
&
\costcx{e_1 \cprod e_2}{}{c_1}
=
\bagS{\tuple{\costcx{e_1}{i}{c_1},\costcx{e_2}{i}{c_1}}}
    {\costcx{e_1}{o}{c_1} \cdot \costcx{e_2}{o}{c_1}}
\\
&
\qqquad\qquad\quad
\preceq
\bagS{\tuple{\costcx{e_1}{i}{c_2},\costcx{e_2}{i}{c_2}}}
    {\costcx{e_1}{o}{c_2} \cdot \costcx{e_2}{o}{c_2}}
=
\costcx{e_1 \cprod e_2}{}{c_2}
\end{align*}

\item For $h = e_1 \uplus e_2,$ we have from the induction hypothesis that
$\costcx{e_1}{}{c_1} \preceq \costcx{e_1}{}{c_2}$ and
$\costcx{e_2}{}{c_1} \preceq \costcx{e_2}{}{c_2},$ therefore:
\begin{align*}
&
\costcx{e_1 \uplus e_2}{}{c_1}
= \sup( \costcx{e_1}{}{c_1}, \costcx{e_2}{}{c_1} )
\\
&
\quad
\preceq \sup( \costcx{e_1}{}{c_2}, \costcx{e_2}{}{c_2} )
= \costcx{e_1 \uplus e_2}{}{c_2}
\end{align*}

\item For $h = \ominus(e),$ we have from the induction hypothesis that
$\costcx{e}{}{c_1} \preceq \costcx{e}{}{c_2},$ therefore:
\begin{align*}
&
\costcx{ \ominus(e) }{}{c_1}
= \costcx{e}{}{c_1} \preceq \costcx{e}{}{c_2}
= \costcx{\ominus(e)}{}{c_2}
\end{align*}

\item For $h = \flt(e),$ we have from the induction hypothesis that
$\costcx{e}{o}{c_1} \le \costcx{e}{o}{c_2},$
$\costcx{e}{io}{c_1} \le \costcx{f}{io}{c_2}$ and 
$\costcx{e}{ii}{c_1} \preceq \costcx{e}{ii}{c_2},$ therefore:
\begin{align*}
&
  \cost{ \flt(e) }{}{c_1}
= \bagS{ \costcx{e}{ii}{c_1} }{\costcx{e}{o}{c_1} \cdot \costcx{e}{io}{c_1}}
\\
&
\quad
\preceq 
  \bagS{ \costcx{e}{ii}{c_2} }{\costcx{e}{o}{c_2} \cdot \costcx{e}{io}{c_2}}
= \cost{ \flt(e) }{}{c_2}
\end{align*}

\item For $h = \sngr(e),$ we have from the induction hypothesis that
$\costcx{e}{}{c_1} \preceq \costcx{f}{}{c_2},$ therefore:
\[
\costcx{ \sngr(e) }{}{c_1}
= \bagS{\costcx{e}{}{c_1}}{} \preceq \bagS{\costcx{e}{}{c_2}}{}
= \costcx{\sngr(e)}{}{c_2}
\]
\end{itemize}
\end{proof}

\costthipNRA*

\begin{proof}
The proof follows via structural induction on $h$ and from the cost
semantics of $\ipNRC$ constructs
as well as the monotonicity of $\tcost(\cdot)$.

\begin{itemize}
  \item For $h = R$ we have:
  $\costt{\inc(R)}{} = \costt{\Delta R}{} = \sz(\Delta R) 
   \prec \sz(R) = \costt{R}{}$
  
  \item For $h = \forx{e_1}{e_2}$ we need to show that:
  \begin{align*}
  &
  \costt{\inc(\forx{e_1}{e_2})}{}
  = \costt{(\forx{\inc(e_1)}{e_2}) \uplus 
           (\forx{e_1}{\inc(e_2)}) \uplus 
           (\forx{\inc(e_1)}{\inc(e_2)}) }{}
  \\
  & 
  = \sup( \costt{\forx{\inc(e_1)}{e_2}}{}, 
          \costt{\forx{e_1}{\inc(e_2)}}{},
          \costt{\forx{\inc(e_1)}{\inc(e_2)}}{} )
  \\
  &
  \prec \costc{\forx{e_1}{e_2}}{}{}
  \end{align*}
 
  Case 1:  $\costt{\inc(e_1)}{} \prec \costt{e_1}{}$
           and $e_2$ is {\em input-independent}, therefore 
           $\inc(e_2) = \uZ$ (from Lemma~\ref{lem:delta-const}):  
  \begin{align*}
  &
  \costt{\inc(\forx{e_1}{e_2})}{} = \costt{\forx{\inc(e_1)}{e_2}}{}
  \\
  &
  = \bagS{ \costcx{e_2}{i}{ \costt{\inc(e_1)}{i} } }
         { \costcx{e_2}{o}{ \costt{\inc(e_1)}{i} } \cdot 
           \costt{\inc(e_1)}{o} }
  \\
  &
  \prec 
  \bagS{ \costcx{e_2}{i}{ \costt{e_1}{i} } }
       { \costcx{e_2}{o}{ \costt{e_1}{i} } \cdot 
         \costt{e_1}{o} }
  = \costt{\forx{e_1}{e_2}}{},
  \end{align*}
  where we used the fact that 
  $\costt{\inc(e_1)}{o} < \costt{e_1}{o}$ and
  $ \costcx{e_2}{}{ \costt{\inc(e_1)}{i} } \preceq 
    \costcx{e_2}{}{ \costt{e_1}{i} }$ (from Lemma~\ref{lem:ipnra-cost-mon}).

  Case 2: $\costt{\inc(e_2)}{} \prec \costt{e_2}{}$ 
          and $e_1$ is {\em input-independent}, therefore 
          $\inc(e_1) = \uZ$ (from Lemma~\ref{lem:delta-const}):
  \begin{align*}
  &
  \costt{\inc(\forx{e_1}{e_2})}{} = \costt{\forx{e_1}{\inc(e_2)}}{} 
  \\
  &
  = \bagS{ \costcx{\inc(e_2)}{i}{ \costt{e_1}{i} } }
         { \costcx{\inc(e_2)}{o}{ \costt{e_1}{i} } \cdot 
           \costt{e_1}{o} }
  \\
  &
  \prec 
  \bagS{ \costcx{e_2}{i}{ \costt{e_1}{i} } }
       { \costcx{e_2}{o}{ \costt{e_1}{i} } \cdot 
         \costt{e_1}{o} } 
  =
  \costt{\forx{e_1}{e_2}}{}, 
  \end{align*}
  where we used the fact that 
  $\costcx{\inc(e_2)}{i}{ \costt{e_1}{i} } \preceq 
   \costcx{e_2}{i}{ \costt{e_1}{i} }$
  and
  $\costcx{\inc(e_2)}{o}{ \costt{e_1}{i} } < 
   \costcx{e_2}{o}{ \costt{e_1}{i} }$.

  Case 3: $\costt{\inc(e_2)}{} \prec \costt{e_2}{}$ and 
          $\costt{\inc(e_1)}{} \prec \costt{e_1}{}.$ 
          We show that each term of            
          the $\sup$ function is less than the cost of the original function:
  \begin{align*}
  &
  \costt{\forx{\inc(e_1)}{e_2}}{}
  \prec \costt{\forx{e_1}{e_2}}{}, \text{ see Case 1.}
  \\
  &
  \costt{\forx{e_1}{\inc(e_2)}}{} 
  \prec \costt{\forx{e_1}{e_2}}{}, \text{ see Case 2.}
  \\
  &
  \cost{\forx{\inc(e_1)}{\inc(e_2)}}{}{c}
  =  \bagS{ \costcx{\inc(e_2)}{i}{ \costt{\inc(e_1)}{i} } }
          { \costcx{\inc(e_2)}{o}{ \costt{\inc(e_1)}{i} } \cdot
            \costt{\inc(e_1)}{o} }
  \\
  & 
  \prec \bagS{ \costcx{e_2}{i}{ \costt{\inc(e_1)}{i} } }
             { \costcx{e_2}{o}{ \costt{\inc(e_1)}{i} } \cdot 
               \costt{\inc(e_1)}{o} }
  = \costt{ \forx{\inc(e_1)}{e_2} }{} 
  \prec \costt{\forx{e_1}{e_2}}{}. 
  \end{align*}

  \item For $h = e_1 \cprod e_2$ we need to show that:
  \begin{align*}
  &
  \costt{\inc(e_1 \cprod e_2)}{} 
  = \costt{e_1 \cprod \inc(e_2) \uplus \inc(e_1) \cprod e_2 \uplus 
          \inc(e_1) \cprod \inc(e_2) }{}
  \\
  &
  \qqquad\qquad\quad\;\,
  = \sup( \costt{e_1 \cprod \inc(e_2)}{}, \costt{\inc(e_1) \cprod e_2}{},
          \costt{\inc(e_1) \cprod \inc(e_2)}{} )
  \prec
  \costt{e_1 \cprod e_2}{} 
  \end{align*}
  Case 1: $\costt{\inc(e_2)}{} \prec \costt{e_2}{}$
          and $e_1$ is {\em input-independent}, therefore
          $\inc(e_1) = \uZ$ (from Lemma~\ref{lem:delta-const}):
  \begin{align*}
  &
  \costt{\inc(e_1 \cprod e_2)}{} 
  = \costt{e_1 \cprod \inc(e_2)}{}
  = \bagS{ \tuple{ \costt{e_1}{i}, \costt{\inc(e_2)}{i} } }
        { \costt{e_1}{o} \cdot \costt{\inc(e_2)}{o} }
  \\
  &
  \qqqquad\qqqquad
  \prec
    \bagS{ \tuple{ \costt{e_1}{i}, \costt{e_2}{i} } }
        { \costt{e_1}{o} \cdot \costt{e_2}{o} }
  = \costt{e_1 \cprod e_2}{}
  \end{align*}

  Case 2: $\costt{\inc(e_1)}{} \prec \costt{e_1}{}$ 
          and $e_2$ is {\em input-independent}:
          Analogous to Case 1.
 
  Case 3: $\costt{\inc(e_1)}{} \prec \costt{e_1}{}$ and
          $\costt{\inc(e_2)}{} \prec \costt{e_2}{}.$ We show that each term 
          of the $\sup$ function is less than the cost of the original function: 
  \begin{align*}
  &
  \costt{e_1 \cprod \inc(e_2)}{} \prec \costt{e_1 \cprod e_2}{},
  \text{ see Case 2.}
  \\
  &
  \costt{\inc(e_1) \cprod e_2}{} \prec \costt{e_1 \cprod e_2}{},
  \text{ see Case 3.}
  \\
  &
  \costt{\inc(e_1) \cprod \inc(e_2)}{}
  = \bagS{ \tuple{ \costt{\inc(e_1)}{i}, \costt{\inc(e_2)}{i} } }
        { \costt{\inc(e_1)}{o} \cdot \costt{\inc(e_2)}{o} }
  \\
  &
  \qqqquad\;\;\;
  \prec
    \bagS{ \tuple{ \costt{e_1}{i}, \costt{\inc(e_2)}{i} } }
        { \costt{e_1}{o} \cdot \costt{\inc(e_2)}{o} }
  = \costt{e_1 \cprod \inc(e_2)}{} \prec
    \costt{e_1 \cprod e_2}{}.
  \end{align*}

  \item For $h = e_1 \uplus e_2$ we have the following cases:
  \\
  Case 1: $\costt{\inc(e_2)}{} \prec \costt{e_2}{}$
          and $e_1$ is {\em input-independent}, therefore
          $\inc(e_1) = \uZ$ (from Lemma~\ref{lem:delta-const}):
  \begin{align*}
  &
  \costt{\inc(e_1 \uplus e_2)}{} 
  = \costt{ \inc(e_2) }{}
    \prec \sup( \costt{e_1}{},\costt{e_2}{} )
  = \costt{e_1 \uplus e_2}{}. 
  \end{align*}

  Case 2: $\costt{\inc(e_1)}{} \prec \costt{e_1}{}$ 
          and $e_2$ is {\em input-independent}:
          Analogous to Case 1.
 
  Case 3: $\costt{\inc(e_1)}{} \prec \costt{e_1}{}$ and
          $\costt{\inc(e_2)}{} \prec \costt{e_2}{}:$
  \begin{align*}
  &
  \costt{\inc(e_1 \uplus e_2)}{} 
  = \costt{\inc(e_1) \uplus \inc(e_2) }{}
  = \sup( \costt{\inc(e_1)}{},\costt{\inc(e_2)}{} ) \prec
    \sup( \costt{e_1}{},\costt{e_2}{} )
  = \costt{e_1 \uplus e_2}{}. 
  \end{align*}

  \item For $h = \ominus(e)$ we have that
         $\costt{\inc(\ominus(e))}{} 
           = \costt{\ominus(\inc(e))}{}
           = \costt{\inc(e)}{}
           \prec \costt{e}{} 
           = \costt{\ominus(e)}{}. 
          $

  \item For $h = \flt(e)$ we have that
        $\costt{\inc(e)}{} \prec \costt{e}{}$, therefore:
          \begin{align*}
          &  
             \costt{\inc(\flt(e))}{} 
           = \costt{\flt(\inc(e))}{}
           = \bagS{ \costt{\inc(e)}{ii} }
                 { \costt{\inc(e)}{o} \cdot \costt{\inc(e)}{oi} }
          \\
          &
          \qqqquad\qqqquad\qquad\quad\;\;\;
           \prec 
             \bagS{ \costt{e}{ii} }
                 { \costt{e}{o} \cdot \costt{e}{oi} } 
           = \costt{\flt(e)}{}, 
          \end{align*}
  where we used the fact that $\costt{\inc(e)}{o} < \costt{e}{o}$ and
  $\costt{\inc(e)}{i} \preceq \costt{e}{i}.$
\end{itemize}
\end{proof}

\section{Shredding $\pNRC$}
\label{app:shred}

\subsection{The shredding transformation}
\label{app:shred-trans}

The full definition of the shredding transformation 
for the constructs of $\pNRC$ can be found in Figure \ref{fig:shred}.
We remark that it produces expressions that no longer
make use of the singleton combinator $\sng(e)$, 
thus their deltas do not generate any deep updates. 

In addition,
we note that only the shreddings of $\sng(e)$ and $\flt(e)$
fundamentally change the contexts, whereas the shreddings of most of the other
operators modify only the flat component of the output
(see $\shred(e_1 \cprod e_2)$, $\shred(\ominus(e))$). 
In fact, if we interpret the output context $\gc{B}$ as a tree, having the same
structure as the nested type $B$, we can see that $\gcf{\sng(e)}$ /
$\gcf{\flt(e)}$ are the only ones able to add / remove a level from the tree.

\begin{figure*}[t!]

{\small
\begin{tabular}{lll}
{$\!
\begin{aligned}
\fcf{ \sngvar{x} } &: \Bag(A^F)
\\
\fcf{ \sngvar{x} } &= \sngvar{\fc{x}}
\\
\gcf{ \sngvar{x} } &: \gc{A}
\\
\gcf{ \sngvar{x} } &= \gc{x}
\end{aligned}
$}
&
{$\!
\begin{aligned}
\fcf{\forx{e_1}{e_2}} &: \Bag(B^F)
\\
\fcf{\forx{e_1}{e_2}} &= \llet{\gc{x}}{\gc{e_1}}\ 
                         \for{\fc{x}}{\fc{e_1}}\collects
                         \fc{e_2} 
\\
\gcf{\forx{e_1}{e_2}} &: B^\gG
\\
\gcf{\forx{e_1}{e_2}} &= \llet{\gc{x}}{\gc{e_1}}\ \gc{e_2} 
\end{aligned}
$}
&
{$\!
\begin{aligned}
\fcf{ \piTC{x}{i} } &: \Bag(A_i^F) 
\\
\fcf{ \piTC{x}{i} } &= \piTC{\fc{x}}{i}
\\
\gcf{ \piTC{x}{i} } &: A_i^\gG 
\\
\gcf{ \piTC{x}{i} } &= \ggc{x}{i}
\end{aligned}
$}
\\[35pt]
\end{tabular}

\begin{tabular}{lll}
{$\!
\begin{aligned}
\fcf{ \etpl } &: \Bag(1)\qqqquad
\\
\fcf{ \etpl } &= \etpl
\\
\gcf{ \etpl } &: 1
\\
\gcf{ \etpl } &= \tuple{}
\end{aligned}
$}
&
{$\!
\begin{aligned}
\fcf{e_1 \cprod e_2} &: \Bag(A_1^F \x A_2^F)
\\
\fcf{e_1 \cprod e_2} &= \fc{e_1} \cprod \fc{e_2}
\\
\gcf{e_1 \cprod e_2} &: A_1^\gG \x A_2^\gG 
\\
\gcf{e_1 \cprod e_2} &= \tuple{\gc{e_1}, \gc{e_2}} 
\end{aligned}
$}
&
{$\!
\begin{aligned}
\fcf{ R } &: \Bag(A^F)
\\
\fcf{ R } &= \forr{r}{R}{ \sh^F_A(r) }
\\
\gcf{ R } &: A^\gG 
\\
\gcf{ R } &= \sh^\gG_A
\end{aligned}
$}
\\[35pt]
{$\!
\begin{aligned}
\fcf{ e_1 \uplus e_2 } &: \Bag(B^F)
\\
\fcf{ e_1 \uplus e_2 } &= \fc{e_1} \uplus \fc{e_2}
\\
\gcf{ e_1 \uplus e_2 } &: B^\gG
\\
\gcf{ e_1 \uplus e_2 } &= \gc{e_1} \dcup \gc{e_2} 
\end{aligned}
$}
&
{$\!
\begin{aligned}
\fcf{ \sng_\iota(e) } &: \Bag(\bL) 
\\
\fcf{ \sng_\iota(e) } &= 
	\inL_{\iota,\VGamma}(\Vrho)
\\
\gcf{ \sng_\iota(e) } &:  ( \bL \arr \Bag(\fc{B}) ) \x \gc{B}
\\
\gcf{ \sng_\iota(e) } &= 
	        \tuple{ [(\iota,\VGamma) \mapsto \fc{e}],
				    \gc{e}
	        } \quad\;\;\, 
\end{aligned}
$}
&
{$\!
\begin{aligned}
\fcf{ \uZ } &: \Bag(B^F)
\\
\fcf{ \uZ } &= \uZ
\\
\gcf{ \uZ } &: B^\gG 
\\
\gcf{ \uZ } &= \uZ_{B^\gG}
\end{aligned}
$}
\\[35pt]
{$\!
\begin{aligned}
\fcf{ \ominus( e ) } &: \Bag(\fc{B})
\\
\fcf{ \ominus( e ) } &=  \ominus( \fc{e} ) 
\\
\gcf{ \ominus( e ) } &:  \gc{B}
\\
\gcf{ \ominus( e ) } &= \gc{e} 
\end{aligned}
$}
&
{$\!
\begin{aligned}
\fcf{ \flt(e) } &: \Bag(\fc{B})
\\
\fcf{ \flt(e) } &= \forr{l}{\fc{e}}{\ggc{e}{1}(l)}  \qquad 
\\
\gcf{ \flt(e) } &: \gc{B}
\\
\gcf{ \flt(e) } &= \ggc{e}{2} 
\end{aligned}
$}
&
{$\!
\begin{aligned}
\fcf{ p(x) } &: \Bag(1)
\\
\fcf{ p(x) } &= p(x)
\\
\gcf{ p(x) } &: 1
\\
\gcf{ p(x) } &= \tuple{}
\end{aligned}
$}
\\[35pt]
\end{tabular}

\begin{tabular}{ll}
{$\!
\begin{aligned}
\fcf{\llet{X}{e_1} e_2} &: \Bag(B^F)
\qquad\quad\;
\\
\gcf{\llet{X}{e_1} e_2} &: B^\gG 
\end{aligned}
$}
&
{$\!
\begin{aligned}
\fcf{\llet{X}{e_1}\ e_2} &= 
\lletBin{\fc{X}}{\fc{e_1}}{\gc{X}}{\gc{e_2}}\ \fc{e_2}
\\
\gcf{\llet{X}{e_1}\ e_2} &= 
\lletBin{\fc{X}}{\fc{e_1}}{\gc{X}}{\gc{e_2}}\ \gc{e_2} 
\end{aligned}
$}
\end{tabular}
}

\caption{The shredding transformation, where $\sh^F_A$ and $\sh^\gG_A$
are described in Figure~\ref{fig:shk}.}
\label{fig:shred}
\end{figure*}

\begin{figure*}[t!]

\begin{subfigure}{\textwidth}

\begin{align*}
&
\shF_{\Base} : \Base \arr \Bag(\Base)
&&
\shF_{A_1 \x A_2} :	(A_1 \x A_2) \arr \Bag(A_1^F \x A_2^F)
&&
\shF_{\Bag(C)} : \Bag(C) \arr \Bag(\bL) 
\\
&
\shF_{\Base}(a) = \sngvar{a}
&&
\shF_{A_1 \x A_2}(a) 
= \forr{\tuple{a_1,a_2}}{\sngvar{a}}{}  
&&
\shF_{\Bag(C)}(v) =  \gdic_C(v)
\\
&&&
\qqquad\quad\;\; 
\fc{\sh}_{A_1}(a_1) \cprod \fc{\sh}_{A_2}(a_2)
\\[-25pt]
\end{align*}
\begin{align*}
&
\shG_{\Base} : 1
&&
\shG_{A_1 \x A_2} :	A_1^\gG \x A_2^\gG
&&
\shG_{\Bag(C)} :	( \bL \darr  \Bag(C^F) ) \x C^\gG 
\\
&
\shG_{\Base} = \tuple{}
&&
\shG_{A_1 \x A_2}
= \tuple{\sh^\gG_{A_1}, \sh^\gG_{A_2}} 
&&
\shGG{1}_{\Bag(C)} = \for{l}{\supp(\gdic^{-1}_C)} \collects
&&
\shGG{2}_{\Bag(C)} 
=	      \shG_C 
\\
&&&&& 
\qqquad\;\;
  [l \mapsto \forr{c}{\gdic^{-1}_C(l)}{\shF_C(c)}] 
\end{align*}

\caption{
$\shF_A : A \arr \Bag(\fc{A}), \shG_A : \gc{A}$
}
\label{fig:shk}

\end{subfigure}
~
\begin{subfigure}{\textwidth}

\begin{align*}
&
\nst_{\Base}[\tuple{}] : \Base \arr \Bag(\Base)
&&
\nst_{\Base}[\tuple{}](\fc{a}) = \sngvar{\fc{a}}
\\
&
\nst_{A_1{\x}A_2}[\gc{a}] : A_1^F {\x} A_2^F \arr \Bag(A_1{\x}A_2)
&&
\nst_{A_1 \x A_2}[\gc{a}](\fc{a}) = 
  \forr{\tuple{\fc{a_1},\fc{a_2}}}{\sngvar{\fc{a}}}{} 
  \nst_{A_1}[\ggc{a}{1}](\fc{a_1}) \cprod
   \nst_{A_2}[\ggc{a}{2}](\fc{a_2})
\\
&
\nst_{\Bag(C)}[\gc{a}] : \bL \arr \Bag(\Bag(C))
&&
\nst_{\Bag(C)}[\gc{a}](l) = 
	\sng( \forr{\fc{c}}{\ggc{a}{1}(l)}{ \nst_C[\ggc{a}{2}](\fc{c}) } )
\end{align*}

\caption{$\nst_A[\gc{a}]  : \fc{A} \arr \Bag(A)$}
\label{fig:nstk}

\end{subfigure}

\caption{Shredding and nesting bag values.}
\label{fig:nst_sh}

\end{figure*}

We define 
$\shF_A : A \arr \Bag(\fc{A}) \text{ and } \shG_A : \gc{A},$ for shredding bag 
values $R : \Bag(A)$, as well as $\nst_{A}[\gc{a}]: \fc{A} \arr \Bag(A)$ 
for converting them back to nested form:
\begin{align*}
\\[-7mm]
&
\fc{R} = \forr{a}{R}{\sh^F_A(a)}  
\qquad
\gc{R} = \shG_A
\qqqquad
R = \forr{\fc{a}}{\fc{R}}{\nst_A[\gc{R}](\fc{a})}, 
\\[-7mm]
\end{align*}
where $\sh^F_A, \sh^\gG_A$ and $\nst_A$
are presented in Figure~\ref{fig:nst_sh}.

When shredding a bag value $R: \Bag(A)$, 
the flat component $R^F : \Bag(A^F)$ is generated
by replacing every nested bag $v:\Bag(C)$ from $R$, with a 
label $l = \tuple{\iota_{v},\tuple{}}$. 
The association between every bag $v : \Bag(C)$, occurring nested somewhere 
inside $R$, and the label $l$ is given via $\gdic_C$ and $\gdic^{-1}_C$:
\begin{align*}
&
\gdic_C : \Bag(C) \arr \Bag(\bL)
&&
\gdic_C(v) = \sngv{l}
\qqqquad 
&&
\gdic^{-1}_C : \bL \darr \Bag(C) 
&&
\gdic^{-1}_C(l) = v.
\end{align*} 

For each label $l$ introduced by $\gdic_C$, $\shG_{\Bag(C)}$ constructs
a dictionary, mapping $l$ to the flat component of the shredded version of 
$v$.
This is done by first using the dictionary $\gdic^{-1}_C$, to obtain $v$
and applying $\shF_C$ to shred its contents.

Converting a shredded bag 
$\fc{R} : \Bag(A^F), \gc{R}: A^\gG$, back to nested form 
can be done via $\forr{x}{\fc{R}}{\nst_A[\gc{R}](x)}$, which replaces
the labels in $\fc{R}$ by their definitions from the context $\gc{R}$, 
as computed by $\nst_A[\gc{a}]$ (Figure \ref{fig:nstk}).
We note that the definitions themselves also have to be recursively 
turned to nested form, which is done in $\nst_{\Bag(C)}$.

\subsection{Example: Label dictionaries}
\label{app:shred:ex-lab-dic}

We give a couple of examples where we contrast the outcome
of label unioning dictionaries with that of applying bag addition on them
(we use $x^n$ to denote $n$ copies of $x$).
\begin{align*}
[l_1 \mapsto \set{b_1}, l_2 \mapsto \set{b_2,b_3}]
\dcup
[l_2 \mapsto \set{b_2, b_3}, l_3 \mapsto \set{b_4}]
&=
[l_1 \mapsto \set{b_1}, 
 l_2 \mapsto \set{b_2,b_3}, 
 l_3 \mapsto \set{b_4}]
\\
[l_1 \mapsto \set{b_1}, l_2 \mapsto \set{b_2,b_3}]
\uplus
[l_2 \mapsto \set{b_2, b_3}, l_3 \mapsto \set{b_4}]
&=
[l_1 \mapsto \set{b_1}, 
 l_2 \mapsto \set{b_2^2,b_3^2}, 
 l_3 \mapsto \set{b_4}]
\\[10pt]
[l_1 \mapsto \set{b_1}, l_2 \mapsto \set{b_2,b_3}]
\dcup
[l_2 \mapsto \set{b_5}, l_3 \mapsto \set{b_4}]
&=
\text{error}
\qquad\;
\\
[l_1 \mapsto \set{b_1}, l_2 \mapsto \set{b_2,b_3}]
\uplus
[l_2 \mapsto \set{b_5}, l_3 \mapsto \set{b_4}]
&=
[l_1 \mapsto \set{b_1}, 
 l_2 \mapsto \set{b_2,b_3,b_5}, 
 l_3 \mapsto \set{b_4}]
\end{align*}
As we can see from these examples, bag addition allows us to modify the
label definitions stored inside the dictionaries, which is otherwise not
possible via label unioning.

\subsection{Consistency of shredded values}
\label{app:shred-val-cons}

Given an input bag $R : \Bag(A)$, its shredding version consists of a flat
component $\fc{R} : \Bag(\fc{A})$ and a context component $\gc{R} : \gc{A}$,
which is essentially a tuple of dictionaries 
$d_k : \bL \arr \Bag(\fc{C})$ such that the definition of any label $l$ in $d_k$
corresponds to a inner bag of type $\Bag(C)$ from $R$.

In order to be able to manipulate shredded values in a consistent manner we
must guarantee that 
i) the union of label dictionaries is always well defined and that 
ii) each label occurring in the flat component of a shredded 
value has a corresponding definition in the associated context component.
More formally:

\begin{definition}
We say that shredded bags 
$\tuple{ \fc{R_i}, \gc{R_i}} : \Bag(\fc{A_i}) \x \gc{A_i}$ are consistent if
the union operation over dictionaries is well-defined between any two
compatible dictionaries in $\gc{R_i}, 1\leq i\leq n$ and if all the elements of 
$\fc{R_i}$ are $\kcons$ wrt.\ their context $\gc{R_i}$, where
$v:\fc{A}$ is $\kcons$ wrt.\ $\gc{v} : \gc{A}$, if:
\begin{itemize}
\item $A = \Base$ or
\item $A = A_1 \x A_2$, $v = \tuple{v_1, v_2}$, 
$v^\gG = \tuple{v^\gG_1, v^\gG_2}$ and
$v_1,v_2$ are $\kcons$ wrt.\ $\gc{v_1}$ and $\gc{v_2}$, respectively,
or
\item $A = \Bag(C)$, $v = l : \bL$, 
$v^\gG = \tuple{v^D, c^\gG} : (\bL \arr \Bag(C^F)) \x C^\gG$, 
there exists a definition for $l$ in $v^D$ (i.e.\ $l \in \supp(v^D)$)
and for every element $c_j$
of the definition $v^D(l) = \biguplus_j \set{c_j}$, 
$c_j$ is $\kcons$ wrt.\ $c^\gG$.
\end{itemize}

\end{definition}

Regarding the first requirement, we note that
the union of label dictionaries $d_1 \dcup d_2$ results in an error when a label 
$l$ is defined in both $d_1$ and $d_2$
(i.e. $l \in \supp(d_1) \cap \supp(d_2)$) 
but the definitions do not match.
Therefore, in order to avoid this scenario a label $l$ must have the same definitions in all dictionaries where it appears.
This is true for shredded input bags, since the shredding function introduces
a fresh label for every inner bag encountered in the process.
Furthermore, this property continues to be true after evaluating 
the shredding of query 
$h[R_i] : \Bag(B)$ :
\begin{align*}
\fcf{h}[\fc{R_i},\gc{R_i}] : \Bag(\fc{B})
\qqquad
\gcf{h}[\fc{R_i},\gc{R_i}] : \gc{B}
\end{align*}
over shredded input bags $\fc{R_i} : \Bag(\fc{A}), \gc{R_i} : \gc{A}$
because
a) the labels introduced by the query 
(corresponding to the shredding of $\sng(f)$ constructs) are  
guaranteed to be fresh
and 
b) within the queries $\fcf{h}$ and $\gcf{h}$ 
dictionaries are combined only via 
label union which doesn't modify label definitions  
(i.e.\ we never apply bag union over dictionaries).

\begin{restatable}{lemma}{shConsistProp}
\label{sh_consist_prop}
Shredding produces consistent values, i.e.\ for any input bags $R_i$, 
their shredding
$\fc{R_i} = \forr{r}{R_i}{ \shF_{A_i}(r) }, \gc{R_i}= \shG_{A_i}$ is consistent.
\end{restatable}

\begin{restatable}{lemma}{shExpConsistProp}
\label{sh_exp_consist_prop}
Shredded $\pNRC$ queries preserve consistency of shredded bags, 
i.e.\ for any $\pNRC$
query $h[R_i]$, 
the output of
$\tuple{ \fc{h}, \gc{h} }[\fc{R_i},\gc{R_i}]$
over consistent shredded bags 
$\tuple{\fc{R_i},\gc{R_i}}$,   
is also consistent.
\end{restatable}

When discussing the update of shredded bags 
$\tuple{ \fc{R_i}, \gc{R_i} }$ by pointwise bag union with
$\tuple{ \Delta \fc{R_i}, \Delta \gc{R_i} }$ we require that both shreddings
are independently consistent. 
Nonetheless, the definition of a label $l$ from
$\gc{R_i}$ will most likely differ from its definition in 
$\Delta \gc{R_i}$ since the first one contains the initial value of the bag
denoted by $l$, while the second one represents its update.
We remark that this does not create a problem wrt. label union of dictionaries
since we only union two dictionaries which are both from $\gc{R_i}$ or
$\Delta \gc{R_i}$, but we never label union a dictionary from $\gc{R_i}$
with a dictionary from $\Delta \gc{R_i}$.  

The definitions provided by $\Delta \gc{R_i}$ can be split in two categories:
i) update definitions for bags that have been initially defined in $\gc{R_i}$;
and ii) fresh definitions corresponding to new labels introduced in the delta
update. 
We require that if a label $l \in \supp(\gc{R_i})$ has an update definition in 
$\Delta \gc{R_i}$, then that label must have the same update definition
in every $\Delta \gc{R_k}, k = 1..n,$ for which $l \in \supp(\gc{R_k}).$ 
This is necessary in order to make sure that the resulting shredded value 
$\tuple{ \fc{R_i} \uplus \Delta \fc{R_i}, \gc{R_i} \uplus \Delta \gc{R_i} }$ 
is also consistent.   
For the fresh definitions we require that their labels are distinct from
any label introduced by $\gc{R_k}, k=1..n$. More formally:

\begin{definition}
\label{def:cons-upd}
We say that update 
$\tuple{ \fc{\Delta R_i}, \gc{\Delta R_i} }$
is consistent wrt. shredded bags
$\tuple{ \fc{R_i}, \gc{R_i} }$
if both $\tuple{ \fc{\Delta R_i}, \gc{\Delta R_i} }$ and
$\tuple{ \fc{R_i}, \gc{R_i} }$ are consistent and 
\begin{itemize}
\item
for every label 
$l \in \supp(\Delta \gc{R_i}) \cap \supp(\gc{R_i}) \cap \supp(\gc{R_k})$ 
then $l \in \supp(\Delta \gc{R_k}), k =1..n.$

\item
for every label $l \in \supp(\Delta \gc{R_i}) \setminus \supp(\gc{R_i})$ 
then $l \notin \supp(\gc{R_k}), k =1..n.$
\end{itemize}
\end{definition}

\begin{lemma}
\label{lem:cons-upd-delta}
Deltas of shredded $\pNRC$ queries preserve consistency of updates, 
i.e. for any $\pNRC$ query $h[R_i]$
and shredded update $\tuple{ \fc{\Delta R_i}, \gc{\Delta R_i} }$
consistent wrt. shredded input $\tuple{ \fc{R_i}, \gc{R_i} }$, 
then the output update
$\tuple{ \inc(\fc{h}), \inc(\gc{h}) }
[\fc{R_i},\gc{R_i},\Delta \fc{R_i},\Delta \gc{R_i}]$
is also consistent wrt. output 
$\tuple{ \fc{h}, \gc{h} }
[\fc{R_i},\gc{R_i}]$.
\end{lemma}
\begin{proof}
The first requirement of Definition~\ref{def:cons-upd} follows from the fact
that if 
$l \in \supp(\inc(\gc{h}_j)) \cap \supp(\gc{h}_j) \cap \supp(\gc{h}_k),$
where 
$\gc{h}_j/\gc{h}_k$ stands for the $j$'th/$k$'th dictionary in $\gc{h},$
then taking delta over $\gc{h}_k$ will also produce a definition for 
$l$ in $\inc(\gc{h}_k)$.

As the delta transformation does not add any new labels we have that: 
\begin{align*}
\supp(\gc{h}_j) 
&\subseteq \supp_h \cup \supp(\gc{R_i})
&
\supp(\inc(\gc{h}_j)) 
&\subseteq \supp_h \cup \supp(\gc{R_i}) \cup \supp(\Delta \gc{R_i}),
\end{align*}
where 
$\supp_h$ represents the labels introduced by the query $h$ itself via
singleton constructs $\sng_\iota(e)$.

For the second requirement of Definition~\ref{def:cons-upd} we note that
if $l \in \supp(\inc(\gc{h}_j)) \setminus \supp(\gc{h}_j)$,   
then $l \in \supp(\Delta \gc{R_i}) \setminus \gc{R_i}$.
Therefore, $l \notin \supp(\gc{h}_k)$ for any dictionary in $\gc{h}.$
\end{proof}

\subsection{Delta transformation for $\ilpNRC$}
\label{sec:app:delta-label}

\deltathilpNRA*
\begin{proof}
The proof follows by structural induction on $h$ and from the
semantics of $\ilpNRC$ constructs.

\begin{itemize}

\item For $h = [l \mapsto e](l') = \forr{x}{\outL_l(l')}{e}$, 
the result follows from the delta of $\forz$ and from the fact that
$\outL_l(l')$ does not depend on the input bags, therefore its delta is $\uZ$.

\item For $h = e_1 \dcup e_2, e_1,e_2 : \bL \arr \Bag(A),$ 
we need to show that for any $l \in \bL$:
\newcommand{\new}[1]{#1^{\text{new}}}
\newcommand{\old}[1]{#1^{\text{old}}}
\newcommand{\ince}[1]{#1^{\Delta}}
\begin{align*}
\mean{ ( \old{e_1} \uplus \ince{e_1} ) \dcup 
       ( \old{e_2} \uplus \ince{e_2}) }{}{l} = 
\mean{ (\old{e_1} \dcup \old{e_2}) \uplus
       (\ince{e_1} \dcup \ince{e_2}) }{}{l},
\end{align*}
where:
$\old{e_k} = e_k[\fc{R_i}, \gc{R_i}],$ and
$\ince{e_k} = \inc(e_k)[\fc{R_i},\gc{R_i}, 
                         \fc{\Delta R_i},\gc{\Delta R_i}],$
with $k = 1,2.$

We assume that update 
$\tuple{\Delta \fc{R_i}, \Delta \gc{R_i}}$ 
is consistent wrt.
input bags $\tuple{\fc{R_i}, \gc{R_i}}$
and from Lemma~\ref{lem:cons-upd-delta} we conclude that
update $\tuple{\ince{e_1},\ince{e_2}}$ is also consistent wrt.
$\tuple{\old{e_1},\old{e_2}}$.

We do a case analysis on $l$ (there are 16 possibilities):
\begin{itemize}
\item 
$l \notin \supp(\old{e_1})$,
$l \notin \supp(\ince{e_1})$,
$l \notin \supp(\old{e_2})$,
$l \notin \supp(\ince{e_2})$.
Trivial.

\item 
$l \in \supp(\old{e_1})$,
$l \in \supp(\ince{e_1})$,
$l \in \supp(\old{e_2})$,
$l \in \supp(\ince{e_2})$.
From the consistency of $\tuple{\old{e_1},\old{e_2}}$
we have that $\old{e_1}(l) = \old{e_2}(l)$.
Similarly, we get that $\ince{e_1}(l) = \ince{e_2}(l)$.
Therefore, we have that:
$(\old{e_1} \uplus \ince{e_1})(l) = 
 (\old{e_2} \uplus \ince{e_2})(l)$ and
$((\old{e_1} \uplus \ince{e_1}) \dcup (\old{e_2} \uplus \ince{e_2}))(l) = 
  (\old{e_1} \uplus \ince{e_1})(l) =
 ((\old{e_1} \dcup \old{e_2}) \uplus (\ince{e_1} \dcup \ince{e_2}))(l)$

\item Two cases lead to a contradiction with the first requirement of a consistent update value, since the label $l$ is defined in both 
$\old{e_1}$ and $\old{e_2}$, but is updated by only one of 
$\ince{e_1}/\ince{e_2}.$
\begin{itemize}
\item 
$l \in \supp(\old{e_1})$,
$l \notin \supp(\ince{e_1})$,
$l \in \supp(\old{e_2})$,
$l \in \supp(\ince{e_2})$.

\item 
$l \in \supp(\old{e_1})$,
$l \in \supp(\ince{e_1})$,
$l \in \supp(\old{e_2})$,
$l \notin \supp(\ince{e_2})$.
\end{itemize}

\item Four cases lead to a contradiction with the second requirement of a 
consistent update value since $\ince{e_1} / \ince{e_2}$ introduce a fresh
definition for a label that already appears in $\old{e_2} / \old{e_1}.$ 

\begin{itemize}
\item 
$l \notin \supp(\old{e_1})$,
$l \in \supp(\ince{e_1})$,
$l \in \supp(\old{e_2})$,
$l \notin \supp(\ince{e_2})$.

\item 
$l \notin \supp(\old{e_1})$,
$l \in \supp(\ince{e_1})$,
$l \in \supp(\old{e_2})$,
$l \in \supp(\ince{e_2})$.

\item 
$l \in \supp(\old{e_1})$,
$l \notin \supp(\ince{e_1})$,
$l \notin \supp(\old{e_2})$,
$l \in \supp(\ince{e_2})$.

\item 
$l \in \supp(\old{e_1})$,
$l \in \supp(\ince{e_1})$,
$l \notin \supp(\old{e_2})$,
$l \in \supp(\ince{e_2})$.
\end{itemize}

\item Two cases follow from the fact that $l$ only appears in 
$\old{e_1}, \old{e_2}$, or $\ince{e_1}, \ince{e_2}$,
which are consistent.
\begin{itemize}
\item 
$l \in \supp(\old{e_1})$,
$l \notin \supp(\ince{e_1})$,
$l \in \supp(\old{e_2})$,
$l \notin \supp(\ince{e_2})$.

\item 
$l \notin \supp(\old{e_1})$,
$l \in \supp(\ince{e_1})$,
$l \notin \supp(\old{e_2})$,
$l \in \supp(\ince{e_2})$.
\end{itemize}

\item The final six cases follow immediately as $l$ appears in dictionaries 
only on the left or only on the right hand side of label union. 
\begin{itemize}
\item 
$l \in \supp(\old{e_1})$,
$l \notin \supp(\ince{e_1})$,
$l \notin \supp(\old{e_2})$,
$l \notin \supp(\ince{e_2})$.

\item 
$l \notin \supp(\old{e_1})$,
$l \in \supp(\ince{e_1})$,
$l \notin \supp(\old{e_2})$,
$l \notin \supp(\ince{e_2})$.

\item 
$l \in \supp(\old{e_1})$,
$l \in \supp(\ince{e_1})$,
$l \notin \supp(\old{e_2})$,
$l \notin \supp(\ince{e_2})$.

\item 
$l \notin \supp(\old{e_1})$,
$l \notin \supp(\ince{e_1})$,
$l \in \supp(\old{e_2})$,
$l \notin \supp(\ince{e_2})$.

\item 
$l \notin \supp(\old{e_1})$,
$l \notin \supp(\ince{e_1})$,
$l \notin \supp(\old{e_2})$,
$l \in \supp(\ince{e_2})$.

\item 
$l \notin \supp(\old{e_1})$,
$l \notin \supp(\ince{e_1})$,
$l \in \supp(\old{e_2})$,
$l \in \supp(\ince{e_2})$.
\end{itemize}

\end{itemize}

\end{itemize}

For the second part, relating the cost and degree of the delta to the cost and
degree of the original query, the proofs are analogous to the cases from
Theorem~\ref{prop:delta-dec-deg} and
Theorem~\ref{prop:ipnra-cost}, when $h = \forx{e_1}{e_2}$ and $e_1$ is
input-independent, and $h = e_1 \uplus e_2$, respectively.

\end{proof}

\subsection{Correctness}
\label{sec:app:shred-corr}

\nstShProp*
\begin{proof}
We do a case by case analysis on the type being shredded:
\begin{itemize}
\item $A=\Base$: 
$
\forr{\fc{a}}{\shF_{\Base}(a)}{\nst_{\Base}[ \tuple{} ](\fc{a})}  
= \forr{\fc{a}}{\sngvar{a}}{\sngvar{\fc{a}}}   
= \sngvar{a}
$

\item $A=A_1 \x A_2$
\begin{align*}
&
\forr{\fc{a}}{\shF_{A_1 \x A_2}(a)}
     {\nst_{A_1 \x A_2}[ \shG_{A_1 \x A_2} ](\fc{a})} =
\\
&
= 
\forr{\fc{a}}
     {(\forr{\tuple{a_1,a_2}}{\sngvar{a}}
            { \fc{\sh}_{A_1}(a_1) \cprod \fc{\sh}_{A_2}(a_2) }) }
     {}
\\
&
\quad
     {(\forr{\tuple{\fc{a_1},\fc{a_2}}}{\sngvar{\fc{a}}}
            { \nst_{A_1}[\gc{a_1}](\fc{a_1}) \cprod
              \nst_{A_2}[\gc{a_2}](\fc{a_2}) }) }     
\\
&
= 
\forr{\tuple{a_1,a_2}}{\sngvar{a}}{}
\forr{\tuple{\fc{a_1},\fc{a_2}}}
     { \fc{\sh}_{A_1}(a_1) \cprod \fc{\sh}_{A_2}(a_2) }
     { \nst_{A_1}[\gc{a_1}](\fc{a_1}) \cprod
       \nst_{A_2}[\gc{a_2}](\fc{a_2}) }
\\
&
= 
\forr{\tuple{a_1,a_2}}{\sngvar{a}}{}
( \forr{\fc{a_1}}
     { \fc{\sh}_{A_1}(a_1) }
     { \nst_{A_1}[\gc{a_1}](\fc{a_1}) } )
\cprod
( \forr{\fc{a_2}}
     { \fc{\sh}_{A_2}(a_2) }
     { \nst_{A_2}[\gc{a_2}](\fc{a_2}) } )
\\
&
= 
\forr{\tuple{a_1,a_2}}{\sngvar{a}}{}
(\sngvar{a_1} \cprod \sngvar{a_2})
= \sngvar{a}
\end{align*}

\item $A=\Bag(C)$
\begin{align*}
&
\forr{l}{ \shF_{\Bag(C)}(a) }
     { \nst_{\Bag(C)}[ \shGG{1}_{\Bag(C)},\shGG{2}_{\Bag(C)} ](l) } =
\\
&
= \forr{l}
       { \gdic_C(a) }
       { \sng( \forr{\fc{c}}{\shGG{1}_{\Bag(C)}(l)}
                    { \nst_C[\shGG{2}_{\Bag(C)}](\fc{c}) } ) }
\\
&
= \forr{l}
       { \gdic_C(a) }
       { \sng( \forr{\fc{c}}
                    { ( \forr{c}{\gdic^{-1}_C(l)}{\shF_C(c)} ) }
                    { \nst_C[\shG_C](\fc{c}) } ) }
\\
&
= \forr{l}
       { \gdic_C(a) }
       { \sng( \forr{c}{\gdic^{-1}_C(l)}
                    { \forr{\fc{c}}{\shF_C(c)} 
                           { \nst_C[\shG_C](\fc{c}) }} ) }
\\
&
= \forr{l}
       { \gdic_C(a) }
       { \sng( \forr{c}{\gdic^{-1}_C(l)}
                    { \sngvar{c} } ) }
= \forr{l}
       { \gdic_C(a) }
       { \sng( \gdic^{-1}_C(l) ) }
= \sngvar{a}
\end{align*}
\end{itemize} 
\end{proof}

\nstNatTransfProp*
\begin{proof}
The proof consists of a case by case analysis on the structure of $h$.
We detail for $h \in \set{\sng(e), \flt(e)},$ as the rest of the cases follow in 
a similar fashion.
\begin{itemize}
\item $h = \sng(e)$
\begin{align*}
&
\llet{ R }{ \forr{ \fc{r} }{ \fc{R} }{ \nst[\gc{R}]( \fc{r} ) } }\ \sng(e)
=
\sng(\llet{ R }{ \forr{ \fc{r} }{ \fc{R} }{ \nst[\gc{R}]( \fc{r} ) } }\ e)
=
\\
&
=
\sng( \forr{ \fc{x} }{ \fc{e} }{ \nst_B[\gc{e}]( \fc{x} ) } )
\\
&
=
\sng( \forr{ \fc{x} }
           {( \forr{l}{ \inL_{\iota,\fc{A}}(\fc{a}) }
                   { [(\iota,\fc{A}) \mapsto \fc{e} ](l) } )}
           { \nst_B[\gc{e}]( \fc{x} ) } )
\\
&
=
\forr{l}{ \inL_{\iota,\fc{A}}(\fc{a}) }{
\sng( \forr{ \fc{x} }
           { [(\iota,\fc{A}) \mapsto \fc{e} ](l) }
           { \nst_B[\gc{e}]( \fc{x} ) } )
}
\\
&
=
\forr{l}{ \inL_{\iota,\fc{A}}(\fc{a}) }{
\nst_{\Bag(B)}[ [(\iota,\fc{A}) \mapsto \fc{e} ], \gc{e} ](l)
}
=
\forr{l}{ \fcf{\sng(e)} }{
\nst_{\Bag(B)}[ \gcf{\sng(e)} ](l)
}
\end{align*}

\item $h = \flt(e)$
\begin{align*}
&
\llet{ R }{ \forr{ \fc{r} }{ \fc{R} }{ \nst[\gc{R}]( \fc{r} ) } }\ \flt(e) 
=
\flt(\llet{ R }{ \forr{ \fc{r} }{ \fc{R} }{ \nst[\gc{R}]( \fc{r} ) } }\ e)
=
\\
&
= 
\flt( \forr{ l }{ \fc{e} }{ \nst_{\Bag(B)}[\gc{e}]( l ) } )
\\
&
= 
\flt( \forr{ l }{ \fc{e} }
           { \sng( \forr{\fc{x}}{\ggc{e}{1}(l)}
                        { \nst_B[\ggc{e}{2}](\fc{x}) } )
           } )
\\
&
= 
\forr{ l }{ \fc{e} }
     {( \forr{\fc{x}}{\ggc{e}{1}(l)}
             { \nst_B[\ggc{e}{2}](\fc{x}) } )}
= 
\forr{\fc{x}}
     {( \forr{ l }{ \fc{e} }{\ggc{e}{1}(l)} )}
     { \nst_B[\ggc{e}{2}](\fc{x}) }
\\
&
= 
\forr{\fc{x}}
     { \fcf{\flt(e)} }
     { \nst_B[\gcf{\flt(e)}](\fc{x}) }
\end{align*}

\end{itemize}
\end{proof}

\subsection{Complexity of Shredding}
\label{app:complex}

In this section we show that 
shredding nested bags can be done in $\TCz$.
By $\NCz$ we refer to the class of languages recognizable by LOGSPACE-uniform
families of circuits of polynomial size and constant depth using and- and or-gates
of bounded fan-in. 
The related complexity class \ACz\ differs from \NCz\
by allowing gates to have unbounded fan-in, while \TCz\ extends \ACz\ by
further permitting so-called majority-gates, that compute ``true'' iff more than
half of their inputs are true.
For details on circuit complexity and the notion of uniformity we refer to
\cite{Joh90}.

We recall that the standard way of representing flat relations when processing
them via circuits is the unary representation, i.e.\ as a collection of bits,
one for each possible tuple that can be constructed from the active domain and the schema, in some canonical ordering, where a bit
being turned on or off signals whether the corresponding tuple is in the
relation or not.
In such a representation (denote by $\FSet$ below), if the active
domain has size $\sigma$, then the number of bits required for
encoding a relation whose schema has $n_f$ fields is $\sigma ^ {n_f}$.
For example, for a relation with a single field, we need $\sigma$ bits to
encode which values from the active domain are present or not. 
We also assume a total order among the elements of the active domain, and that
the bits of $\FSet$ are in lexicographical order of the tuples they represent.

In the case of bags, whose elements have an associated multiplicity, we work
with circuits that compute the multiplicity of tuples modulo $2^k$, for some fixed
$k$.
Thus, for every possible tuple in a bag we use $k$ bits instead of a
single one, in order to encode the multiplicity of that tuple as a binary
number.
In the following we use $\FBag$ to refer to this representation of bags.

For nested values the $\FSet$ representation discussed above is no longer
feasible since it suffers from an exponential blowup with every nesting level.
This becomes apparent when we consider that representing in unary an inner bag
with $n_t$ possible tuples requires $2^{n_t}$ bits. 
Consequently, for a nested value we use an alternate representation $\FStr$, as
a relation $S(p,s)$ which encodes the string representation (over a non-fixed alphabet that includes the active domain, the possible atomic field values) of the value by
mapping each position $p$ in the string to its corresponding symbol $s$.

\begin{example}
The string representation $\set{\tuple{a,\set{b,c}},\tuple{d,\set{e,f}}}$,
of a nested value of type $\Bag(\Base \x \Bag(\Base))$,
is encoded by relation $S(p,s)$ as follows (we show tuples as columns to save space):
\begin{center}
{\setlength{\tabcolsep}{3pt}
\begin{tabular}{l}
\begin{tabular}{ >{$}l<{$} |
>{$}c<{$} >{$}c<{$} >{$}c<{$} >{$}c<{$} >{$}c<{$} >{$}c<{$} >{$}c<{$} >{$}c<{$}
>{$}c<{$} >{$}c<{$} >{$}c<{$} >{$}c<{$} >{$}c<{$} >{$}c<{$} >{$}c<{$} >{$}c<{$}
>{$}c<{$} >{$}c<{$} >{$}c<{$} >{$}c<{$} >{$}c<{$}
}
	\hline			
p & 
1 &  2 &  3 &  4 &  5 &  6 &  7 &  8 &  9 & 10 & 11 
  & 12 & 13 & 14 & 15 & 16 & 17 & 18 & 19 & 20 & 21 
 \\ 
s & 
\{ & \la & a & , & \{ & b & , & c & \} & \ra & ,
   & \la & d & , & \{ & e & , & f & \} & \ra & \} 
 \\ \hline  
\end{tabular}
\\[8pt]
\end{tabular}
}
\end{center}
\end{example}

For a particular input size $n$, the active domain of $S$ consists of
$\sigma_{ext}$ symbols including the active domain of the database, delimiting
symbols $``\la",``\ra",``,",``\{",``\}"$, as well as an additional symbol for
each possible position in the string (i.e.\ $\sigma_{ext} = \sigma + 5 + n$).
We remark that the $\FSet$ representation of $S$ requires $\sigma_{ext}^2$ bits
and thus remains polynomial in the size of the input.

This representation may seem to require justification, since strings over an unbounded
alphabet may seem undesirable. We note that the representation is fair in the sense that
it does not require a costly (exponential) blow-up from the practical string representation
that could be used to store the data on a real storage device such as a disk; we use
a relational representation of the string and the canonical representation of relations
as bit-sequences that is standard in circuit complexity. The one way we could have been even
more faithful would have been to start with exactly the bit-string
representation by which an (XML, JSON, or other) nested dataset would be stored on a disk.
This -- breaking up the active domain values into bit sequences -- is however
avoided for the same reason it is avoided in the case of the study of the circuit complexity of queries on flat relations -- reconstructing the active domain from the bit string dominates the cost of query evaluation.

We can now give our main results of this section.

\begin{theorem}
\label{lem:shred-TCz}
Shredding a nested bag from $\FStr$ representation to a flat
bag ($\FBag$) representation is in \TCz.
\end{theorem}
\begin{proof}
To obtain our result
we take advantage of the fact that first-order logic with majority-quantifiers
(FOM) is in \TCz~\cite{BIS1990}, and express the shredding of a nested value as
a set of FOM queries over the $S(p,s)$ relation.

We start by defining a family of queries $\Value_A(i,j)$ or testing whether a
closed interval $(i,j)$ from the input contains a value of a particular type $A$:
\begin{align*}
\Value_{\Base}(i,j) \coloneq\ 
&
S_{\Base}(i) \wedge i = j
\\
\Value_{A_1 \x A_2}(i,j) \coloneq\ 
&
S_{\la}(i) \wedge S_{\ra}(j) \wedge
\exists k. \Pair_{A_1,A_2}(i+1,k,j-1)
\\
\Pair_{A_1,A_2}(i,k,j) \coloneq\ 
&
S_{,}(k) \wedge \Value_{A_1}(i,k-1) \wedge \Value_{A_2}(k+1,j)
\\
\Value_{\Bag(C)}(i,j) \coloneq\ 
&
S_{\{}(i) \wedge S_{\}}(j) \wedge 
( j = i+1 \vee \Seq_C(i+1,j-1) )
\\
\Seq_{C}(i,j) \coloneq\ 
&
\exists k,l. \Elem_{C}(i,k,l,j) \wedge
	\forall k,l. 
	\Elem_{C}(i,k,l,j) \implyA
	( \Pred_C(i,k) \wedge \Succ_C(l,j) )
\\
\Elem_{C}(i,k,l,j) \coloneq\ 
&
(i \leq k \wedge l \leq j \wedge \Value_C(k,l))
\\
\Pred_C(i,k) \coloneq\ 
&
i = k \vee 
(S_{,}(k-1) \wedge \exists k'. i \leq k' \wedge \Value_C(k',k-2))
\\
\Succ_C(l,j) \coloneq\ 
&
l = j \vee 
(S_{,}(l+1) \wedge \exists l'. l' \leq j \wedge \Value_C(l+2,l'))
\end{align*}
where $S_{\Base}(i)$ is true 
iff we have a $\Base$ symbol at position $i$ in the input string 
(and analogously for 
$S_{\{}(i), S_{\}}(i), S_{\la}(i), S_{\ra}(i)$ and $S_{,}(i)$).
When determining if an interval $(i,j)$ contains a bag value of type $\Bag(C)$
we test if it is either empty, i.e.\ $j = i+1$ or if it encloses a sequence of
$C$ elements (using $\Seq_C$), i.e.\ it has at least one $C$ element and each
element is preceded by another $C$ element or is the first in the sequence, and
succeeded by another $C$ element or is the last in the sequence.
We use auxiliary queries: $\Elem_C(i,k,l,j)$, which returns true iff the
interval $(i,j)$ contains a value of type $C$ between indices $k$ and $l$, and 
$\Succ_C(l,j)$ / $\Pred_C(i,k)$ which returns true iff the intervals
$(l,j)$ / $(i,k)$ are either empty or they begin, respectively end, with a value
of type $C$.

For shredding the value contained in an interval $(i,j)$ of the input we define
the following family of queries $\fc{\Shred}_{A}(i,j,p,s)$:
\begin{align*}
\fc{\Shred}_{\Base}(i,j,p,s) \coloneq\ 
& 
i \leq p \wedge p \leq j \wedge S(p,s) 
\\
\fc{\Shred}_{A_1 \x A_2}(i,j,p,s) \coloneq\ 
& 
\exists k. \Pair_{A_1,A_2}(i+1,k,j-1) \wedge (
\fc{\Shred}_{A_1}(i+1,k-1,p,s) 
\ \vee\ \fc{\Shred}_{A_1}(k+1,j-1,p,s) 
)
\\
\fc{\Shred}_{\Bag(C)}(i,j,p,s) \coloneq\ 
& 
p = i \wedge s = i,
\end{align*}
where the shredding of bag values results in their replacement with a unique
identifier, i.e.\ the index of their first symbol, that acts as a label.
Additionally, the definitions of these labels, i.e.\ the shredded versions of
the bags they replace are computed via:
\begin{align*}
\Dict_{C}(p,s) \coloneq\ 
&
\exists i,j,k,l. 
\Value_{\Bag(C)}(i,j) \wedge 
\Elem_{C}(i+1,k,l,j-1) \wedge (
(p = k-1 \wedge s = i)
\vee
\fc{\Shred}_C(k,l,p,s)
)
,
\end{align*}
where we prepend to each shredded element in the output the label of the bag to
which it belongs (we can do that by reusing the index of the preceding $\{$ or
comma present in the original input).
We build a corresponding relation $\Dict_C$ for every bag type $\Bag(C)$
occurring in the input.
These relations encode a flat representation of the input, 
as bags of type $\Bag(\bL \x \fc{C})$, where each tuple uses a fixed number
of of symbols, therefore we no longer make use of delimiting symbols.

For our example input, 
we only have two bag types, $\Bag(\Base \x \Bag(\Base))$ and $\Bag(\Base)$, 
and their corresponding relations are:
\newline
{\setlength{\tabcolsep}{4pt}
\begin{tabular}{l c l}
\\[-7pt]
$\Dict_{\Base\,\x \Bag(\Base)}(p,s) \coloneq$
&\qqquad&
$\Dict_{\Base}(p,s) \coloneq$
\\[2pt]
\begin{tabular}{ >{$}l<{$} |
>{$}c<{$} >{$}c<{$} >{$}c<{$} >{$}c<{$} >{$}c<{$} >{$}c<{$}
}
	\hline			
p & 1 & 3 &  5 & 11 & 13 & 15 
 \\ 
s & 1 & a &  5 &  1 &  d & 15  
 \\ \hline
\end{tabular}
&\qqquad\qqquad&
\begin{tabular}{ >{$}l<{$} |
>{$}c<{$} >{$}c<{$} >{$}c<{$} >{$}c<{$} >{$}c<{$} >{$}c<{$} >{$}c<{$} >{$}c<{$}
}
	\hline			
p & 
5 & 6 & 7 & 8 & 15 & 16 & 17 & 18 
 \\ 
s & 
5 & b & 5 & c & 15 & e & 15 & f 
 \\ \hline
\end{tabular}
\\[14pt]
\end{tabular}
}.
\newline
The flat values that they encode are 
$\set{\tuple{1,a,5},\tuple{1,d,15}} : \Bag(\bL \x Base \x \bL)$ and
$\set{\tuple{5,b},\tuple{5,c},\tuple{15,e},\tuple{15,f}}: \Bag(\bL \x \Base)$.

However, the $\Dict_C$ relations cannot be immediately used to produce the
sequence of tuples that they encode since the indices $p$ associated with their
symbols are non-consecutive.
To address this issue we define:
\begin{align*}
\ToSeq[X](p',s) \coloneq\ 
&
\exists p. X(p,s) \wedge p' = \#u(\exists w. X(u,w) \wedge u \leq p ),
\end{align*} 
which maps each index $p$ in relation $X(p,s)$ to an index $p'$ corresponding to
its position relative to the other indices in $X$. 
To do so we used predicate $p' = \#u\Phi(u)$ to count the number of positions
$u$ for which $\Phi(u)$ holds, since it is expressible in FOM~\cite{BIS1990}.

Finally, we determine the shredded version of an input value $x: \Bag(B)$, based
on its $\FStr$ representation $S(p,s)$, as 
$\fc{S}(p,s) \coloneq \ToSeq[\Dict_B(p,s) \wedge s \neq 1]$ where we
filter out from $\Dict_B(p,s)$ those symbols denoting that a tuple
belongs to the top level bag, identified by label $1$.
The shredding context is
defined by a collection of relations $\gc{S} \coloneq \gc{\Shred}_B$, where:
\begin{align*}
\gc{\Shred}_{\Base} \coloneq\ 
&
\uZ
&&&
\gc{\Shred}_{A_1 \x A_2} \coloneq\ 
&
\tuple{ \gc{\Shred}_{A_1} , \gc{\Shred}_{A_2} }
&&&
\gc{\Shred}_{\Bag(C)} \coloneq\ 
&
\tuple{ \ToSeq[\Dict_C] , \gc{\Shred}_{C} }
\end{align*}

The last step that remains is to convert the resulting flat bags from the
current representation (as $X(p,s)$ relations in $\FSet$ form) to the $\FBag$
representation.
We recall that each such relation encodes a sequence of tuples such that each
consecutive group of $n_f$ symbols (according to their positions $p$) stands for
a particular tuple in the bag, where $n_f$ is the number of fields in the tuple.
Additionally, since the bits in the $\FSet$ representation are lexicographically
ordered it follows that each consecutive group of $\sigma_{ext}$ bits contains
the unary representation of the symbol located at that position. 
Therefore, we can find out how many copies of a particular tuple $t$
are in the bag by counting (modulo $2^k$) for how many groups of 
$n_f \cdot \sigma_{ext}$ bits we have unary representations of symbols that
match the symbols in $t$.
By performing this counting for all possible tuples $t$ in the output bag we
obtain the $\FBag$ representation of $X(p,s)$.
We note that both testing whether particular bits are set and counting modulo
$k$ are in \TCz. 

Since $\fc{S}(p,s)$ and $\gc{S}$ can be defined via FOM queries,
and since their conversion from $X(p,s)$ relations in $\FSet$ form to the $\FBag$
representation uses a \TCz\ circuit, this concludes our proof that
shredding nested values from $\FStr$ to $\FBag$ representation can be done in
\TCz.
\end{proof}

\end{document}